\documentclass[a4paper,UKenglish,numberwithinsect,cleveref, autoref, thm-restate]{lipics-v2021}

\usepackage{amsfonts,amsthm,amsmath}
\usepackage{amstext}
\usepackage{cite}
\usepackage[justification=centering]{caption}
\usepackage{subcaption}
\usepackage[ruled, linesnumbered]{algorithm2e}
\usepackage{graphics}
\usepackage{todonotes}
\graphicspath{figures/}

\usepackage{algorithmic}
\newcommand{\OO}{\mathcal{O}}

\newcommand{\NP}{\textsf{NP}}

\newcommand{\FPT}{\textsf{FPT}}
\newcommand{\XP}{\textsf{XP}}
\newcommand{\WO}{\textsf{W[1]}\xspace}
\newcommand{\WOH}{\textsf{W[1]}-hard\xspace}

\newcommand{\NPH}{\textsf{NP}-hard\xspace}

\newcommand{\vc}{\mathsf{VC}}
\newcommand{\vn}{\mathsf{vc}}
\newcommand{\vi}{v_{\mathsf{init}}}
\newcommand{\gr}{\mathsf{Graph}}
\newcommand{\ind}{\mathsf{EQ}}

\newcommand{\indd}{\mathsf{IND}}
\newcommand{\W}{\mathsf{NeiSub}}
\newcommand{\U}{\mathsf{NeiSubsets}}
\newcommand{\ds}{\displaystyle}

\newcommand{\cg}{{\sc CGE}\xspace}
\newcommand{\cgF}{{\sc Collective Graph Exploration}\xspace}

\hideLIPIcs  

\Crefname{paragraph}{Section}{Sections}
\Crefname{observation}{Observation}{Observations}

\title{Collective Graph Exploration Parameterized by Vertex Cover} 

\author{Siddharth Gupta}{University of Warwick, UK}{siddharth.gupta.1@warwick.ac.uk}{https://orcid.org/0000-0003-4671-9822}{Supported by Engineering and Physical Sciences Research Council (EPSRC) grant EP/V007793/1.}

\author{Guy Sa'ar}{Ben Gurion University of the Negev, Israel}{saag@post.bgu.ac.il}{}{Supported in part by the Israeli Smart Transportation Research Center and by the Lynne and William Frankel Center for Computing Science at Ben-Gurion University.}

\author{Meirav Zehavi}{Ben Gurion University of the Negev, Israel}{meiravze@bgu.ac.il}{https://orcid.org/0000-0002-3636-5322}{Supported by the European Research Council (ERC) grant titled PARAPATH.}

\authorrunning{S. Gupta, G. Sa'ar, and M. Zehavi} 

\Copyright{Siddharth Gupta, Guy Sa'ar, and Meirav Zehavi} 



\relatedversion{} 

\nolinenumbers 

\EventEditors{Neeldhara Misra and Magnus Wahlstr\"{o}m}
\EventNoEds{2}
\EventLongTitle{18th International Symposium on Parameterized and Exact Computation (IPEC 2023)}
\EventShortTitle{IPEC 2023}
\EventAcronym{IPEC}
\EventYear{2023}
\EventDate{September 6--8, 2023}
\EventLocation{Amsterdam, The Netherlands}
\EventLogo{}
\SeriesVolume{285}
\ArticleNo{6}

\begin{document}

\maketitle

\begin{abstract}

We initiate the study of the parameterized complexity of the \cgF (\cg) problem. In \cg, the input consists of an undirected connected graph $G$ and a collection of $k$ robots, initially placed at the same vertex $r$ of $G$, and each one of them has an energy budget of $B$. The objective is to decide whether $G$ can be \emph{explored} by the $k$ robots in $B$ time steps, i.e., there exist $k$ closed walks in $G$, one corresponding to each robot, such that every edge is covered by at least one walk, every walk starts and ends at the vertex $r$, and the maximum length of any walk is at most $B$. Unfortunately, this problem is \NPH even on trees [Fraigniaud {\em et~al.}, 2006]. Further, we prove that the problem remains \WOH parameterized by $k$ even for trees of treedepth $3$. Due to the \textsf{para-NP}-hardness of the problem parameterized by treedepth, and motivated by real-world scenarios, we study the parameterized complexity of the problem parameterized by the vertex cover number ($\vn$) of the graph, and prove that the problem is fixed-parameter tractable (\textsf{FPT}) parameterized by $\vn$. Additionally, we study the optimization version of \cg, where we want to optimize $B$, and design an approximation algorithm with an additive approximation factor of $O(\vn)$.

\end{abstract}

\newpage


\section{Introduction}\label{sec:intro}
\cgF (\cg) is a well-studied problem in computer science and robotics, with various real-world applications such as network management and fault reporting, pickup and delivery services, searching a network, and so on. The problem is formulated as follows: given a set of robots (or agents) that are initially located at a vertex of an undirected graph, the objective is to explore the graph as quickly as possible and return to the initial vertex. A graph is \emph{explored} if each of its edges is visited by at least one robot. In each time step, every robot may move along an edge that is incident to the vertex it is placed at. The total time taken by a robot is the number of edges it traverses. The exploration time is the maximum time taken by any robot. In many real-world scenarios, the robots have limited energy resources, which motivates the minimization of the exploration time~\cite{DBLP:journals/mst/0001DK18}.

The \cg problem can be studied in two settings: \emph{offline} and \emph{online}. In the offline setting, the graph is known to the robots beforehand, while in the online setting, the graph is unknown and revealed incrementally as the robots explore it. While \cg has received considerable attention in the online setting, much less is known in the offline setting (\cref{sec:relatedWorks}). Furthermore, most of the existing results in the offline setting are restricted to trees. Therefore, in this paper, we investigate the \cg problem in the offline setting for general graphs, and present some approximation and parameterized algorithms with respect to the vertex cover number of the graph.

\subsection{Related Works}\label{sec:relatedWorks}
As previously mentioned, the \cg problem is extensively studied in the online setting, where the input graph is unknown. As we study the problem in the offline setting in this paper, we only give a brief overview of the results in the online setting, followed by the results in the offline setting.

Recall that, in the online setting, the graph is unknown to the robots and the edges are revealed to a robot once the robot reaches a vertex incident to the edge. The usual approach to analyze any online algorithm is to compute its \emph{competitive ratio}, which is the worst-case ratio between the cost of the online and the optimal offline
algorithm. Therefore, the first algorithms for \cg focused on the competitive ratios of the algorithms. In~\cite{DBLP:journals/networks/FraigniaudGKP06}, an algorithm for \cg for trees with competitive ratio $O(\frac{k}{\log k})$ was given. Later in~\cite{DBLP:journals/jco/HigashikawaKLT14}, it was shown that this competitive ratio is tight. Another line of work studied the competitive ratio as a function of the vertices and the depth of the input tree~\cite{DBLP:journals/trob/BrassCGX11,DBLP:journals/tcs/DisserMNSS20,DBLP:journals/jco/HigashikawaKLT14,DBLP:conf/mfcs/DyniaKHS06,DBLP:journals/iandc/DereniowskiDKPU15,DBLP:conf/sirocco/OrtolfS14}.
We refer the interested readers to a recent paper by Cosson {\em et~al.}~\cite{DBLP:journals/corr/abs-2301-13307}  and the references within for an in-depth discussion about the results in the online setting. 
 
We now discuss the results in the offline setting. In~\cite{DBLP:journals/dam/AverbakhB96}, it was shown that the \cg problem for edge-weighted trees is \NP-hard even for two robots. In~\cite{DBLP:journals/dam/AverbakhB97,DBLP:journals/dam/NagamochiO04}, an $(2-2/(k+1))$-approximation was given for the optimization version of \cg for edge-weighted trees where we want to optimize $B$. In~\cite{DBLP:journals/networks/FraigniaudGKP06}, the \NP-hardness was shown for \cg for unweighted trees as well. In~\cite{DBLP:conf/arcs/DyniaKS06}, a $2$-approximation was given for the optimization version of \cg for unweighted trees where we want to optimize $B$. In the same paper, it was shown that the optimization version of the problem for unweighted trees is \XP\ parameterized by the number of robots.

\subsection{Our Contribution and Methods}
In this paper, we initiate the study of the \cg problem for general unweighted graphs in the offline setting and obtain the following three results. We first prove that \cg is \FPT\ parameterized by $\vn$, where $\vn$ is the vertex cover number of the input graph. Specifically, we prove the following theorem.

\begin{restatable}[]{theorem}{fptTheorem}\label{th:fptVc}
	\cg is in FPT parameterized by $\vn(G)$, where $G$ is the input graph.
\end{restatable}

We then study the optimization version of \cg where we want to optimize $B$ and design an approximation algorithm with an additive approximation factor of $O(\vn)$. Specifically, we prove the following theorem.

\begin{restatable}[]{theorem}{approxTheorem}\label{the:approx}
	There exists an approximation algorithm for \cg that runs in time $\OO((|V(G)|+|E(G)|)\cdot k)$, and returns a solution with an additive approximation of $8\cdot\vn(G)$, where $G$ is the input graph and $k$ is the number of robots. 
\end{restatable}

Finally, we show a border of (in-)tractability by proving that \cg is \WOH parametrized by $k$, even for trees of treedepth $3$. Specifically, we prove the following theorem.

\begin{restatable}[]{theorem}{hardTheorem}\label{th:w1hard}
	\cg is \WOH with respect to $k$ even on trees whose treedepth is bounded by $3$.	
\end{restatable}

We first give an equivalent formulation of \cg based on Eulerian cycles (see \cref{obs:equivsol}). We obtain the \FPT\ result by using Integer Linear Programming (ILP).
By exploiting the properties of vertex cover and the conditions given by our formulation, we show that a potential solution can be encoded by a set of variables whose size is bounded by a function of vertex cover. 
%
%
%

To design the approximation algorithm, we give a greedy algorithm that satisfies the conditions given by our formulation. Again, by exploiting the properties of vertex cover, we show that we can satisfy the conditions of our formulation by making optimal decisions at the independent set vertices and using approximation only at the vertex cover vertices.
%
%

To prove the W-hardness, we give a reduction from a variant of {\sc Bin Packing}, called {\sc Exact Bin Packing} (defined in \cref{sec:prelims}). We first prove that {\sc Exact Bin Packing} is \WOH  even when the input is given in unary. We then give a reduction from this problem to \cg to obtain our result.

\subsection{Choice of Parameter}
As mentioned in the previous section, we proved that \cg is \WOH parameterized by $k$ even on trees of treedepth $3$. This implies that we cannot get an \FPT\ algorithm parameterized by treedepth and $k$ even on trees, unless $\FPT =\WO$. Thus, we study the problem parameterized by the vertex cover number of the input graph, a slightly weaker parameter than the treedepth. 

Our choice of parameter is also inspired by several practical applications. For instance, consider a delivery network of a large company. The company has a few major distributors that receive the products from the company and can exchange them among themselves. There are also many minor distributors that obtain the products only from the major ones, as this is more cost-effective. The company employs $k$ delivery persons who are responsible for delivering the products to all the distributors. The delivery persons have to start and end their routes at the company location. Since each delivery person has a maximum working time limit, the company wants to minimize the maximum delivery time among them. This problem can be modeled as an instance of \cg by constructing a graph $G$ that has a vertex for the company and for each distributor and has an edge between every pair of vertices that correspond to locations that can be reached by a delivery person. The $k$ robots represent the $k$ delivery persons and are placed at the vertex corresponding to the company. Clearly, $G$ has a small vertex cover, as the number of major distributors is much smaller than the total number of distributors.

For another real-world example where the vertex cover is small, suppose we want to cover all the streets of the city as fast as possible using $k$ agents that start and end at a specific street. The city has a few long streets and many short streets that connect to them. This situation is common in many urban areas. We can represent this problem as an instance of \cg by creating a graph $G$ that has a vertex for each street and an edge between two vertices if the corresponding streets are adjacent. The $k$ robots correspond to the $k$ agents. Clearly, $G$ has a small vertex cover, as the number of long streets is much smaller than the total number of streets.

\section{Preliminaries}\label{sec:prelims}
For $k\in \mathbb{N}$, let $[k]$ denote the set $\{1,2,\ldots, k\}$. For a multigraph $G$, we denote the set of vertices of $G$ and the multiset of edges of $G$ by $V(G)$ and $E(G)$, respectively. For $u\in V(G)$, the {\em set of neighbors} of $u$ in $G$ is $\mathsf{N}_G(u)=\{v\in V~|~\{u,v\}\in E(G)\}$. When $G$ is clear from the context, we refer to $\mathsf{N}_G(u)$ as $\mathsf{N}(u)$. The {\em multiset of neighbors} of $u$ in $G$ is the multiset $\widehat{\mathsf{N}}_G(u)=\{v\in V~|~\{u,v\}\in E(G)\}$ (with repetition). When $G$ is clear from the context, we refer to $\widehat{\mathsf{N}}_G(u)$ as $\widehat{\mathsf{N}}(u)$. The {\em degree} of $u$ in $G$ is $|\widehat{\mathsf{N}}_G(u)|$ (including repetitions). Let $\widehat{E}$ be a multiset with elements from $E(G)$. Let $\mathsf{Graph}(\widehat{E})$ denote the multigraph $(V',\widehat{E})$, where $V'=\{u~|~\{u,v\}\in \widehat{E}\}$.
A multigraph $H$ is a \emph{submultigraph} of a multigraph $G$ if $V(H) \subseteq V(G)$ and $E(H) \subseteq E(G)$. Let $V'\subseteq V(G)$. We denote the submultigraph induced by $V'$ by $G[V']$, that is, $V(G[V'])=V'$ and $E(G[V'])=\{\{u,v\}\in E(G)~|~ u,v\in V'\}$. Let $U\subseteq V(G)$. Let $G\setminus U$ denote the subgraph $G[V(G)\setminus U]$ of $G$. 

An {\em Eulerian cycle} in a multigraph $\widehat{G}$ is a cycle that visits every edge in $E(\widehat{G})$ exactly once. A {\em vertex cover} of $G$ is $V'\subseteq V(G)$ such that for every $\{u,v\}\in E(G)$, at least one among $u$ and $v$ is in $V'$. The {\em vertex cover number} of $G$ is $\vn(G)=\mathsf{min}\{|V'|~|~V'$ is a vertex cover of $G\}$. When $G$ is clear from context, we refer to $\vn(G)$ as $\vn$. A {\em path} $P$ in $G$ is $(v_0,\ldots,v_\ell)$, where (i) for every $0\leq i\leq \ell$, $v_i\in V(G)$, and (ii) for every $0\leq i\leq \ell-1$, $\{v_i,v_{i+1}\}\in E(G)$ (we allow repeating vertices). The {\em length} of a path $P=(v_0,\ldots,v_\ell)$, denoted by $|P|$, is the number of edges in $P$ (including repetitions), that is, $\ell$. The set of vertices of $P$ is $V(P)=\{v_0,\ldots,v_{\ell-1}\}$. The multiset of edges of $P$ is $E(P)=\{\{v_i,v_{i+1}\}~|~0\leq i\leq \ell-1\}$ (including repetitions). A {\em cycle} $C$ in $G$ is a path $(v_0,\ldots,v_\ell)$ such that $v_0=v_\ell$. A {\em simple cycle} is a cycle $C=(v_0,\ldots,v_\ell)$ such that for every $0\leq i<j\leq \ell-1$, $v_i\neq v_j$. An {\em isomorphism} of a multigraph $G$ into a multigraph $G'$ is a bijection $\alpha:V(G)\rightarrow V(G')$, such that $\{u,v\}$ appears in $E(G)$ $\ell$ times if and only if $\{\alpha(u),\alpha(v)\}$ appears in $E(G')$ $\ell$ times, for an $\ell\in \mathbb{N}$. For a multiset $A$, we denote by $2^A$ the {\em power set} of $A$, that is, $2^A=\{B~|~B\subseteq A\}$. Let $A$ and $B$ be two multisets. Let $A\setminus B$ be the multiset $D\subseteq A$ such that every $d\in A$ appears exactly $\mathsf{max}\{0,d_A-d_B\}$ times in $D$, where $d_A$ and $d_B$ are the numbers of times $d$ appears in $A$ and $B$, respectively. A {\em permutation} of a multiset $A$ is a bijection $\mathsf{Permut}_A:A\rightarrow [|A|]$.

\begin{definition} [{\bf $v_{\mathsf{init}}$-Robot Cycle}]\label{def:RobotWalk}
	Let $G$ be a graph, let $v_{\mathsf{init}}\in V(G)$. A {\em $v_{\mathsf{init}}$-robot cycle} is a cycle $\mathsf{RC}=(v_0=v_{\mathsf{init}},v_1,v_2,\ldots,v_\ell=v_{\mathsf{init}})$ in $G$ for some $\ell$. 
\end{definition}

When $v_{\mathsf{init}}$ is clear from the context, we refer to a $v_{\mathsf{init}}$-robot cycle as a robot cycle.


       
\begin{definition} [{\bf Solution}]\label{def:Sol}
		Let $G$ be a graph, $v_{\mathsf{init}}\in V(G)$ and $k\in \mathbb{N}$. A {\em solution} for $(G,\vi,k)$ is a set of $k$ $\vi$-robot cycles $\{\mathsf{RC}_1,\ldots,\mathsf{RC}_k\}$ with $E(G)\subseteq E(\mathsf{RC}_1)\cup E(\mathsf{RC}_2)\cup \dots \cup E(\mathsf{RC}_k)$. Its {\em value} is $\mathsf{val}(\{\mathsf{RC}_1,\ldots,\mathsf{RC}_k\})=\mathsf{max}\{|E(\mathsf{RC}_1)|,|E(\mathsf{RC}_2)|,\ldots, |E(\mathsf{RC}_k)|\}$ (see \cref{fig:robotCycleA} for an illustration).
\end{definition}

\begin{figure}
	\centering
	\begin{subfigure}[t]{.47\textwidth}
		\centering
		\includegraphics[page=6]{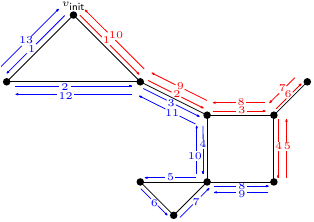}
		\caption{}
		\label{fig:robotCycleA}
	\end{subfigure}\hfill
	\begin{subfigure}[t]{.47\textwidth}
		\centering
		\includegraphics[page=7]{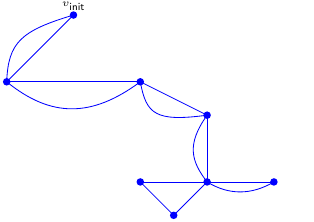}
		\caption{}
		\label{fig:robotCycleB}
	\end{subfigure}
	\caption{(a) An illustration of a graph $G$ (drawn in black) and a solution for $(G, v_{\mathsf{init}}, k = 2)$. The $2$ robot cycles are shown by red and blue edges where the edge labels show the  order in which the edges were covered by the respective robots. (b) The Robot Cycle-Graph for the robot cycle drawn in blue.} 
	\label{fig:robotCycle}
\end{figure}

\begin{definition} [{\bf Collective Graph Exploration with $k$ Agents}]\label{def:kRoPro}
	The \cgF (\cg) problem with $k$ agents is: given a connected graph $G$, $v_{\mathsf{init}}\in V(G)$ and $k\in\mathbb{N}$, find the minimum $B$ such that there exists a solution $\{\mathsf{RC}_1,\ldots,\mathsf{RC}_k\}$ where $\mathsf{val}(\{\mathsf{RC}_1,\ldots,\mathsf{RC}_k\})$ $=B$.
\end{definition}

\begin{definition} [{\bf Collective Graph Exploration with $k$ Agents and Budget $B$}]\label{def:kRoProDec}
	The \cgF (\cg) problem with $k$ agents and budget $B$ is: given a 
	connected graph $G$, $v_{\mathsf{init}}\in V(G)$ and $k,B\in\mathbb{N}$, find a solution $\{\mathsf{RC}_1,\ldots,\mathsf{RC}_k\}$ where $\mathsf{val}(\{\mathsf{RC}_1,\ldots,\mathsf{RC}_k\})\leq B$, if such a solution exists; otherwise, return ``no-instance''.
\end{definition}

\begin{definition} [{\bf Bin Packing}]\label{def:BinPack}
	The {\sc Bin Packing} problem is: given a finite set $I$ of items, a size $s(i)\in \mathbb{N}$ for each $i\in I$, a positive integer $B$ called bin capacity and a positive integer $k$, decide whether there is a partition of $I$ into disjoint sets $I_1,\ldots,I_k$ such that for every $1\leq j\leq k$, $\sum_{i\in I_j} s(i)\leq B$.
\end{definition}

\begin{definition} [{\bf Exact Bin Packing}]\label{def:ExBinPack}
	The {\sc Exact Bin Packing} problem is: given a finite set $I$ of items, a size $s(i)\in \mathbb{N}$ for each $i\in I$, a positive integer $B$ called bin capacity and a positive integer $k$ such that $\sum_{i\in I}s(i)=B\cdot k$, decide whether there is a partition of $I$ into disjoint sets $I_1,\ldots,I_k$ such that for every $1\leq j\leq k$, $\sum_{i\in I_j} s(i)=B$.
\end{definition}

\begin{definition}[{\bf{Integer Linear Programming}}] In the {\sc Integer Linear Programming Feasibility} (ILP) problem, the input consists of $t$ variables $x_1, x_2, \ldots, x_t$ and a set of $m$ inequalities of the following form:
\[\begin{array}{*{9}{@{}c@{}}}
	a_{1,1}x_1 & + & a_{1,2}x_1 &+ & \cdots & + & a_{1,p}x_t & \leq & b_1 \\
	a_{2,1}x_1 & + & a_{2,2}x_2 & + & \cdots & + & a_{2,p}x_t & \leq & b_2 \\
	\vdots    &   & \vdots    &   &        &   & \vdots    &   & \vdots \\
	a_{m,1}x_1 & + & a_{m,2}x_2 & + & \cdots & + & a_{m,p}x_t & \leq & b_m \\
\end{array}\]
where all coefficients $a_{i,j}$ and $b_i$ are required to integers. The task is to check whether there exist integer values for every variable $x_i$ so that all inequalities are satisfiable. 
\end{definition}

\begin{theorem}[\cite{DBLP:journals/combinatorica/FrankT87,DBLP:journals/mor/Lenstra83,DBLP:journals/mor/Kannan87}]\label{the:runningTimeILP}
An ILP instance of size $m$ with $t$ variables can be solved in time $t^{\OO(t)}\cdot m^{\OO(1)}$. 
\end{theorem}



\section{Reinterpretation Based on Eulerian Cycles}\label{sec:ourapp}

Our approach to \cg with $k$ agents is as follows. Let $G$ be a connected graph, let $\vi\in V(G)$ and let $k\in \mathbb{N}$. Let $\{\mathsf{RC}_1,\ldots,\mathsf{RC}_k\}$ be a solution, let $1\leq i\leq k$ and denote $\mathsf{RC}_i=(v_0=v_{\mathsf{init}},v_1,v_2,\ldots,v_\ell=v_{\mathsf{init}})$ for some $\ell\in \mathbb{N}$. If we define a multiset $\widehat{E}_\mathsf{RC_i}=\{\{v_j,v_{j+1}\}~|~0\leq j\leq \ell-1\}$, then, clearly, $\mathsf{RC}_i=(v_0=v_{\mathsf{init}},v_1,v_2,\ldots,v_\ell=v_{\mathsf{init}})$ is an Eulerian cycle in $\gr(\widehat{E}_\mathsf{RC_i})$. We call this graph the {\em $\mathsf{RC}_i$-graph} (see \cref{fig:robotCycleB}):

\begin{definition} [{\bf Robot Cycle-Graph}]\label{def:RWGra}
	Let $G$ be a graph, let $v_{\mathsf{init}}\in V(G)$ and let $\mathsf{RC}=(v_0=v_{\mathsf{init}},v_1,v_2,\ldots,v_\ell=v_{\mathsf{init}})$ be a robot cycle. The {\em $\mathsf{RC}$-graph}, denoted by $\gr(\mathsf{RC})$, is the multigraph $\gr(\widehat{E}_\mathsf{RC})$, where $\widehat{E}_\mathsf{RC}=\{\{v_i,v_{i+1}\}~|~0\leq i\leq \ell-1\}$ is a multiset.
\end{definition}

\begin{observation}\label{obs:robIsEu}
	Let $G$ be a graph, let $v_{\mathsf{init}}\in V(G)$ and let $\mathsf{RC}=(v_0=v_{\mathsf{init}},v_1,v_2,\ldots,v_\ell=v_{\mathsf{init}})$ be a robot cycle. Then $\mathsf{RC}$ is an Eulerian cycle in $\gr(\mathsf{RC})$.
\end{observation}

On the opposite direction, let $\widehat{E}$ be a multiset with elements from $E(G)$, and assume that $\vi\in V(\gr(\widehat{E}))$.  Let $\mathsf{RC}=(v_0,v_1,v_2,\ldots,v_\ell=v_0)$ be an Eulerian cycle in $\gr(\widehat{E})$ and assume, without loss of generality, that $v_0=v_\ell=v_{\mathsf{init}}$. It is easy to see that $\mathsf{RC}$ is a robot cycle in $G$:

\begin{observation}\label{obs:EuiIsRob}
	Let $G$ be a graph, let $v_{\mathsf{init}}\in V(G)$, let $\widehat{E}$ be a multiset with elements from $E(G)$ and assume that $\vi\in V(\gr(\widehat{E}))$. Let $\mathsf{RC}=(v_0=v_{\mathsf{init}},v_1,v_2,\ldots,v_\ell=v_{\mathsf{init}})$ be an Eulerian cycle in $\gr(\widehat{E})$. Then, $\mathsf{RC}$ is a robot cycle in $G$.
\end{observation}

From \cref{obs:robIsEu,obs:EuiIsRob}, we get that finding a solution is equal to find $k$ multisets $\widehat{E}_1,\ldots,\widehat{E}_k$ such that: (i) for every $1\leq i\leq k$, $\vi\in V(\gr(\widehat{E}_i))$ (ii) for every $1\leq i\leq k$, there exists an Eulerian cycle in $\gr(\widehat{E}_i)$ and (iii) $E(G)\subseteq \widehat{E}_1\cup \ldots \cup\widehat{E}_k$, that is, each $e\in E$ appears at least once in at least one of $\widehat{E}_1,\ldots,\widehat{E}_k$.

Recall that, in a multigraph $\widehat{G}$, there exists an Eulerian cycle if and only if $\widehat{G}$ is connected and each $v\in V(\widehat{G})$ has even degree in $\widehat{G}$ \cite{bollobas1998modern}. Thus, we have the following lemma:

\begin{lemma}\label{obs:equivsol}
	Let $G$ be a connected graph, let $\vi\in V(G)$ and let $k,B\in \mathbb{N}$. Then, $(G,\vi,k,B)$ is a yes-instance of \cg if and only if there exist $k$ multisets $\widehat{E}_1,\ldots,\widehat{E}_k$ with elements from $E(G)$, such that the following conditions hold: 
	\begin{enumerate}
		\item For every $1\leq i\leq k$, $\vi\in V(\gr(\widehat{E}_i))$. \label{obs:equivsol1}
		\item\label{con:2EquiSol} For every $1\leq i\leq k$, $\gr(\widehat{E}_i)$ is connected, and every vertex in $\gr(\widehat{E}_i)$ has even degree. \label{obs:equivsol2}
		\item $E(G)\subseteq \widehat{E}_1\cup \ldots \cup\widehat{E}_k$.\label{obs:equivsol3}
		\item$ \mathsf{max}\{|\widehat{E}_1|,\ldots, |\widehat{E}_k|\}\leq B$.\label{obs:equivsol4}
	\end{enumerate} 
\end{lemma}

\section{High-Level Overview}\label{sec:overview}

\subsection{FPT Algorithm with Respect to Vertex Cover}

Our algorithm is based on a reduction to the ILP problem. 
We aim to construct linear equations that verify the conditions in \cref{obs:equivsol}.


\subsubsection{Encoding $\widehat{E}_i$ by a Valid Pair}
First, we aim to satisfy the ``local'' conditions of \cref{obs:equivsol} for each robot, that is, Conditions~\ref{obs:equivsol1} and~\ref{obs:equivsol2}. Let us focus on the ``harder'' condition of the two, that is, Condition~\ref{obs:equivsol2}. We aim to encode any potential $\widehat{E}_i$ by smaller subsets whose union is $\widehat{E}_i$. In addition, we would like the ``reverse'' direction as well: every collection of subsets that we will be able to unite must create some valid $\widehat{E}_i$. Note that we have two goals to achieve when uniting the subsets together: (i) derive a connected graph, where (ii) each vertex has even degree. In the light of this, the most natural encoding for the subsets are cycles, being the simplest graphs satisfying both aforementioned goals. Indeed, every cycle is connected, and a graph composed only of cycles is a graph where every vertex has even degree. Here, the difficulty is to maintain the connectivity of the composed graph. On the positive side, observe that every cycle in the input graph $G$ has a non-empty intersection with any vertex cover $\vc$ of $G$. So, we deal with the connectivity requirement as follows. We seek for a graph $\overline{G}$ that is essentially (but not precisely) a subgraph of $G$ that is (i) ``small'' enough, and (ii)  for every valid $\widehat{E}_i$, there exists $\mathsf{CC}\subseteq E(\overline{G})$ such that $\gr(\mathsf{CC})$ is a ``submultigraph'' of $\gr(\widehat{E}_i)$, $\gr(\mathsf{CC})$ is connected, and $V(\gr(\mathsf{CC}))\cap \vc=V(\gr(\widehat{E}_i))\cap \vc$. 

\noindent{\bf Equivalence Graph $G^*$.} A first attempt to find such a graph is as follows. We define an equivalence relation on $V(G)\setminus \vc$ based on the sets of neighbors of the vertices in $V(G)\setminus \vc$ (see \cref{def:ER}) (see the $4$ equivalence classes of the graph $G$ in \cref{fig:fptA}). We denote the set of equivalence classes induced by this equivalence relation by $\mathsf{EQ}$. Then $G^*$ is the graph defined as follows (for more details, see \cref{def:VCG}).

\begin{definition} [{\bf Equivalence Graph $G^*$}]
	Let $G^*$ be the graph that: (i) contains $\vc$, and the edges having both endpoints in $\vc$, and (ii) where every equivalence class $u^*\in \ind$ is represented by a single vertex adjacent to the neighbors of some $u\in u^*$ in $G$ (which belong to $\vc$). See \cref{fig:fptB}.
\end{definition}

Unfortunately, this attempt fails, as we might need to use more than one vertex from the same $u^*\in \ind$ in order to maintain the connectivity. E.g., see \cref{fig:fpt1B}. If we  delete $r_2$ and $y_5$, which are in the same equivalence class (in $G$) as $r_1$ and $y_6$, respectively, then the graph is no longer connected. 

\begin{figure}[t]
	\centering
	\begin{subfigure}[t]{.3\textwidth}
		\centering
		\includegraphics[page=27, width=\textwidth]{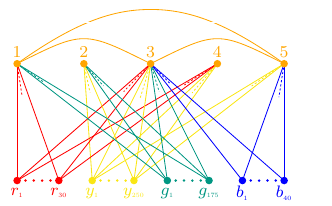}
		\caption{$G$}
		\label{fig:fptA}
	\end{subfigure}\hfill
	\begin{subfigure}[t]{.3\textwidth}
		\centering
		\includegraphics[page=28, width=\textwidth]{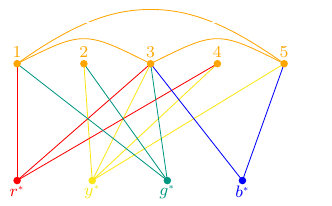}
		\caption{$G^*$}
		\label{fig:fptB}
	\end{subfigure}\hfill
	\begin{subfigure}[t]{.3\textwidth}
		\centering
		\includegraphics[page=29, width=\textwidth]{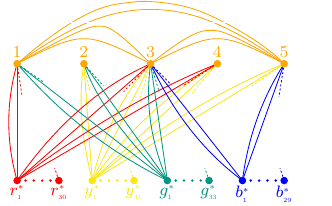}
		\caption{$\overline{G}$}
		\label{fig:fptC}
	\end{subfigure}
	\caption{An illustration of a graph $G$ (in (a)), and its corresponding graphs $G^*$ (in (b)) and $\overline{G}$ (in (c)). The vertex cover vertices and their edges are shown in orange. The $4$ equivalence classes and their vertices are shown by red, yellow, green, and blue.}
	\label{fig:fpt}
\end{figure}

\noindent{\bf The Multigraph $\overline{G}$.} So, consider the following second attempt. We use the aforementioned graph $G^*$, but instead of one vertex representing each $u^*\in \ind$, we have $\mathsf{min}\{|u^*|,2^{|\mathsf{N}_{G^*}({u^*})|}\}$ vertices. Observe that given a connected subgraph $G'$ of $G$, and two vertices $u,u'\in u^*$ such that $\mathsf{N}_{G'}(u)=\mathsf{N}_{G'}(u')$, it holds that $G\setminus\{u'\}$ remains connected (e.g., see  \cref{fig:fpt1A} and \ref{fig:fpt1B}. The connectivity is still maintained even after deleting all but one vertex in the same equivalence class (in $G$) having same neighbourhood). Therefore, we have enough vertices for each $u^*\in \ind$ in the graph, and its size is a function of $|\vc|$; so, we obtained the sought graph $\overline{G}$. 
Now, we would like to have an additional property for $\mathsf{CC}$, which is that every vertex in $\gr(\mathsf{CC})$ has even degree in it. To this end, we add to $\overline{G}$ more vertices for each $u^*\in \ind$. See  \cref{fig:fpt1C}. The vertex $g_{13}$ having the same neighbours as $g_{11}$ in $H$ and being in the same equivalence class (in $G$) as $g_{11}$ is added to make the degrees of $1$ and $2$ even. We have the following definition for $\overline{G}$ (for more details, see \cref{def:GbarGraph} and the discussion before this definition). 

\begin{definition} [{\bf The Multigraph $\overline{G}$}]
	Let $\overline{G}$ be the graph that: (i) contains $\vc$, and the edges having both endpoints in $\vc$, (ii) for every equivalence class $u^*\in \ind$, there are exactly  $\mathsf{min}\{|u^*|,2^{|\mathsf{N}_{G^*}({u^*})|}+|\vc|^2\}$ vertices, adjacent to the neighbors of some $u\in u^*$ in $G$ (which belong to $\vc$), (iii) each edge in $\overline{G}$ appears exactly twice in $E(\overline{G})$ (for technical reasons). See \cref{fig:fptC}.
\end{definition} 


\noindent{\bf A Skeleton of $\widehat{E}_i$.} We think of $\gr(\mathsf{CC})$ as a ``skeleton'' of a potential $\widehat{E}_i$. By adding cycles with a vertex from $V(\gr(\mathsf{CC}))\cap \vc$, we maintain the connectivity, and since every vertex in $\gr(\mathsf{CC})$ has even degree, then by adding a cycle, this property is preserved as well. We have the following definition for a skeleton.

\begin{figure}[t]
	\centering
	\begin{subfigure}[t]{.3\textwidth}
		\centering
		\includegraphics[page=30, width=\textwidth]{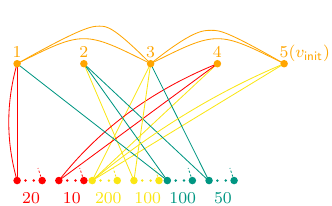}
		\caption{$\gr(\widehat{E})$}
		\label{fig:fpt1A}
	\end{subfigure}\hfill
	\begin{subfigure}[t]{.3\textwidth}
		\centering
		\includegraphics[page=31, width=\textwidth]{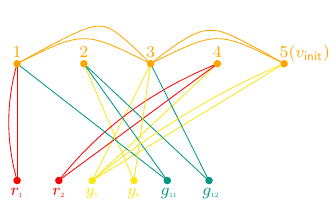}
		\caption{$H$}
		\label{fig:fpt1B}
	\end{subfigure}\hfill
	\begin{subfigure}[t]{.3\textwidth}
		\centering
		\includegraphics[page=32, width=\textwidth]{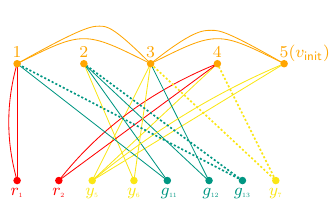}
		\caption{$\gr(\mathsf{CC})$}
		\label{fig:fpt1C}
	\end{subfigure}
	\begin{subfigure}[t]{.3\textwidth}
	\centering
	\includegraphics[page=35, width=\textwidth]{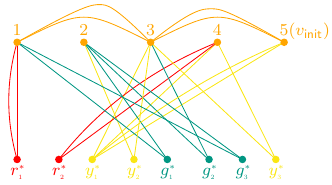}
	\caption{$\gr(\mathsf{CC}')$}
	\label{fig:robotAllocation}
\end{subfigure}

	\caption{The graphs shown here are with respect to the graph $G$ shown in \cref{fig:fptA}. An illustration of (a) a graph $\gr(\widehat{E})$, (b) the graph $H$ obtained by deleting all but one vertex from the same equivalence class in $G$ and have the same neighbours in $\gr(\widehat{E})$, (c) the graph $\gr(\mathsf{CC})$ where $\mathsf{CC}$ is a skeleton of $\widehat{E}$ obtained from the graph in (b) by adding four more edges from $\gr(\widehat{E}) \setminus H$, and (d) the graph $\gr(\mathsf{CC}')$ where $\mathsf{CC}'$ is the skeleton in $\overline{G}$ that is derived from the skeleton $\mathsf{CC}$.}
	\label{fig:fpt1}
\end{figure}

\begin{definition} [{\bf A Skeleton $\mathsf{CC}$}]
	A {\em skeleton} of $\widehat{E}_i$ is $\mathsf{CC}\subseteq \widehat{E}_i$ such that: (i) $\gr(\mathsf{CC})$ is a ``submultigraph'' of $\overline{G}$, (ii) $\gr(\mathsf{CC})$ is connected, $\vi \in V(\gr(\mathsf{CC}))$ and every vertex in $\gr(\mathsf{CC})$ has even degree, and (iii) $V(\gr(\mathsf{CC}))\cap \vc=V(\gr(\widehat{E}_i))\cap\vc$ (See \cref{fig:robotAllocation}).
\end{definition} 

\noindent{\bf An $\widehat{E}_i$-Valid Pair.} In \cref{obs:equivso}, we prove that we might assume that the $\widehat{E}_i$'s are {\em nice multisets}, that is, a multiset where every element appears at most twice. In \cref{lem:4cycandCC}, we prove that every $\widehat{E}_i$ (assuming $\widehat{E}_i$ is nice) can be encoded by a skeleton $\mathsf{CC}$ (See \cref{fig:fpt1C}.) and a multiset ${\cal C}$ of cycles (of length bounded by $2|\vc|$). We say that $(\mathsf{CC},{\cal C})$ is an {\em $\widehat{E}$-valid pair} (for more details, see \cref{def:ValPa} and \cref{sec:ValPa}).

\begin{definition} [{\bf A Valid Pair}]
	A pair $(\mathsf{CC},{\cal C})$, where $\mathsf{CC}$ is a skeleton of $\widehat{E}_i$ and ${\cal C}$ is a multiset of cycles in $\gr(\widehat{E}_i)$, is an {\em $\widehat{E}_i$-valid pair} skeleton $\mathsf{CC}$ if: 
	\begin{enumerate}
		\item The length of each cycle in ${\cal C}$ is bounded by $2|\vc|$.
		\item At most $2|\vc|^2$ cycles in ${\cal C}$ have length other than $4$. 
		\item $\mathsf{CC}\cup \bigcup_{C\in {\cal C}}E(C)=\widehat{E}_i$ (being two multisets). 
	\end{enumerate}
\end{definition}

\subsubsection{Robot and Cycle Types}
Now, obviously, the number of different cycles in $G$ (of length bounded by $2|\vc|$) is potentially huge. Fortunately, it is suffices to look at cycles in $G^*$ in order to preserve Condition~\ref{obs:equivsol2} of \cref{obs:equivsol}: assume that we have a connected $\gr(\mathsf{CC})$ such that every vertex in $\gr(\mathsf{CC})$ has even degree in it, and a multiset of cycles with a vertex from $V(\gr(\mathsf{CC}))\cap \vc$ in $G^*$. By replacing each vertex that represents $u^*\in \ind$ by any $u\in u^*$, the connectivity preserved, and the degree of each vertex is even. 

Thus, each robot is associated with a {\em robot type} $\mathsf{RobTyp}$, which includes a skeleton $\mathsf{CC}$ of the multiset $\widehat{E}_i$ associated with the robot (along other information discussed later). In order to preserve Condition~\ref{obs:equivsol1} of \cref{obs:equivsol}, we also demand that $\vi\in V(\gr(\mathsf{CC}))$. Generally, for each type we define, we will have a variable that stands for the number of elements of that type. We are now ready to present our first equation of the ILP reduction:

\smallskip\noindent{\bf Equation~\ref{ilp:1}: Robot Type for Each Robot.} In this equation, we ensure that the total sum of robots of the different robot types is exactly $k$, that is, there is exactly one robot type for each robot:

1. $\displaystyle{\sum_{\mathsf{RobTyp}\in \mathsf{RobTypS}}x_\mathsf{RobTyp}=k}$.

In addition, the other ``pieces'' of the ``puzzle'', that is, the cycles, are also represented by types: Each cycle $C$ of length at most $2|\vc|$ in $G^*$ is represented by a {\em cycle type}, of the form $\mathsf{CycTyp}=(C,\mathsf{RobTyp})$ (along other information discussed later), where $\mathsf{RobTyp}$ is a robot type that is ``able to connect to $C$'', that is, $V(\gr(\mathsf{CC}))\cap \vc\cap V(C)\neq\emptyset$ for $\mathsf{RobTyp}=\mathsf{CC}$. Similarly, we will have equations for our other types.

\noindent{\bf Satisfying the Budget Restriction.} Now, we aim to satisfy the budget condition (Condition~\ref{obs:equivsol2} of \cref{obs:equivsol}), that is, for every $i\in[k]$, $|\widehat{E}_i|\leq B$. Let $i\in [k]$ and let $(\mathsf{CC},{\cal C})$ be an $\widehat{E}$-valid pair. So, $\widehat{E}_i=\mathsf{CC}\cup(\bigcup_{C\in {\cal C}}E(C))$ (being a union of two multisets). Now, in \cref{lem:4cycandCC}, we prove that ``most'' of the cycles in ${\cal C}$ are of length $4$, that is, for every $2\leq j\leq 2|\vc|, j\neq 4$, the number of cycles of length $j$ in ${\cal C}$ is bounded by $2|\vc|^2$. Therefore, we add to the definition of a robot type also the number of cycles of length exactly $j$, encoded by a vector $\mathsf{NumOfCyc}=(N_2,N_3,N_5,N_6,\ldots,N_{2|\vc|})$. So, for now, a robot type is $\mathsf{RobTyp}=(\mathsf{CC},\mathsf{NumOfCyc})$. Thus, in order to satisfy the budget condition, we verify that the budget used by all the robots of a robot type $\mathsf{RobTyp}=(\mathsf{CC},\mathsf{NumOfCyc})$, is as expected together. First, we ensure that the number of cycles of each length $2\leq j\leq 2|\vc|, j\neq 4$, is exactly as the robot type demands, times the number of robots associated with this type, that is, $N_j\cdot x_\mathsf{RobTyp}$. So, we have the following equation:

\smallskip\noindent{\bf Equation~\ref{ilp:5}: Assigning the Exact Number of Cycles of Length Other Than $4$ to Each Robot Type.} We have the following notation: $\mathsf{CycTypS}(\mathsf{RobTyp},j)$ is the set of cycle types for cycles of length $j$ assigned to a robot of robot type $\mathsf{RobTyp}$. 

5. For every robot type $\mathsf{RobTyp}=(\mathsf{CC},\mathsf{NumOfCyc})$ and for every $2\leq j\leq 2|\vc|$, $j\neq 4$, $\displaystyle{\sum_{{\mathsf{CycTyp}\in \mathsf{CycTypS}(\mathsf{RobTyp},j)}}x_\mathsf{CycTyp}=N_j\cdot x_\mathsf{RobTyp}}$, where $\mathsf{NumOfCyc}=(N_2,N_3,N_5,N_6,\ldots,$ $N_{2|\vc|})$.

Observe that once this equation is satisfied, we are able to arbitrary allocate $N_j$ cycles of length $j$ to each robot of type $\mathsf{RobTyp}$. So, in order to verify the budget limitation, we only need to deal with the cycles of length $4$. 
Now, notice that the budget left for a robot of type $\mathsf{RobTyp}=(\mathsf{CC},\mathsf{NumOfCyc})$ for the cycles of length $4$ is $B-(|\mathsf{CC}|+\sum_{2\leq j\leq 2|\vc|,j\neq 4}N_j\cdot j)$, where $\mathsf{NumOfCyc}=(N_2,N_3,N_5,N_6,\ldots,N_{2|\vc|})$. Now, the maximum number of cycles we can add to a single robot of type $\mathsf{RobTyp}$ is the largest number which is a multiple of $4$, that is less or equal to $B-\mathsf{Bud}(\mathsf{RobTyp})$. So, for every robot type $\mathsf{RobTyp}$ let $\mathsf{CycBud}(\mathsf{RobTyp})=\lfloor(B-(|\mathsf{CC}|+\sum_{2\leq j\leq 2|\vc|,j\neq 4}N_j\cdot j))\cdot \frac{1}{4}\rfloor\cdot 4$. Notice that $\mathsf{CycBud}(\mathsf{RobTyp})$ is the budget left for the cycles of length $4$. Thus, we have the following equation:

\smallskip\noindent{\bf Equation~\ref{ilp:6}: Verifying the Budget Limitation.} This equation is defined as follows. 

\noindent6. For every $\mathsf{RobTyp}\in \mathsf{RobTypS}$,\\ $\displaystyle{\sum_{{\mathsf{CycTyp}\in \mathsf{CycTypS}(\mathsf{RobTyp},4)}}4\cdot x_\mathsf{CycTyp}\leq x_\mathsf{RobTyp}\cdot \mathsf{CycBud}(\mathsf{RobTyp})}$.

By now, we have that there exist $\widehat{E}_1,\ldots,\widehat{E}_k$ that satisfy Conditions~\ref{obs:equivsol1}, \ref{obs:equivsol2} and~\ref{obs:equivsol4} of \cref{obs:equivsol} if and only if Equations~\ref{ilp:1}, \ref{ilp:5} and~\ref{ilp:6} can be satisfied.

\noindent{\bf Covering Edges with Both Endpoints in $\vc$.} Now, we aim to satisfy Condition~\ref{obs:equivsol3} of \cref{obs:equivsol}, that is, we need to verify that every edge is covered by at least one robot. First, we deal with edges with both endpoints in $\vc$. Here, for every $\{u,v\}$ such that $u,v\in \vc$, we just need to verify that at least one cycle or one of the $\mathsf{CC}$'s contains $\{u,v\}$. This we can easily solve by the following equation:

\smallskip\noindent{\bf Equation~\ref{ilp:4}: Covering Each Edge With Both Endpoints in $\vc$.} We have the following notations:
For every $\{u,v\}\in E$ such that $u,v\in \vc$, (i) let $\mathsf{CycTypS}(\{u,v\})$ be the set of cycle types $\mathsf{CycTyp}=(C,\mathsf{RobTyp})$ where $C$ covers $\{u,v\}$,
and (ii) let $\mathsf{RobTypS}(\{u,v\})$ be the set of robot types $\mathsf{RobTyp}=(\mathsf{CC},\mathsf{NumOfCyc})$ where $\mathsf{CC}$ covers $\{u,v\}$.
In this equation, we ensure that each $\{u,v\}\in E(G)$ with both endpoints in $\vc$ is covered at least once:

4. For every $\{u,v\}\in E$ such that $u,v\in \vc$,

\noindent$\displaystyle{\sum_{{\mathsf{CycTyp}\in \mathsf{CycTypS}(\{u,v\})}}x_\mathsf{CycTyp}+\sum_{{\mathsf{RobTyp}\in \mathsf{RobTypS}(\{u,v\})}}}$ $x_\mathsf{RobTyp}\geq 1$.

Let $\mathsf{RobTypS}$ be the set of the robot types, and let $\mathsf{CycTypS}$ be the set of cycle types.

\noindent{\bf Covering Edges with an Endpoint in $V(G)\setminus \vc$.} Now, we aim to cover the edges from $E(G)$ with (exactly) one endpoint in $V(G)\setminus \vc$. Here, we need to work harder. Let $x_z$, for every $z\in \mathsf{RobTypS}\cup \mathsf{CycTypS}$, be values that satisfy Equations~\ref{ilp:1} and \ref{ilp:4}--\ref{ilp:6}. As for now, we will arbitrary allocate cycles to robots according to their types. Then, we will replace every $u^*\in V(\gr(\mathsf{CC}_i))$ and $u^*\in V(C)$, for every cycle $C$ allocated to the $i$-th robot, by an arbitrary $u\in u^*$. Then, we will define $\widehat{E}_i$ as the union of edge set of the cycles and $\mathsf{CC}_i$ we obtained. We saw that due to Equations~\ref{ilp:1}, \ref{ilp:5} and~\ref{ilp:6}, 
Conditions~\ref{obs:equivsol1}, \ref{obs:equivsol2} and~\ref{obs:equivsol4} of \cref{obs:equivsol} are satisfied. In addition, due to Equations~\ref{ilp:4}, we ensure that each $\{u,v\}\in E(G)$ with both endpoints in $\vc$ is covered.
The change we need to do in order to cover edges with an endpoint in $V(G)\setminus \vc$ is to make a smarter choices for the replacements of $u^*$ vertices.

\subsubsection{Vertex Type}
\noindent{\bf Allocation of Multisets with Elements from $\mathsf{N}_{G^*}({u^*})$.} Observe that each $u^*_j\in V(\gr(\mathsf{CC}_i))$ that is replaced by some $u\in u^*$, covers the multiset of edges $\{\{u,v\}~|~v\in \widehat{\mathsf{N}}_{\gr(\mathsf{CC}_i)}({u^*_j})\}$. In addition, every $u^*\in V(C)$ that is replaced by $u\in u^*$, covers the multiset of edges $\{\{u,v\},\{u,v'\}\}$, where $v$ and $v'$ are the vertices right before and right after $u$ in $C$, respectively. Now, in order to cover every edge with an endpoint in $V(G)\setminus \vc$, we need to cover the set $\{\{u,v\}~|~v\in \mathsf{N}_{G^*}({u^*})\}$ for every $u\in u^*\in \ind$. Therefore, we would like to ensure that the union of multisets of neighbors ``allocated'' for each $u$, when we replace some $u^*$ by $u$, contains $\{\{u,v\}~|~v\in \mathsf{N}_{G^*}({u^*})\}$. 

\begin{figure}[t]
	\centering
	\includegraphics[page=34, width=0.3\textwidth]{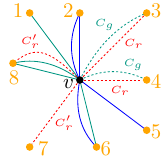}
	\caption{An illustration of the parts of a solution around an independent set vertex $v$. The three colors represent the parts of the multisets corresponding to three robots. The solid edges belong to the skeleton of the specific robot. The dashed edges belong to a cycle, labelled in the figure, of the multiset of the cycles corresponding to the specific robot. The vertex type of $v$ derived from the solution shown in the figure is $(v^*, \{\{1,6,8,8\}, \{3,4\}, \{7,8\}, \{2,2,5,6\}\})$, where $v \in v^* \in \ind$.}
	\label{fig:vertexType}
\end{figure}

\noindent{\bf The Set $\U$ of Multisets Needed to Allocate to a Vertex.} Now, the reverse direction holds as well: let $\widehat{E}_1,\ldots, \widehat{E}_k$ be multisets satisfying the conditions of \cref{obs:equivsol}, for every $i\in[k]$, let $(\mathsf{CC}_i,{\cal C}_i)$ be an $\widehat{E}_i$-valid pair, and let $u\in u^*\in \ind$. Consider the following multisets (*): (i) for every $i\in [k]$ such that $u\in V(\gr(\mathsf{CC}_i))$, the multiset $\widehat{\mathsf{N}}_{\gr(\mathsf{CC}_i)}(u)$; (ii) for every $i\in [k]$ and $C\in{\cal C}_i$ and every appearance of $u$ in $C$, the multiset $\{v,v'\}$, where $v$ and $v'$ are the vertices in $C$ right before and right after the appearance of $u$. By Condition~\ref{obs:equivsol3} of \cref{obs:equivsol}, every edge appears in at least one among $\widehat{E}_1,\ldots, \widehat{E}_k$. So, as for every $i\in[k]$, $\widehat{E}_i=\mathsf{CC}_i\bigcup_{C\in {\cal C}_i}E(C)$, the union of the multisets in (*) obviously contains $\mathsf{N}_{G^*}({u^*})$, e.g. see \cref{fig:vertexType}. We would like to store the information of these potential multisets that ensures we covered $\mathsf{N}_{G^*}({u^*})$. The issue is that there might be a lot of multisets, as $u$ might appear in many $\widehat{E}_i$'s. Clearly, it is sufficient to store one copy of each such multiset, as we only care that the union of the multisets contains $\mathsf{N}_{G^*}({u^*})$. Now, as we assume that $\widehat{E}_1,\ldots, \widehat{E}_k$ are nice multisets, each element in every multiset we derived appears at most twice in that multiset. In addition, since every edge in $E(\overline{G})$ appears at most twice, for each skeleton $\mathsf{CC}\subseteq E(\overline{G})$, each edge appears at most twice in $E(\gr(\mathsf{CC}))$. So, for each $u^*_j\in V(\gr(\mathsf{CC}))$ we replace by some $u\in u^*$, in the multiset of neighbors that are covered, every element appears at most twice. Moreover, since the degree is even, we have that the number of element in each multiset is even.

For a set $A$ we define the multiset $A\times \mathsf{2}=\{a,a~|~a\in A\}$. That is, each element in $A$ appears exactly twice in $A\times \mathsf{2}$.
Thus, we have the following definition for a vertex type (for more details, see \cref{sec:VerTyp}).

\begin{definition} [{\bf Vertex Type}]
	Let $G$ be a connected graph and let $\vc$ be a vertex cover of $G$. Let $u^*\in \ind$ and let $\mathsf{NeiSubsets}\subseteq 2^{\mathsf{N}_{G^*}(u^*)\times \mathsf{2}}$. Then, $\mathsf{VerTyp}=(u^*,\mathsf{NeiSubsets})$ is a {\em vertex type} if
	for every $\mathsf{NeiSub}\in \mathsf{NeiSubsets}$, $|\mathsf{NeiSub}|$ is even, and $\mathsf{N}_{G^*}(u^*)\subseteq \bigcup\mathsf{NeiSubsets}$.
\end{definition}

Now, given $\widehat{E}_1,\ldots, \widehat{E}_k$ satisfying the conditions of \cref{obs:equivsol}, for every $i\in[k]$, an $\widehat{E}_i$-valid pair $(\mathsf{CC}_i,{\cal C}_i)$, and $u\in u^*\in \ind$, we derive the vertex type of $u$ as follows. We take the set $\mathsf{NeiSubsets}$ of multisets as described in (*). Clearly, $(u^*,\mathsf{NeiSubsets})$ is a vertex type (for more details, see \cref{def:VTD} and \cref{lem:VerDer}). 

For the reverse direction, we will use vertex type in order to cover the edges incident to each $u\in u^*\in \ind$. Let $\mathsf{VerTypS}$ be the set of vertex types. We have a variable $x_z$ for every $z\in \mathsf{VerTypS}$. First, each $u\in u^*\in \ind$ is associated with exactly one vertex type $\mathsf{VerTyp}=(u^*,\mathsf{NeiSubsets})$, for some $\mathsf{NeiSubsets}$. To achieve this, we first ensure that for every $u^*\in \ind$, the total sum of $x_z$ for $z\in\mathsf{VerTypS}_{u^*}$, is exactly $|u^*|$, where $\mathsf{VerTypS}_{u^*}=\{(u^*,\U)\in \mathsf{VerTypS}\}$. 

\smallskip\noindent{\bf Equation~\ref{ilp:2}: Vertex Type for Each Vertex.} This equation is defined as follows.

\noindent2. For every $u^*\in \ind$, $\displaystyle{\sum_{\mathsf{VerTyp}\in \mathsf{VerTypS}_{u^*}} x_\mathsf{VerTyp}=|u^*|}$

Given values for the variables that satisfy the equation, we arbitrary determine a vertex type $(u^*,\mathsf{NeiSubsets})$ for each $u\in u^*$, such that there are exactly $x_\mathsf{VerTyp}$ vertices of type $\mathsf{VerTyp}$. 

\smallskip\noindent{\bf Allocation Functions of Multisets to Vertex Types.} Now, let $u\in u^*\in \ind$ of a vertex type $(u^*,\mathsf{NeiSubsets})$. We aim that when we do the replacements of $u^*$'s by vertices from $u^*$, each $u$ gets an allocation of at least one of any of the multisets in $\mathsf{NeiSubsets}$. This ensures that we covered all of the edges adjacent to $u$. Instead of doing this for each $u\in u^*\in \ind$, we will ensure that each $\W\in\mathsf{NeiSubsets}$ is allocated for vertices of type $(u^*,\mathsf{NeiSubsets})$ at least $x_\mathsf{VerTyp}$ times. To this end, we add more information for the robot types. For a robot type with a skeleton $\mathsf{CC}$, recall that we replace each $u^*_j\in V(\gr({\mathsf{CC})})$ by some $u\in u^*$. The robot type also determines what is the vertex type of $u$ that replaces $u^*_j$. In particular, we add to the robot type an allocation for each of  $\{(u^*_j,\widehat{\mathsf{N}}_{\gr({\mathsf{CC})}}(u^*_j))~|~u^*_j\in V(\gr({\mathsf{CC})})\}$, that is, a function $\mathsf{Alloc}_{\gr(\mathsf{CC})}$ from this set into $\mathsf{VerTypS}$ (e.g., a robot of a robot type associated with the skeleton illustrated by \cref{fig:robotAllocation}, needs to allocate the pair $(r^*_1,\{1,1\})$, along with the other pairs shown in the figure). Observe that $u_j^*$ is the vertex being replaced, and $\widehat{\mathsf{N}}_{\gr({\mathsf{CC})}}(u^*_j)$ is the multiset of neighbors that are covered. So, we demand that each $(u^*_j,\widehat{\mathsf{N}}_{\gr({\mathsf{CC})}}(u^*_j))$ is allocated to a vertex type $(u^*,\mathsf{NeiSubsets})$ that ``wants'' to get $\widehat{\mathsf{N}}_{\gr({\mathsf{CC})}}(u^*_j)$, that is, $\widehat{\mathsf{N}}_{\gr({\mathsf{CC})}}(u^*_j)\in \mathsf{NeiSubsets}$ (e.g, a robot of a robot type associated with the skeleton illustrated by \cref{fig:robotAllocation}, might allocate $(r^*_1,\{1,1\})$ to a vertex type $(r^*,\{\{1,1\},\{3,4\}\})$). Now, we are ready to define a robot type as follows (for more details, see \cref{sec:RobTyp}).


\begin{definition} [{\bf Robot Type}]
	A {\em robot type} is $\mathsf{RobTyp}=(\mathsf{CC},\mathsf{Alloc}_{\gr(\mathsf{CC})},\mathsf{NumOfCyc})$ such that:
	\begin{enumerate}
		\item $\mathsf{CC}\subseteq E(\overline{G})$.
		\item $\gr(\mathsf{CC})$ is connected, every vertex in $\gr(\mathsf{CC})$ has even degree and\\ $\vi\in V(\gr(\mathsf{CC}))$.
		\item $\mathsf{Alloc}_{\gr(\mathsf{CC})}$ is an allocation of $\{(u^*_j,\widehat{\mathsf{N}}_{\gr({\mathsf{CC})}}(u^*_j))~|~u^*_j\in V(\gr({\mathsf{CC})})\}$ to vertex types. 
		\item $\mathsf{NumOfCyc}=(N_2,N_3,N_5,N_6,\ldots,N_{2|\vc|})$, where $0\leq N_i\leq 2|\vc|^2$ for every $2\leq i\leq 2|\vc|$, $i\neq 4$.
	\end{enumerate}
\end{definition}

Similarly, we add to a cycle type with a cycle $C$ in $G^*$ an allocation of the multiset $\{\{v,v'\}~|~u^*\in V(C)$, $v$ and $v'$ are the vertices appears right before and right after $u^*\}$ to vertex types (given by a function $\mathsf{PaAlloc}_C$). Now, we are ready to define a cycle type as follows (for more details, see \cref{sec:CycTyp}).

\begin{definition} [{\bf Cycle Type}]
	Let $C\in \mathsf{Cyc}_{G^*}$, let $\mathsf{PaAlloc}_C$ be an allocation of $\{\{v,v'\}~|~u^*$ $\in V(C)$, $v$ and $v'$ are the vertices appears right before and right after $u^*\}$ to vertex types, and let $\mathsf{RobTyp}=(\mathsf{CC},\mathsf{Alloc}_{\gr(\mathsf{CC})},\mathsf{NumOfCyc})$ be a robot type. Then, $\mathsf{CycTyp}=(C,\mathsf{PaAlloc}_C,\mathsf{RobTyp})$ is a {\em cycle type} if $V(\gr(\mathsf{CC}))\cap V(C)\cap \vc\neq \emptyset$.
\end{definition}

We have the following notations.

For every $\mathsf{VerTyp}=(u^*,\U)\in \mathsf{VerTypS}$, every $\W=\{v,v'\}\in \U$ and $1\leq j\leq 2|\vc|$, $\mathsf{CycTypS}(\mathsf{VerTyp},\W,j)$ is the set of cycle types that assign $\W$ to $\mathsf{VerTyp}$ exactly $j$ times. For every $\mathsf{VerTyp}=(u^*,\U)\in \mathsf{VerTypS}$, every $\W\in \U$ and $1\leq j\leq 2^{|\vc|}+|\vc|^2$,
$\mathsf{RobTypS}(\mathsf{VerTyp},\W,j)$ is the set of robot types that assign $\W$ to $\mathsf{VerTyp}$ exactly $j$ times.
Finally, we have the following equation:

\smallskip\noindent{\bf Equation~\ref{ilp:3}: Assigning Enough Subsets for Each Vertex Type.} The equation is defined as follows.

\noindent3. For every $\mathsf{VerTyp}=(u^*,\U)\in \mathsf{VerTypS}$, and every $\W\in \U$, 

\noindent$\displaystyle{\sum_{j=1}^{2{|\vc|}}\sum_{{\mathsf{CycTyp}\in \mathsf{CycTypS}(\mathsf{VerTyp},\W,j)}}j\cdot x_\mathsf{CycTyp}+}$\\  $\displaystyle{\sum_{j=1}^{2^{|\vc|}+|\vc|^2}\sum_{{\mathsf{RobTyp}\in \mathsf{RobTypS}(\mathsf{VerTyp},\W,j)}}j\cdot x_\mathsf{RobTyp}\geq x_\mathsf{VerTyp}}$.

\subsubsection{The Correctness of The Reduction}
We denote the ILP instance associated with Equations~\ref{ilp:1}--\ref{ilp:6} by $\mathsf{Reduction}(G,\vi,k,B)$. Now, we give a proof sketch for the correctness of the reduction:

\begin{lemma}
	Let $G$ be a connected graph, let $\vi\in V(G)$ and let $k,B\in \mathbb{N}$. Then, $(G,\vi,k,B)$ is a yes-instance of \cg, if and only if $\mathsf{Reduction}(G,\vi,k,B)$ is a yes-instance of the {\sc Integer Linear Programming}.
\end{lemma}

\begin{proof}
	Let $x_z$, for every $z\in \mathsf{VerTypS}\cup \mathsf{RobTypS}\cup \mathsf{CycTypS}$, be values satisfying Equations~\ref{ilp:1}--\ref{ilp:6}. For every vertex type $\mathsf{VerTyp}=(u^*,\mathsf{NeiSubsets})$ and each $\W\in \U$, let $\mathsf{Alloc}(\mathsf{VerTyp},\W)$ be the set of every allocation of $\W$ to $\mathsf{VerTyp}$ by cycles or robots. We arbitrary allocate each element in $\mathsf{Alloc}(\mathsf{VerTyp},\W)$ to a vertex in $u^*$, such that every vertex $u\in u^*$ of type $\mathsf{VerTyp}$ gets at least one allocation. Due to Equation~\ref{ilp:3}, we ensure we can do that. Then, we replace every $u^*\in V(\gr(\mathsf{CC}_i))$ and every $u^*\in V(C)$ (for every $C\in {\cal C}_i$) by the $u\in u^*$ derived by the allocation. This ensures we covered every edge adjacent to a vertex in $V(G)\setminus \vc$. As seen in this overview, the other conditions of \cref{obs:equivsol} hold (the full proof is in \cref{sec:forDir}).
	
	For the reverse direction, let $\widehat{E}_1,\ldots,\widehat{E}_k$ be multisets satisfying the conditions of \cref{obs:equivsol}. For every $1\leq i\leq k$ let $(\mathsf{CC}_i,{\cal C}_i)$ be an $\widehat{E}_i$-valid pair. Then, we first derive the vertex type of each $u\in V(G)\setminus \vc$, according to its equivalence class in $\ind$, and the set of multisets derived from $((\mathsf{CC}_i,{\cal C}_i))_{1\leq i\leq k}$ (e.g. see \cref{fig:vertexType})(for more details, see \cref{def:VTD} and \cref{lem:VerDer}). Then, we derive the robot type $\mathsf{RobTyp}=(\mathsf{CC},\mathsf{Alloc}_{\gr(\mathsf{CC})},\mathsf{NumOfCyc})$ for each $i\in [k]$: (i) the skeleton $\mathsf{CC}$ is determined by $\mathsf{CC}_i$ (e.g., see \cref{fig:robotAllocation})), (ii) $\mathsf{NumOfCyc}$ is determined by the number of cycle of each length in ${\cal C}_i$ and (iii) the allocation of the multisets of $\gr(\mathsf{CC}_i)$ is determined by the vertex types of $u\in V(\gr(\mathsf{CC}_i))\cap (V(G)\setminus \vc)$ we have already computed (for more details, see \cref{def:RTD} and \cref{lem:RobDer}). Then, for every $i\in[k]$ and every $C'\in {\cal C}_i$, we determine the cycle type $\mathsf{CycTyp}=(C,\mathsf{PaAlloc}_C,\mathsf{RobTyp})$ of $C'$: (i) $C$ is determined by $C'$ (we replace each $u\in u^*\in \ind$ in $C'$ by $u^*$), (ii) $\mathsf{RobTyp}$ is the robot type of $i$ we have already computed, and (iii) $\mathsf{PaAlloc}_C$  is determined by the vertex types of $u\in V(C')\cap (V(G)\setminus \vc)$ we have already computed (for more details, see \cref{def:CTD} and \cref{lem:CycDer}). Then, for every $z\in \mathsf{VerTypS}\cup \mathsf{RobTypS}\cup \mathsf{CycTypS}$, we define $x_z$ to be the number of elements of type $z$. As seen in this overview, the values of the variables satisfy Equations~\ref{ilp:1}--\ref{ilp:6} (the full proof is in \cref{sec:revDir}).
\end{proof}

Observe that the number of variables is bounded by a function of $|\vc|$, so we will get an FPT runtime with respect to $\vn$. We analyze the runtime of the algorithm in \cref{sec:runtime}. 
Thus, we conclude the correction of \cref{th:fptVc}.

\subsection{Approximation Algorithm with Additive Error of $\OO(\vn)$}
Our algorithm is based on a greedy approach.
Recall that, our new goal (from Lemma~\ref{obs:equivsol}) is to find $k$ multisets $\widehat{E}_1,\ldots,\widehat{E}_k$ such that for every $1\leq i\leq k$, $\vi\in V(\gr(\widehat{E}_i))$, $\gr(\widehat{E}_i)$ is connected and each $u\in V(\gr(\widehat{E}_i))$ has even degree in $\gr(\widehat{E}_i)$. Now, assume that we have a vertex cover $\vc$ of $G$ such that $G[\vc]$ is connected and $\vi\in \vc$ (e.g., see the orange vertices in Figure~\ref{fig:approx1B}), and let $I=V\setminus \vc$. 
We first make the degree of every vertex in $I$ even in $G$, by duplicating an arbitrary edge for vertices having odd degree (e.g., see the green edges in Figure~\ref{fig:approx1C}). Observe that, after these operations, $G$ may be a multigraph (e.g., see the graph in Figure~\ref{fig:approx1C}).

We initialize $\widehat{E}_1,\ldots,\widehat{E}_k$ with $k$ empty sets. We partition the set of edges of $G$ with one endpoint in $I$ in the following manner. We choose the next multiset from $\widehat{E}_1,\ldots,\widehat{E}_k$ in a round-robin fashion and put a pair of edges, not considered so far, incident to some vertex $v \in I$, in the multiset (e.g., see the red, blue and green edges in Figure~\ref{fig:approx1D}).
%
%
This ensures that the degree of every vertex in $I$ is even in each multiset (e.g., see the Figures~\ref{fig:approx1E}--\ref{fig:approx1G}). 
Let $\widehat{E}'_1,\ldots,\widehat{E}'_k$ be multisets satisfying the conditions of Lemma~\ref{obs:equivsol}. Then, due to Condition~\ref{con:2EquiSol} of Lemma~\ref{obs:equivsol}, the degree of every vertex is even in every multiset $\widehat{E}'_i$. 
Thus, the total number of edges (with repetition) incident to any vertex in $\gr(\widehat{E}'_1 \cup \ldots \cup \widehat{E}'_k)$ is even.
Therefore, there must be at least one additional repetition for at least one edge of every vertex with odd degree in $G$. So, adding an additional edge to each vertex with odd degree is ``must'' and it does not ``exceed'' the optimal budget. Then, we partition the edges with both endpoints in $\vc$, in a balanced fashion, as follows. We choose an edge, not considered so far, and add it to a multiset with minimum size. 

Observe that, after this step, we have that: i) every edge of the input graph belongs to at least one of the multisets $\widehat{E}_i$, ii) the degree of each vertex of $I$ in each multiset is even, and iii) we have not exceeded the optimal budget (e.g., see the Figures~\ref{fig:approx2A}--\ref{fig:approx2D}). We still need to ensure that i) $\gr(\widehat{E}_i)$ is connected, for every $i \in [k]$, ii) the degree of each vertex of $\vc$ in each multiset is even, and iii) $\vi\in V(\gr(\widehat{E}_i))$ for every $i\in [k]$.
Next, we add a spanning tree of $G[\vc]$ to each of the $\widehat{E}_i$, in order to make $\gr(\widehat{E}_i)$ connected and to ensure that $\vi\in V(\gr(\widehat{E}_i))$ (e.g., see the Figures~\ref{fig:approx2F}--\ref{fig:approx2H}). Lastly, we add at most $|\vc|$ edges, with both endpoints in $\vc$, to every $\widehat{E}_i$ in order to make the degree of each $u\in \vc$ even in each of the multiset (e.g., see the Figures~\ref{fig:approx2I}--\ref{fig:approx2K}). Observe that the multisets $\widehat{E}_1, \ldots, \widehat{E}_k$ satisfy the conditions of Lemma~\ref{obs:equivsol}.
Moreover, we added at most $\OO(|\vc|)$ additional edges to each $\widehat{E}_i$, comparing to an optimal solution.

\begin{figure}[t]
	\centering
	\begin{subfigure}[t]{.33\textwidth}
		\centering
		\includegraphics[page=9, width=\textwidth]{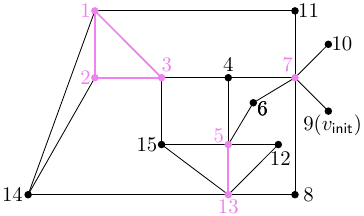}
		\caption{}
		\label{fig:approx1A}
	\end{subfigure}\hfil
	\begin{subfigure}[t]{.33\textwidth}
		\centering
		\includegraphics[page=10, width=\textwidth]{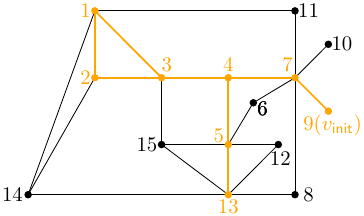}
		\caption{}
		\label{fig:approx1B}
	\end{subfigure}\hfil
	\begin{subfigure}[t]{.33\textwidth}
		\centering
		\includegraphics[page=11, width=\textwidth]{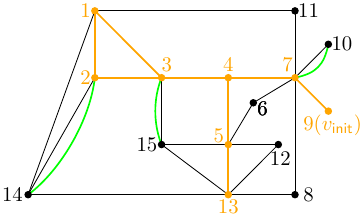}
		\caption{}
		\label{fig:approx1C}
	\end{subfigure}\hfil
	\begin{subfigure}[t]{.33\textwidth}
		\centering
		\includegraphics[page=12, width=\textwidth]{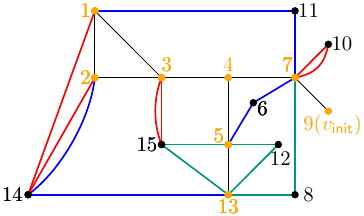}
		\caption{}
		\label{fig:approx1D}
	\end{subfigure}\hfil
	\begin{subfigure}[t]{.33\textwidth}
		\centering
		\includegraphics[page=13, width=\textwidth]{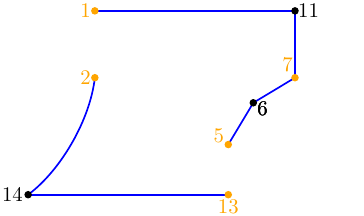}
		\caption{}
		\label{fig:approx1E}
	\end{subfigure}\hfil
	\begin{subfigure}[t]{.33\textwidth}
		\centering
		\includegraphics[page=14, width=\textwidth]{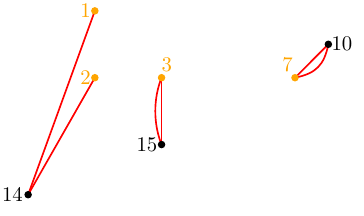}
		\caption{}
		\label{fig:approx1F}
	\end{subfigure}\hfil
	\begin{subfigure}[t]{.33\textwidth}
		\centering
		\includegraphics[page=15, width=\textwidth]{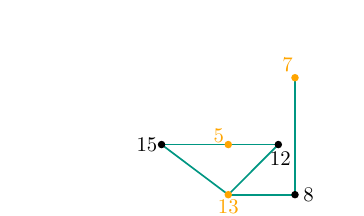}
		\caption{}
		\label{fig:approx1G}
	\end{subfigure}
	\caption{Illustration of the execution of Lines~\ref{alg1:L2}--\ref{alg1:L20} of Algorithm~\ref{alg:AppAlg} on the instance shown in Figure~\ref{fig:approx1A}. The orange vertices in Figures~\ref{fig:approx1B}--\ref{fig:approx1G} are the vertex cover vertices.
		(a) Running example for Algorithm~\ref{alg:AppAlg}: A graph $G$, a vertex cover $\vc$ (drawn in violet) of $G$, the vertex $\vi = 9$, and $k = 3$. (b) The connected vertex cover $\vc'$ (drawn in orange) of $G$ obtained from $\vc$ after executing Lines \ref{alg1:L2}--\ref{alg1:L3} of Algorithm~\ref{alg:AppAlg}. (c) The green edges are added to make the degree of independent set vertices even (Lines~\ref{alg1:L4}--\ref{alg1:L9} of Algorithm~\ref{alg:AppAlg}). (d) Balanced partition of edges incident to independent set vertices to the three robots, shown by red, blue and green edges (Lines~\ref{alg1:L10}--\ref{alg1:L20} of Algorithm~\ref{alg:AppAlg}). (e-g) The graphs induced by the multiset corresponding to each of the three robots after Line~\ref{alg1:L20} of Algorithm~\ref{alg:AppAlg}.}
	\label{fig:approx1}
\end{figure}

\begin{figure}
	\centering
	\begin{subfigure}[t]{.33\textwidth}
		\centering
		\includegraphics[page=16, width=\textwidth]{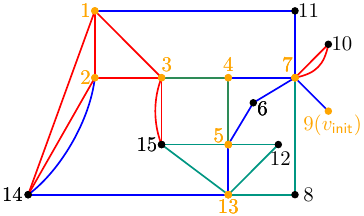}
		\caption{}
		\label{fig:approx2A}
	\end{subfigure}\hfil
	\begin{subfigure}[t]{.33\textwidth}
		\centering
		\includegraphics[page=17, width=\textwidth]{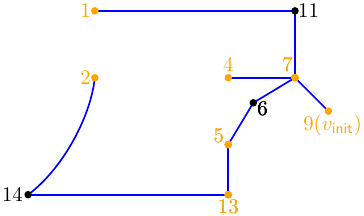}
		\caption{}
		\label{fig:approx2B}
	\end{subfigure}\hfil
	\begin{subfigure}[t]{.33\textwidth}
		\centering
		\includegraphics[page=18, width=\textwidth]{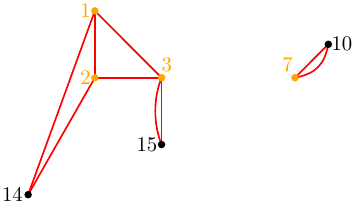}
		\caption{}
		\label{fig:approx2C}
	\end{subfigure}\hfil
	\begin{subfigure}[t]{.33\textwidth}
		\centering
		\includegraphics[page=19, width=\textwidth]{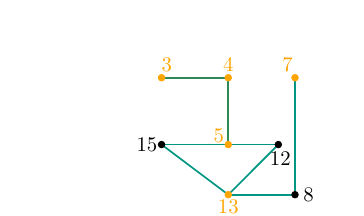}
		\caption{}
		\label{fig:approx2D}
	\end{subfigure}\hfil
	\begin{subfigure}[t]{.33\textwidth}
		\centering
		\includegraphics[page=20, width=\textwidth]{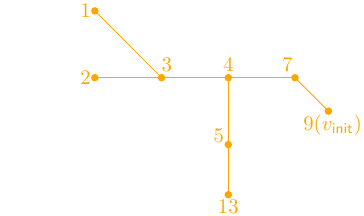}
		\caption{}
		\label{fig:approx2E}
	\end{subfigure}\hfil
	\begin{subfigure}[t]{.33\textwidth}
		\centering
		\includegraphics[page=21, width=\textwidth]{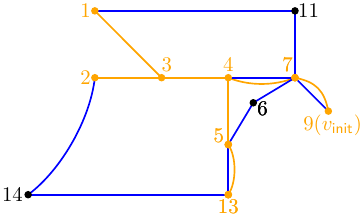}
		\caption{}
		\label{fig:approx2F}
	\end{subfigure}\hfil
	\begin{subfigure}[t]{.33\textwidth}
		\centering
		\includegraphics[page=22, width=\textwidth]{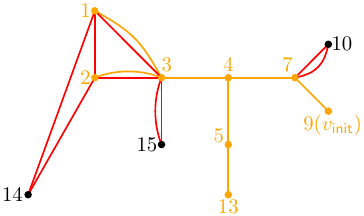}
		\caption{}
		\label{fig:approx2G}
	\end{subfigure}\hfil
	\begin{subfigure}[t]{.33\textwidth}
		\centering
		\includegraphics[page=23, width=\textwidth]{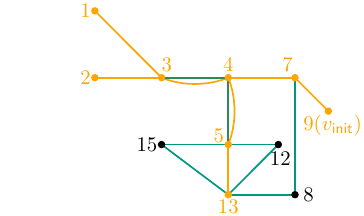}
		\caption{}
		\label{fig:approx2H}
	\end{subfigure}\hfil
	\begin{subfigure}[t]{.33\textwidth}
		\centering
		\includegraphics[page=24, width=\textwidth]{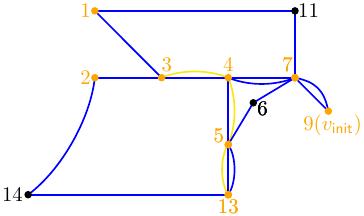}
		\caption{}
		\label{fig:approx2I}
	\end{subfigure}\hfil
	\begin{subfigure}[t]{.33\textwidth}
		\centering
		\includegraphics[page=25, width=\textwidth]{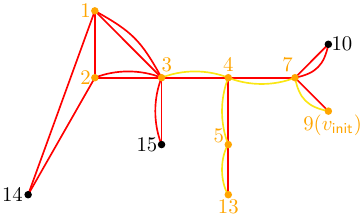}
		\caption{}
		\label{fig:approx2J}
	\end{subfigure}\hfil
	\begin{subfigure}[t]{.33\textwidth}
		\centering
		\includegraphics[page=26, width=\textwidth]{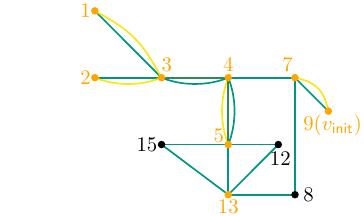}
		\caption{}
		\label{fig:approx2K}
	\end{subfigure}
	
	\caption{Illustration of the execution of Lines~\ref{alg1:L24}--\ref{alg1:L31} of Algorithm~\ref{alg:AppAlg} on the instance shown in Figure~\ref{fig:approx1A}. The orange vertices in Figures~\ref{fig:approx2B}--\ref{fig:approx2K} are the vertex cover vertices. (a) Balanced partition of all the edges, also including the ones incident only to vertex cover vertices to the three robots, shown by red, blue and green edges (after Line~\ref{alg1:L24} of Algorithm~\ref{alg:AppAlg}). (b-d) The graphs induced by the multisets corresponding to each of the three robots after Line~\ref{alg1:L24} of Algorithm~\ref{alg:AppAlg}. (e) A spanning tree $T$ of the graph induced by the vertex cover $\vc'$ (Line~\ref{alg1:L25} of Algorithm~\ref{alg:AppAlg}). (f-h) The graphs induced by the multisets corresponding to each of the three robots after adding the spanning tree $T$ (Lines~\ref{alg1:L25}--\ref{alg1:L28} of Algorithm~\ref{alg:AppAlg}). (i-k) The graphs induced by the multisets corresponding to each of the three robots after making the degree of each vertex in $\vc'$ even (Lines~\ref{alg1:L29}--\ref{alg1:L31} of Algorithm~\ref{alg:AppAlg}). The yellow edges correspond to the edges added by Algorithm~\ref{alg:MakeVCInAEvenDeg}.}
	\label{fig:approx2}
\end{figure}


\section{FPT Algorithm with Respect to Vertex Cover}\label{sec:FPTAlg}

In this section, we present an FPT algorithm with respect to the vertex cover number of the input graph $G$:

\fptTheorem*

Our algorithm is based on a reduction to the {\sc Integer Linear Programming} (ILP) problem. 

Recall that by Lemma~\ref{obs:equivsol}, an instance $(G,\vi,k,B)$ of \cg is a yes-instance if and only if there exist $k$ multisets $\widehat{E}_1,\ldots,\widehat{E}_k$ with elements from $E(G)$ such that:
\begin{enumerate}
	\item For every $1\leq i\leq k$, $\vi\in V(\gr(\widehat{E}_i))$.
	\item For every $1\leq i\leq k$, $\gr(\widehat{E}_i)$ is connected, and every vertex in it has even degree.
	\item $E(G)\subseteq \widehat{E}_1\cup \ldots \cup\widehat{E}_k$.
	\item$ \mathsf{max}\{|\widehat{E}_1|,\ldots, |\widehat{E}_k|\}\leq B$.
\end{enumerate}

Now, we show that it is enough to look at ``simpler'' multisets satisfying the conditions of Lemma~\ref{obs:equivsol}. By picking a solution that minimizes $\sum^k_{i=1}|\widehat{E}_i|$, it holds that for every $1\leq i\leq k$ and $\{u,v\}\in \widehat{E}_i$, $\{u,v\}$ appears at most twice in $\widehat{E}_i$:

\begin{observation}\label{obs:equivso}
	Let $G$ be a connected graph, let $\vi\in V(G)$ and let $k,B\in \mathbb{N}$. Assume that there exist $\widehat{E}_1,\ldots,\widehat{E}_k$ such that the conditions of Lemma~\ref{obs:equivsol} hold. Then, there exist $\widehat{E}'_1,\ldots,\widehat{E}'_k$ such that the conditions of Lemma~\ref{obs:equivsol} hold and for every $1\leq i\leq k$, each $\{u,v\}\in \widehat{E}'_i$ appears at most twice in $\widehat{E}'_i$.
\end{observation}

\begin{proof}
Let $\widehat{E}_1,\ldots,\widehat{E}_k$ be multisets satisfying the conditions of Lemma~\ref{obs:equivsol}. Assume that there exists $1\leq i\leq k$ and $\{u,v\}\in \widehat{E}_i$ that appears more than twice in $\widehat{E}_i$. Let $\widehat{E}'_i= \widehat{E}_i\setminus \{\{u,v\},\{u,v\}\}$. Observe that since every vertex in $\gr(\widehat{E}_i)$ has even degree, then every vertex in $\gr(\widehat{E}'_i)$ has even degree. Thus, it is easy to see that $\widehat{E}_1,\ldots,\widehat{E}_{i-1},\widehat{E}'_{i},\widehat{E}_{i+1},\ldots \widehat{E}_{k}$ are also multisets satisfying the conditions of Lemma~\ref{obs:equivsol}. 
\end{proof}

We aim to construct linear equations that verify the conditions in Lemma~\ref{obs:equivsol}.  
We begin with some definitions that we will use later, when we present our reduction. 

\subsection{Encoding $\widehat{E}_i$ by a Valid Pair}\label{sec:ValPa}

Given $\widehat{E}_1,\ldots,\widehat{E}_k$ such that conditions of Observation~\ref{obs:equivso} hold, we show how to ``encode'' each $\widehat{E}_i$ by a different structure, which will be useful later. For this purpose, we present some definitions. 
Let $G$ be a connected graph and let $\vc$ be a vertex cover of $G$, given as input. We define an equivalence relation on $V(G)\setminus \vc$ based on the sets of neighbors of the vertices in $V(G)\setminus \vc$:



\begin{definition} [{\bf Equivalence Relation for the Independent Set}]\label{def:ER}
Let $G$ be a connected graph and let $\vc$ be a vertex cover of $G$. Let $\indd=V(G)\setminus \vc$. For every $u,v\in \indd$, $u$ is {\em equal} to $v$ if $\mathsf{N}(u)=\mathsf{N}(v)$. We denote the set of equivalence classes induced by this equivalence relation by $\mathsf{EQ}_{G,\vc}$.
\end{definition}

When $G$ and $\vc$ are clear from context, we refer to $\mathsf{EQ}_{G,\vc}$ as $\mathsf{EQ}$.

Next, we define the {\em equivalence graph of $G$ and $\vc$}, denoted by $G^*$. It contains every vertex from $\vc$, and the edges having both endpoints in $\vc$. In addition, every equivalence class $u^*\in \ind$ is represented by a single vertex adjacent to the neighbors of some $u\in u^*$ in $G$ (which belong to $\vc$), e.g. see \cref{fig:fptB}.  

\begin{definition} [{\bf Equivalence Graph $G^*$}]\label{def:VCG}
	Let $G$ be a connected graph and let $\vc$ be a vertex cover of $G$. The {\em equivalence graph of $G$ and $\vc$} is $G^*=\gr(E^*)$ where $E^*=\{\{u^*,v\}~|~u^*\in \ind, \exists u \in V(G)$ s.t. $  \{u,v\}\in E(G) \wedge u\in u^*\}\cup \{\{u,v\}\in E(G)~|~u,v\in \vc\}$.
\end{definition}

To construct $k$ multisets $\widehat{E}_1,\ldots,\widehat{E}_k$ that satisfy the conditions of Lemma~\ref{obs:equivsol}, we first deal with Condition~\ref{obs:equivsol2}: given $1\leq i\leq k$, we need to ensure that $\gr(\widehat{E}_i)$ is connected, and every vertex in it has even degree. For this purpose, we ``encode'' $\widehat{E}_i$, for every $1\leq i\leq k$, as follows. First, we construct a ``small'' subset $\mathsf{CC}_i\subseteq \widehat{E}_i$ such that: (i) $\gr(\mathsf{CC}_i)$ is connected; (ii) every vertex in $\gr(\mathsf{CC}_i)$ has even degree; and (iii) $V(\gr(\mathsf{CC}_i))\cap \vc=\gr(\widehat{E}_i)\cap \vc$. In words, Condition (iii) states that every vertex from $\gr(\widehat{E}_i)$ that belongs to $\vc$ also belongs to $\gr(\mathsf{CC}_i)$. Observe that every vertex will have even degree in $\gr(\widehat{E}_i\setminus\mathsf{CC}_i)$, and every cycle in $\gr(\widehat{E}_i\setminus\mathsf{CC}_i)$ will contain at least one vertex from $\vc$. We will later use these properties to ``encode'' the edges from $\widehat{E}_i\setminus\mathsf{CC}_i$.

 Now, we turn to show how to construct $\mathsf{CC}_i$. For this purpose, we define the graph $\overline{G}$. This graph is similar to $G^*$, but has $\mathsf{NumVer}(u^*)=\mathsf{min}\{|u^*|,2^{|\mathsf{N}_{G^*}({u^*})|}+|\vc|^2\}$ vertices for each $u^*\in \ind$. In addition, each edge in $\overline{G}$ appears exactly twice in $E(\overline{G})$, e.g. see \cref{fig:fptC}. We will see later that we can construct $\mathsf{CC}_i$ as a submultigraph of $\overline{G}$.

\begin{definition} [{\bf The Multigraph $\overline{G}$}]\label{def:GbarGraph}
		Let $G$ be a connected graph and let $\vc$ be a vertex cover of $G$. For every $u^*\in \ind$, let $\mathsf{NumVer}(u^*)=\mathsf{min}\{|u^*|,2^{|\mathsf{N}_{G^*}({u^*})|}+|\vc|^2\}$. Then, $\overline{G}=\gr(\overline{E})$ where $\overline{E}=\{\{u^*_i,v\},\{u^*_i,v\}~|~u^*\in \ind, \exists u \in V(G)$ s.t. $\{u,v\}\in E(G) \wedge u\in u^*,1\leq i\leq \mathsf{NumVer}(u^*)\}\cup \{\{u,v\},\{u,v\}~|~\{u,v\}\in E(G), u,v\in \vc\}$.
	\end{definition}

Next, for a submultigraph $H$ of $G$, we define the term {\em $\overline{G}$-submultigraph} (e.g. see \cref{fig:fpt1B}):

\begin{definition} [{\bf $\overline{G}$-Submultigraph}]\label{def:Gbarsubmultigraph}
	Let $G$ be a connected graph and let $\vc$ be a vertex cover of $G$. Let $\widehat{E}$ be a multiset with elements from $E(G)$ such that every $\{u,v\}\in \widehat{E}$ appears at most twice in $\widehat{E}$. A submultigraph $H$ of $\gr(\widehat{E})$ is a {\em $\overline{G}$-submultigraph} if for every $u^*\in \ind$, $|V(H)\cap u^*|\leq \mathsf{NumVer}(u^*)$.
\end{definition}

Given a $\overline{G}$-submultigraph $H$, we define the operation $\overline{G}(H)$, which returns a submultigraph $\overline{H}$ of $\overline{G}$ isomorphic to $H$ with isomorphism $\alpha:V(H)\rightarrow V(\overline{H})$, where every vertex $u\in u^*\cap V(H)$ is mapped (by $\alpha$) to a vertex in $\overline{G}$ that belongs to the same equivalence class as $u$ (e.g., see \cref{fig:gBarOp}):

\begin{definition}[{\bf The Operation $\overline{G}$}]\label{def:GbarsubmultigraphAc}
	Let $G$ be a connected graph, let $\vc$ be a vertex cover of $G$ and let $H$ be a $\overline{G}$-submultigraph. Then, $\overline{G}(H)$ is a submultigraph $\overline{H}$ of $\overline{G}$ for which there exists a isomorphism $\alpha:V(H)\rightarrow V(\overline{H})$ satisfying the following conditions:

\begin{enumerate}
	\item For every $u\in \vc$, $\alpha(u)=u$.\label{def:GbarsubmultigraphAc1}
	\item For every $u\in u^*\in \ind$, there exists $j\in[\mathsf{NumVer}(u^*)]$ such that $\alpha(u)=u^*_j$.\label{def:GbarsubmultigraphAc2}
\end{enumerate}
\end{definition}

\begin{figure}
	\centering
	\includegraphics[page=33, width=0.5\textwidth]{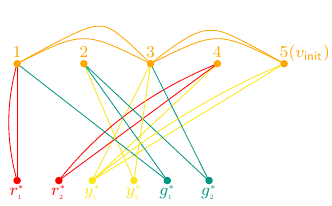}
	\caption{The graph $\overline{G}(H)$ for the graph $H$ in \cref{fig:fpt1B}. The isomorphism $\alpha$ for independent vertices is defined as follows: $\alpha(r_1) = r^*_1$, $\alpha(r_2) = r^*_2$, $\alpha(y_5) = y^*_1$, $\alpha(y_6) = y^*_2$, $\alpha(g_{11}) = g^*_1$, and $\alpha(g_{12}) = g^*_2$.}
	\label{fig:gBarOp}
\end{figure}

Observe that since $H$ is a $\overline{G}$-submultigraph, $\mathsf{NumVer}(u^*)$ and the number of appearance of each edge are large enough to ensure that $\overline{G}(H)$ is well defined (although not uniquely defined).

Assuming we have computed $\mathsf{CC}_i$ (which we will do soon in Lemma~\ref{lem:4cycandCC}), we now explain the intuition for how to encode the edges in $\widehat{E}_i\setminus\mathsf{CC}_i$. Recall that every vertex has even degree in $\gr(\widehat{E}_i\setminus\mathsf{CC}_i)$. Therefore, if $\gr(\widehat{E}_i\setminus\mathsf{CC}_i)$ is not empty, there exists a cycle $C$ in $\gr(\widehat{E}_i\setminus\mathsf{CC}_i)$, and, in particular, a simple one, as is implied by the following observation:

\begin{observation}\label{obs:existCycEven}
	Let $G$ be a non-empty multigraph. Assume that every vertex in $G$ has even degree in it. Then, there exists a simple cycle in $G$.
\end{observation}

 Now, observe that every vertex also has even degree in $\gr(\widehat{E}_i\setminus(\mathsf{CC}_i\cup E(C)))$, as implied by the following observation:

\begin{observation}\label{obs:deletCycEven}
	Let $G$ be a multigraph. Assume that every vertex in $G$ has even degree in it. Let $C$ be a cycle in $G$. Then, every vertex in $\gr(E(G)\setminus E(C))$ has even degree in it.
\end{observation}

So, the same reasoning can be reapplied. Therefore, we can encode $\widehat{E}_i\setminus\mathsf{CC}_i$ by a multiset of cycles. In particular, in the following lemma, we show that we can encode $\widehat{E}_i\setminus\mathsf{CC}_i$ by a multiset of cycles ${\cal C}$ such that all the cycles except for ``few'' of them are of length $4$; in addition, the cycles from ${\cal C}$ that are of length other than $4$ are simple.

\begin{lemma}\label{lem:4cyc}
Let $G$ be a multigraph such that each $\{u,v\}\in E(G)$ appears at most twice in $E(G)$, and let $\vc$ be a vertex cover of $G$. Assume that every vertex in $G$ has even degree. Then, there exists a multiset of cycles ${\cal C}$ in $G$ such that:
\begin{enumerate}
	\item At most $2|\vc|^2$ cycles in ${\cal C}$ have length other than $4$.
	\item The cycles in ${\cal C}$ that have length other than $4$ are simple.
	\item $\bigcup_{C\in {\cal C}}E(C)=E(G)$ (being two multisets).
\end{enumerate}
\end{lemma}

\begin{proof}
	First, assume that there are more than $2|\vc|^2$ edges in $E(G)$. Let $\indd=V(G)\setminus \vc$. Every $u\in V(G)$ has even degree, so let $2d_u$ be the degree of $u$ in $G$, where $d_u\in \mathbb{N}$. For every $u\in \indd$, let $\widehat{A}_u=\{\{\{u,v_i\},\{u,v_i'\}\}_{i=1}^{d_u}\}$ be a partition of the edges incident to $u$ in $G$ into multisets of two elements, and let $\widehat{A}=\bigcup_{u\in \indd} \widehat{A}_u$. Observe that there are more than $2|\vc|^2$ edges in $G$, and the number of edges in $G$ with both endpoints in $\vc$ is bounded by $2{|\vc| \choose 2}\leq |\vc|^2$ (${|\vc| \choose 2}$ different edges, each appears at most twice). So, there are more than $|\vc|^2$ edges with one endpoint in $\indd$. This implies that we have more than $\frac{1}{2}|\vc|^2$ multisets in $\widehat{A}$. Now, every multiset in $\widehat{A}$ consists of two edges, each edge is incident to exactly one vertex from $\vc$. Notice that the number of different options to choose two vertices from $\vc$ is bounded by $\frac{1}{2}|\vc|^2$. Therefore, since there are more than $\frac{1}{2}|\vc|^2$ multisets in $\widehat{A}$, there exist two multisets $\{\{u,v_i\},\{u,v_i'\}\},\{\{\tilde{u},v_j\},\{\tilde{u},v_j'\}\}$ such that $v_i=v_j$ and $v'_i=v'_j$. So, $C=(u,v_i,\tilde{u},v'_i,u)$ is a cycle of length $4$ in $G$. From Observation~\ref{obs:deletCycEven}, the degree of every vertex in $G$ remains even after deleting the edges of $C$ from $G$. Then, after deleting the edges of $C$ from $G$ (and inserting it into ${\cal C}$), we can use the same argument to find yet another cycle. 
	
	 If $G$ has $2|\vc|^2$ edges or less (and it is not empty), then since the degree of every vertex is even, from Observation~\ref{obs:existCycEven}, we can find a cycle---and, in particular, a simple one---in the multigraph and delete it (and inserting it into ${\cal C}$). 
	 
	The process ends when there are no edges in the multigraph. Clearly, by our construction, the conditions of the lemma hold. This ends the proof. 
\end{proof}

We have the following observation, which will be useful later:

\begin{observation}\label{obs:cycleVc}
	Let $G$ be a connected multigraph, let $\vc$ be a vertex cover of $G$ and let $C$ be a cycle in $G$. Then, $V(C)\cap \vc\neq\emptyset$.
\end{observation}

Let $\widehat{E}_1,\ldots,\widehat{E}_k$ satisfy the conditions of Lemma~\ref{obs:equivsol}. We encode each $\widehat{E}_i$ as a pair $(\mathsf{CC}_i,{\cal C}_i)$ where $\mathsf{CC}_i\cup \bigcup_{C\in {\cal C}_i}E(C)=\widehat{E}_i$, and: (i) $\gr(\mathsf{CC}_i)$ is a $\overline{G}$-submultigraph; (ii) $\gr(\mathsf{CC}_i)$ is connected, $\vi \in V(\gr(\mathsf{CC}_i))$ and every vertex in $\gr(\mathsf{CC}_i)$ has even degree; and (iii) ${\cal C}_i$ is a multiset of cycles satisfying the conditions of Lemma~\ref{lem:4cyc}. We call such a pair an {\em $\widehat{E}_i$-valid pair}:

\begin{definition}[{\bf Valid Pair}]\label{def:ValPa}
	Let $G$ be a connected graph, let $\vi\in V(G)$ and let $\vc$ be a vertex cover of $G$. Let $\widehat{E}$ be a multiset with elements from $E(G)$ such that $\gr(\widehat{E})$ is connected, $\vi\in V(\gr(\widehat{E}))$, and every vertex in $\gr(\widehat{E})$ has even degree. Let $\mathsf{CC}\subseteq \widehat{E}$, and let ${\cal C}$ be a multiset of cycles in $\gr(\widehat{E})$. Then, $(\mathsf{CC},{\cal C})$ is an {\em $\widehat{E}$-valid pair} if the following conditions are satisfied:
	\begin{enumerate}
		\item $\gr(\mathsf{CC})$ is a $\overline{G}$-submultigraph. \label{lemCC:con1}
		\item $\gr(\mathsf{CC})$ is connected, $\vi \in V(\gr(\mathsf{CC}))$ and every vertex in $\gr(\mathsf{CC})$ has even degree.  \label{lemCC:con2}
		\item $V(\gr(\mathsf{CC}))\cap \vc=V(\gr(\widehat{E}))\cap\vc$. \label{lemCC:con3}
		\item At most $2|\vc|^2$ cycles in ${\cal C}$ have length other than $4$. \label{lemCC:con4}
		\item The cycles in ${\cal C}$ of length other than $4$ are simple. \label{lemCC:con5}
		\item $\mathsf{CC}\cup \bigcup_{C\in {\cal C}}E(C)=\widehat{E}$. \label{lemCC:con6}
	\end{enumerate}
\end{definition}

Next, we invoke Lemma~\ref{lem:4cyc} in order to prove the following. Let $\widehat{E}_1,\ldots,\widehat{E}_k$ satisfy the conditions of Lemma~\ref{obs:equivsol}, and for every $1\leq i\leq k$ and $\{u,v\}\in \widehat{E}_i$, $\{u,v\}$ appears at most twice in $\widehat{E}_i$ (see Observation~\ref{obs:equivso}). Then, there exists $(\mathsf{CC}_i,{\cal C}_i)$ such that $(\mathsf{CC}_i,{\cal C}_i)$ is an $\widehat{E}_i$-valid pair:

\begin{lemma}\label{lem:4cycandCC}
	Let $G$ be a connected graph, let $\vi\in V(G)$ and let $\vc$ be a vertex cover of $G$. Let $\widehat{E}$ be a multiset with elements from $E(G)$. Assume that $\gr(\widehat{E})$ is connected, $\vi\in V(\gr(\widehat{E}))$, every vertex in $\gr(\widehat{E})$ has even degree and every $\{u,v\}\in \widehat{E}$ appears at most twice in $\widehat{E}$. Then, there exist $\mathsf{CC}\subseteq \widehat{E}$ and a multiset of cycles ${\cal C}$ in $\gr(\widehat{E})$ such that $(\mathsf{CC},{\cal C})$ is an $\widehat{E}$-valid pair.
\end{lemma}


\begin{proof}
First, we construct $\mathsf{CC}$. For this purpose, we obtain a multigraph $H$ from $\gr(\widehat{E})$ as follows: Starting with $H=\gr(\widehat{E})$. While there exist $u,u'\in u^*\cap V(H) $ such that $\mathsf{N}_{H}(u)=\mathsf{N}_{H}(u')$, delete $u'$. Observe that, since $\gr(\widehat{E})$ is connected, then so is $H$. Also, observe that for every $u\in u^*\in \ind$, the number of options for $\mathsf{N}_{H}(u)$ is bounded by $2^{|\mathsf{N}_{G^*}(u^*)|}$. So, for every $u^*\in \ind$, $|u^*\cap V(H)|\leq 2^{|\mathsf{N}_{G^*}(u^*)|}$. Notice that for every $u^*\in \ind$, we might have more than one vertex from $u^*$ in $H$ since vertices that have the same neighbors in $G$, might have different neighbors in $H$. In addition, each $u\in u^*\in \ind$ has even degree in $H$ (being the same as its degree in $\gr(\widehat{E})$), e.g. see \cref{fig:fpt1B}.

 Now, let $v\in \vc$ that has odd degree in $H$ (if one exists). Notice that since the number of vertices with odd degree is even in every multigraph, there exists $v'\in \vc\cap V(H)$ such that $v'\neq v$ and $v$ has odd degree in $H$. We build a path $P$ from $v$ to such a $v'$ in $G$ using edges in $\widehat{E}\setminus E(H)$ as follows. Since $v$ has odd degree in $H$ and even degree in $\gr(\widehat{E})$, there exists $\{u,v\}\in \widehat{E}\setminus E(H)$. We add $\{u,v\}$ to $P$. Now, if $u$ has odd degree in $H$, then $u\in \vc$ and we finish; otherwise, $u$ has even degree in $H$, so it has odd degree in $\gr(E(H)\cup E(P))$, and so there exists $\{u,u'\}\in \widehat{E}\setminus (E(H)\cup E(P))$. This process is finite, and ends when $P$ reaches $v'\in \vc$ with odd degree in $H$. Now, since there exists a path $P$ from $v$ to $v'$ in $\gr(\widehat{E})$ such that $E(P)\cap E(H)=\emptyset$, there exists a simple path $P'$ from $v$ to $v'$ in $\gr(\widehat{E})$, such that $E(P')\cap E(H)=\emptyset$. Observe that $|P'|\leq 2|\vc|$, and there are at most $|\vc|$ vertices in $P'$ from $V(G)\setminus \vc$. We add the edges of $P'$ to $H$, that is, $H\gets \gr(E(H)\cup E(P'))$ (e.g. see the dashed paths in \cref{fig:fpt1C}), and continue with this process until every vertex in $H$ has even degree. Observe that, by the end of this process, we have added at most $|\vc|^2$ vertices from $V(G)\setminus \vc$ to $H$. So, in particular, for every $u^*\in \ind$, we have added at most $|\vc|^2$ vertices from $(V(G)\setminus \vc)\cap u^*$ to $H$.
 
 Overall, for every $u^*\in \ind$, $|u^*\cap V(H)|\leq 2^{|\mathsf{N}_{G^*}(u^*)|}+|\vc|^2$. In addition, notice that $H$ is a submultigraph of $\gr(\widehat{E})$, so each $\{u,v\}\in E(H)$ appears at most twice in $E(H)$. Moreover, observe that $V(H)\cap \vc= V(\gr(\widehat{E}))\cap \vc$, and since we assume that $\vi\in \vc$, then $\vi \in V(H)$. We define $\mathsf{CC}=E(H)$, so Conditions \ref{lemCC:con1}--\ref{lemCC:con3} of Definition~\ref{def:ValPa} are satisfied. Now, observe that every vertex in $\gr(\widehat{E}\setminus \mathsf{CC})$ has even degree, and $\vc\cap V(\gr(\widehat{E}\setminus \mathsf{CC}))$ is a vertex cover of $\gr(\widehat{E}\setminus \mathsf{CC})$. Therefore, from Lemma~\ref{lem:4cyc}, there exists a multiset of cycles ${\cal C}$ in $\gr(\widehat{E}\setminus \mathsf{CC})$ such that Conditions~\ref{lemCC:con4}--\ref{lemCC:con6} of Definition~\ref{def:ValPa} are satisfied. So, $(\mathsf{CC},{\cal C})$ is an $\widehat{E}$-valid pair. This completes the proof.
\end{proof}


\subsection{Vertex Type}\label{sec:VerTyp}

Now, to construct the equations for the ILP instance, we present the variables. First, we have a variable for each {\em vertex type}. We begin by showing how we derive the vertex type of a vertex from a solution, and then we define formally the term vertex type. For this purpose, we introduce some definitions.

An intuition for a vertex type is as follows. Let $\widehat{E}_1,\ldots,\widehat{E}_k$ such that the conditions of Lemma~\ref{obs:equivsol} hold, and for every $1\leq i\leq k$ let $(\mathsf{CC}_i,{\cal C}_i)$ be an $\widehat{E}_i$-valid pair. For every $u\in u^*\in \ind$, we derive multisets with elements from $\mathsf{N}_{G^*}(u^*)$ as follows. First, for every $1\leq i\leq k$, we derive the multiset of neighbors of $u$ in $\gr(\mathsf{CC}_i)$, that is $\widehat{\mathsf{N}}_{\gr(\mathsf{CC}_i)}(u)$. Second, for every $1\leq i\leq k$ and $C\in {\cal C}_i$ such that $u\in V(C)$, we derive the multiset $\{v,v'\}$, where $v$ and $v'$ are the vertices appear right before and right after $u$ in $C$, respectively. Now, recall that in a solution, every edge is covered by at least one robot. So, given a vertex $u\in u^*\in \ind$ and an edge $\{u,v\}\in E(G)$, there exists $1\leq i\leq k$ such that $\{u,v\}$ is covered by the $i$-th robot, that is, $\{u,v\}\in \widehat{E}_i$. Since $\mathsf{CC}_i\cup \bigcup_{C\in {\cal C}_i}E(C)=\widehat{E}_i$, either $\{u,v\}\in \mathsf{CC}_i$ or $\{u,v\}\in E(C)$, for some $C\in {\cal C}_i$. Thus, $u$ is in at least one multiset we derived. In the vertex types, we consider all the possible options for such multisets.


In the next definition, for a given cycle $C$, we define the {\em pairs of edges} of $C$, denoted $\mathsf{EdgePairs}(C)$: for each $u\in V(G)\setminus \vc$, we have the pair $(\{u,v\},\{u,v'\})$, where $v$ and $v'$ are the vertices appear right before and right after $u$ in $C$. Later, we will derive for each such a pair, the multiset $\{v,v'\}$ and ``allocate'' it to $u$.

Let $G$ be a connected graph and let $\vc$ be a vertex cover of $G$. We denote the set of cycles of length at most $\mathsf{max}\{4,2|\vc|\}$ in $G$ by $\mathsf{Cyc}_G$. Let $C=(v_0,\ldots,v_\ell=v_0)\in \mathsf{Cyc}_G$. We denote by $\ind(C)$ the cycle in $G^*$ obtained from $C$ by replacing each $v_i\in V(C)\cap (V(G)\setminus \vc)$ by $u^*\in \ind$, where $v_i\in u^*$, for every $1\leq i\leq \ell$.

\begin{definition} [{\bf Pairs of a Cycle}]\label{def:POC}
	Let $G$ be a connected graph, let $\vc$ be a vertex cover of $G$ and let $C=(v_0,\ldots,v_\ell=v_0)\in \mathsf{Cyc}_{G}\cup \mathsf{Cyc}_{G^*}$. Then, the {\em pairs of edges} of $C$ is the multiset $\mathsf{EdgePairs}(C)=\{\{\{v_{i-1},v_{i}\},\{v_{i},v_{i+1}\}\}~|~1\leq i\leq \ell-1, v_i\in \mathsf{IND}\cup \ind\}$. 
\end{definition}


 Recall that, given a multigraph $G$ and $u\in V(G)$, the multiset of neighbors of $u$ in $G$ is denoted by $\widehat{\mathsf{N}}_G(u)=\{v\in V~|~\{u,v\}\in E(G)\}$ (with repetition). 

\begin{definition} [{\bf Deriving Vertex Types From a Solution}]\label{def:VTD}
	Let $G$ be a connected graph, let $\vi\in V(G)$, let $k,B\in \mathbb{N}$, let $\vc$ be a vertex cover of $G$ and let $\widehat{E}_1,\ldots,\widehat{E}_k$ such that conditions of Lemma~\ref{obs:equivsol} hold. For every $1\leq i\leq k$ let $(\mathsf{CC}_i,{\cal C}_i)$ be an $\widehat{E}_i$-valid pair. In addition, for every $1\leq i\leq k$ and $u\in u^*\in \ind$ such that $u\in V(\gr(\widehat{E}_i))$, let $\mathsf{NeiPairs}_{{\cal C}_i}(u)=\{\{v,v'\}~|~C\in {\cal C}_i, \{\{u,v\},\{u,v'\}\}\in \mathsf{EdgePairs}(C)\}$ be a set. For every $u\in u^*\in \ind$, let $\mathsf{NeiSubsets}(u)=\bigcup_{1\leq i\leq k}(\{\widehat{\mathsf{N}}_{\gr(\mathsf{CC}_i)}(u)\}\cup \mathsf{NeiPairs}_{{\cal C}_i}(u))$ be a set. Then, for every $u\in u^*\in\ind$, $\mathsf{DerVerTyp}(\{(\mathsf{CC}_i,{\cal C}_i)\}_{1\leq i\leq k},u)=(u^*,\mathsf{NeiSubsets}(u))$.
\end{definition}



Whenever $\{(\mathsf{CC}_i,{\cal C}_i)\}_{1\leq i\leq k}$ is clear from context, we refer to $\mathsf{DerVerTyp}(\{(\mathsf{CC}_i,{\cal C}_i)\}_{1\leq i\leq k},$ $u)$ as $\mathsf{DerVerTyp}(u)$. Observe that every element in $\mathsf{NeiPairs}_{{\cal C}_i}(u)$ is a multiset of two vertices (which might be the same vertex).

Now, observe that each vertex in every multiset derived in Lemma~\ref{def:VTD} appears at most twice: each edge appears at most twice in $\mathsf{CC}_i$, and every multiset derived from a cycle has exactly two elements in it. In addition, since the degree of each vertex is even in $\gr(\mathsf{CC}_i)$, every multiset derived in Lemma~\ref{def:VTD} has an even number of elements. So, we will consider only multisets with these restriction.

For a set $A$ we define the multiset $A\times \mathsf{2}=\{a,a~|~a\in A\}$. That is, each element in $A$ appears exactly twice in $A\times \mathsf{2}$.

Now, we define the term vertex type:

\begin{definition} [{\bf Vertex Type}]\label{def:VT}
	Let $G$ be a connected graph and let $\vc$ be a vertex cover of $G$. Let $u^*\in \ind$ and let $\mathsf{NeiSubsets}\subseteq 2^{\mathsf{N}_{G^*}(u^*)\times \mathsf{2}}$. Then, $\mathsf{VerTyp}=(u^*,\mathsf{NeiSubsets})$ is a {\em vertex type} if
 for every $\mathsf{NeiSub}\in \mathsf{NeiSubsets}$, $|\mathsf{NeiSub}|$ is even, and $\mathsf{N}_{G^*}(u^*)\subseteq \bigcup\mathsf{NeiSubsets}$.
\end{definition}

We denote the set of vertex types by $\mathsf{VerTypS}$.

In the following lemma we show the ``correctness'' of Definition~\ref{def:VTD}, that is, for every $u\in \mathsf{IND}$, $\mathsf{DerVerTyp}(u)$ is indeed a vertex type. 



\begin{lemma}\label{lem:VerDer}
Let $G$ be a connected graph, let $\vi\in V(G)$, let $k,B\in \mathbb{N}$ and let $\vc$ be a vertex cover of $G$. Let $\widehat{E}_1,\ldots,\widehat{E}_k$ such that conditions of Lemma~\ref{obs:equivsol} hold and for every $1\leq i\leq k$ let $(\mathsf{CC}_i,{\cal C}_i)$ be an $\widehat{E}_i$-valid pair. Then, for every $u\in \mathsf{IND}$, $\mathsf{DerVerTyp}(u)$ is a vertex type.
\end{lemma}

\begin{proof}
We show that, for every $u\in u^*\in \ind$, $\mathsf{DerVerTyp}(u)=(u^*,\mathsf{NeiSubsets}(u))$ is a vertex type. Let $u\in u^*\in \ind$. First, we show that for every $\mathsf{NeiSub}\in \mathsf{NeiSubsets}(u)$, $\mathsf{NeiSub}\in 2^{\mathsf{N}_{G^*}(u^*)\times \mathsf{2}}$ and $|\mathsf{NeiSub}|$ is even. For every $\mathsf{NeiSub}\in \mathsf{NeiSubsets}(u)$ there exists $1\leq i\leq k$ such that $\mathsf{NeiSub}\in \{\widehat{\mathsf{N}}_{\gr(\mathsf{CC}_i)}(u)\}\cup \mathsf{NeiPairs}_{{\cal C}_i}(u)$. 
\begin{itemize}
	\item If $\mathsf{NeiSub}=\widehat{\mathsf{N}}_{\gr(\mathsf{CC}_i)}(u)$, then $\mathsf{NeiSub}=\{v\in V(\gr(\mathsf{CC}_i))~|~\{u,v\}\in \mathsf{CC}_i\}$. Since $\mathsf{CC}_i$ is a $\overline{G}$-submultigraph, each $\{u,v\}\in \mathsf{CC}_i$ appears at most twice in $\mathsf{CC}_i$. Then, each $v\in \mathsf{NeiSub}$ appears at most twice in $\mathsf{NeiSub}$, and so $\mathsf{NeiSub}\in 2^{\mathsf{N}_{G^*}(u^*)\times \mathsf{2}}$. In addition, the degree of $u$ in $\gr(\mathsf{CC}_i)$ is even, so $|\mathsf{NeiSub}|$ is even.
	\item  Otherwise, $\mathsf{NeiSub}\in \mathsf{NeiPairs}_{{\cal C}_i}(u)$, and every multiset in $\mathsf{NeiPairs}_{{\cal C}_i}(u)$ has exactly two vertices, so $\mathsf{NeiSub}\in 2^{\mathsf{N}_{G^*}(u^*)\times \mathsf{2}}$ and $|\mathsf{NeiSub}|$ is even.
\end{itemize} Now, we show that $\mathsf{N}_{G^*}(u^*)\subseteq \bigcup \mathsf{NeiSubsets}(u)$. Let $v\in \mathsf{N}_{G^*}(u^*)$. From Condition~\ref{obs:equivsol3} of Lemma~\ref{obs:equivsol}, $E(G)\subseteq \widehat{E}_1\cup \ldots \cup\widehat{E}_k$. Then, there exists $1\leq i\leq k$ such that $\{u,v\}\in \widehat{E}_i$. Since  $(\mathsf{CC}_i,{\cal C}_i)$ is an $\widehat{E}_i$-valid pair, $\widehat{E}_i=\mathsf{CC}_i\cup \bigcup_{C\in {\cal C}_i}E(C)$. So, if $\{u,v\}\in \mathsf{CC}_i$, then $v\in \widehat{\mathsf{N}}_{\gr(\mathsf{CC}_i)}(u)$; otherwise, there exists $C\in {\cal C}_i$ such that $\{u,v\}\in E(C)$. Therefore, there exists $v'\in V(G)$ such that $\{\{u,v\},\{u,v'\}\}\in \mathsf{EdgePairs}(C)$, so $\{v,v'\}\in \mathsf{NeiPairs}_{{\cal C}_i}(u) $, thus $v\in \bigcup \mathsf{NeiSubsets}(u)$. So, for every $u\in u^*\in \ind$, $\mathsf{DerVerTyp}(u)=(u^*,\mathsf{NeiSubsets}(u))$ is a vertex type. 
\end{proof}

\subsection{Robot Type}\label{sec:RobTyp}

Now, we continue to introduce the variables we need for the ILP instance. We have a variable for each {\em robot type}. First, we show how we derive a robot type for each robot from a solution, and then we will present the definition of an abstract type. An intuition for a robot type is as follows. In Definition~\ref{def:VTD}, where we derive vertex types, we saw how we derive $\mathsf{NeiSubsets}(u)$ for each $u\in V(G)\setminus \vc$. Some of the multisets in $\mathsf{NeiSubsets}(u)$ are derived from $\gr(\mathsf{CC}_i)$, where $(\mathsf{CC}_i,{\cal C}_i)$ is an $\widehat{E}_i$-valid pair associated with each robot $i$. Now, we look on the ``puzzle'' from the perspective of the robots. For each $u\in V(\gr(\mathsf{CC}_i)) \setminus \vc$, the multiset $\widehat{\mathsf{N}}_{\gr(\mathsf{CC}_i)}(u)$ is ``allocated'' for the vertex type of $u$.
For an induced submultigraph  $\overline{H}$ of $\overline{G}$, we present two notations: {\em the multiset $\mathsf{NeiOfInd}(\overline{H})$}, and an {\em allocation} for this multiset. When $\overline{H}=\overline{G}(\gr(\mathsf{CC}_i))$, $\mathsf{NeiOfInd}(\overline{H})$ is the set of pairs $(u^*_i,\W)$ we will later allocate to vertex types.



	
	\begin{definition} [{\bf $\mathsf{NeiOfInd}(\overline{H})$}]\label{def:SOP}
		Let $G$ be a connected graph and let $\vc$ be a vertex cover of $G$. Let $\overline{H}$ be a submultigraph of $\overline{G}$. Then, $\mathsf{NeiOfInd}(\overline{H})=\{(u^*_i,\widehat{\mathsf{N}}_{\overline{H}}(u^*_i))~|~u^*_i\in V(\overline{H})\}$ as a multiset.
	\end{definition}
	
	A {\em vertex allocation} of $\mathsf{NeiOfInd}(\overline{H})$ is a function that assigns a vertex type to each $(u^*_i,\W)\in \mathsf{NeiOfInd}(\overline{H})$. We allocate $(u^*_i,\W)$ to a vertex type that ``expects'' to get $\W$, that is, a vertex type $(u^*,\U)$ where $\W\in \U$.
	
		\begin{definition} [{\bf Vertex Allocation of $\mathsf{NeiOfInd}(\overline{H})$}]\label{def:SOPP}
		Let $G$ be a connected graph and let $\vc$ be a vertex cover of $G$. Let $\overline{H}$ be a submultigraph of $\overline{G}$. A {\em vertex allocation} of $\mathsf{NeiOfInd}(\overline{H})$ is a function $\mathsf{Alloc}_{\overline{H}}:\mathsf{NeiOfInd}(\overline{H})\rightarrow \mathsf{VerTypeS}$ such that for every $(u^*_i,\W)\in \mathsf{NeiOfInd}(\overline{H})$, $\mathsf{Alloc}_{\overline{H}}(u^*_i,\W)=(u^*,\U)$ where\\ $\W \in \U$.
	\end{definition}
	


In addition to the vertex allocation, the type of each robot $i$ is associated with a vector of non-negative integers $\mathsf{NumOfCyc}=(N_{i,2},N_{i,3},N_{i,5},N_{i,6},\ldots,N_{i,2|\vc|})$: for every $2\leq j\leq 2|\vc|$, $j\neq 4$, $N_{i,j}$ is the number of cycles of length exactly $j$ in ${\cal C}_i$, where $(\mathsf{CC}_i,{\cal C}_i)$ be an $\widehat{E}_i$-valid pair.

\begin{definition} [{\bf Deriving Robot Types From a Solution}]\label{def:RTD}
	Let $G$ be a connected graph, let $\vi\in V(G)$, let $k,B\in \mathbb{N}$, let $\vc$ be a vertex cover of $G$ and let $\widehat{E}_1,\ldots,\widehat{E}_k$ such that the conditions of Lemma~\ref{obs:equivsol} hold. For every $1\leq i\leq k$, let $(\mathsf{CC}_i,{\cal C}_i)$ be an $\widehat{E}_i$-valid pair. For every $1\leq i\leq k$, let $G'_i=\overline{G}(\gr(\mathsf{CC}_i))$ with an isomorphism $\alpha_i:V(G_i')\rightarrow V(\gr(\mathsf{CC}_i))$, and let $\mathsf{CC}'_i=E(G_i')$. For every $1\leq i\leq k$ and $(u^*_j,\W)\in \mathsf{NeiOfInd}(G'_i)$, let $\mathsf{Alloc}_{G_i'}((u^*_j,\W))=\mathsf{DerVerTyp}(\alpha_i(u^*_j))$. For every $1\leq i\leq k$ and $2\leq j\leq 2|\vc|$, $j\neq 4$, let $N_{i,j}$ be the number of cycles of size $j$ in ${\cal C}_i$, and for every $1\leq i\leq k$ let $\mathsf{NumOfCyc}_i=(N_{i,2},N_{i,3},N_{i,5},N_{i,6},\ldots,N_{i,2|\vc|})$. Then, for every $1\leq i\leq k$, let $\mathsf{DerRobTyp}(\{(\mathsf{CC}_j,{\cal C}_j)\}_{1\leq j\leq k},i)=(\mathsf{CC}'_i,\mathsf{Alloc}_{G_i'},\mathsf{NumOfCyc}_i)$.
\end{definition}

Whenever $\{(\mathsf{CC}_i,{\cal C}_i)\}_{1\leq i\leq k}$ is clear from context, we refer to $\mathsf{DerRobTyp}(\{(\mathsf{CC}_i,{\cal C}_i)\}_{1\leq i\leq k},$ $i)$ as $\mathsf{DerRobTyp}(i)$.


Now, we define the term robot type. As mentioned, a robot type is first associated with $\mathsf{CC}\subseteq E(\overline{G})$ and a vertex allocation of $\mathsf{NeiOfInd}(\gr(\mathsf{CC}))$. We demand that $\gr(\mathsf{CC})$ is connected, every vertex in $\gr(\mathsf{CC})$ has even degree and $\vi\in V(\gr(\mathsf{CC}))$, similarly to Conditions~\ref{lemCC:con2} and~\ref{lemCC:con3} of Definition~\ref{def:ValPa}. This way we ensure that the multiset $\widehat{E}$ we will build for the a robot makes $\gr(\widehat{E})$ connected: we will later associate with a robot only cycles $C$ such that $V(C)\cap V(\gr(\mathsf{CC}))\neq \emptyset$. In addition, we also ensure that every vertex in $\gr(\widehat{E})$ will have even degree in $\gr(\widehat{E})$: each such vertex has even degree in $\gr(\mathsf{CC})$, and we will add only cycles to this graph, so the degree of each vertex will remain even. Therefore, we ensure that the ``local'' properties of each robot, given by Conditions~\ref{obs:equivsol1} and~\ref{obs:equivsol2} of Lemma~\ref{obs:equivsol}, are preserved. The vector of non-negative integers $\mathsf{NumOfCyc}=(N_2,N_3,N_5,N_6,\ldots,N_{2|\vc|})$ determines how many cycles of each length (other than $4$) we will associate with the robot. Observe that, due to Condition~\ref{lemCC:con4} of Definition~\ref{def:ValPa}, each $N_j$ is bounded by $2|\vc|^2$.

\begin{definition} [{\bf Robot Type}]\label{def:RT}
	Let $G$ be a connected graph, let $\vi\in V(G)$ and let $\vc$ be a vertex cover of $G$. Then, $\mathsf{RobTyp}=(\mathsf{CC},\mathsf{Alloc}_{\gr(\mathsf{CC})},\mathsf{NumOfCyc})$ is a {\em robot type} if the following conditions are satisfied:
	\begin{enumerate}
		\item $\mathsf{CC}\subseteq E(\overline{G})$.\label{def:RT1}
		\item $\gr(\mathsf{CC})$ is connected, every vertex in $\gr(\mathsf{CC})$ has even degree and\\ $\vi\in V(\gr(\mathsf{CC}))$.\label{def:RT2}
		\item $\mathsf{Alloc}_{\gr(\mathsf{CC})}$ is a vertex allocation of $\mathsf{NeiOfInd}(\gr(\mathsf{CC}))$.\label{def:RT3}
		\item $\mathsf{NumOfCyc}=(N_2,N_3,N_5,N_6,\ldots,N_{2|\vc|})$, where $0\leq N_i\leq 2|\vc|^2$ for every $2\leq i\leq 2|\vc|$, $i\neq 4$.\label{def:RT4}
	\end{enumerate}
\end{definition}

We denote the set of robot types by $\mathsf{RobTypS}$.

In the following lemma we show the ``correctness'' of Definition~\ref{def:RTD}, that is, for every $1\leq i \leq k$, $\mathsf{DerRobTyp}(i)$ is indeed a robot type.

\begin{lemma}\label{lem:RobDer}
	Let $G$ be a connected graph, let $\vi\in V(G)$, let $k,B\in \mathbb{N}$ and let $\vc$ be a vertex cover of $G$. Let $\widehat{E}_1,\ldots,\widehat{E}_k$ such that conditions of Lemma~\ref{obs:equivsol} hold, and for every $1\leq i\leq k$, let $(\mathsf{CC}_i,{\cal C}_i)$ be an $\widehat{E}_i$-valid pair. Then, for every $1\leq i\leq k$, $\mathsf{DerVerTyp}(i)$ is a robot type.
\end{lemma}

\begin{proof}
Let $1\leq i\leq k$. We show that $\mathsf{DerVerTyp}(i)$ is a robot type by proving that the conditions of Definition~\ref{def:RT} hold. 
Since $(\mathsf{CC}_i,{\cal C}_i)$ is an $\widehat{E}_i$-valid pair, from Condition~\ref{lemCC:con1} of Definition~\ref{def:ValPa},  $\gr(\mathsf{CC}_i)$ is a $\overline{G}$-submultigraph. So, $G'_i=\overline{G}(\gr(\mathsf{CC}_i))$ is a submultigraph of $\overline{G}$, with an isomorphism $\alpha_i:V(\overline{G})\rightarrow V(\gr(\mathsf{CC}_i))$ satisfies the conditions of Definition~\ref{def:GbarsubmultigraphAc}. Thus, $\mathsf{CC}'_i=E(G_i')\subseteq E(\overline{G})$, and therefore, Condition~\ref{def:RT1} of Definition~\ref{def:RT} holds.

 From Condition~\ref{lemCC:con2} of Definition~\ref{def:ValPa}, $\gr(\mathsf{CC}_i)$ is connected, $\vi \in V(\gr(\mathsf{CC}_i))$ and every vertex in $\gr(\mathsf{CC}_i)$ has even degree. So, $\gr(\mathsf{CC}'_i)$ is connected, every vertex in $\gr(\mathsf{CC}'_i)$ has even degree and $\vi\in V(\gr(\mathsf{CC}'_i))$. Thus, Condition~\ref{def:RT2} of Definition~\ref{def:RT} holds. 
 
 Now, let $(u^*_j,\mathsf{NeiSub})\in \mathsf{NeiOfInd}(G'_i)$. Then, $\mathsf{Alloc}_{G_i'}((u^*_j,\mathsf{NeiSub}))=\mathsf{DerVerTyp}(\alpha_i(u^*_j))$, so $\mathsf{Alloc}_{G_i'}((u^*_j,\mathsf{NeiSub}))=(z^*,\mathsf{NeiSubsets}(\alpha_i(u^*_j)))$, where $\mathsf{NeiSubsets}(\alpha_i(u^*_j))=\bigcup_{1\leq t\leq k}$ $(\{\widehat{\mathsf{N}}_{\gr(\mathsf{CC}_t)}(\alpha_i(u^*_j))\}\cup \mathsf{NeiPairs}_{{\cal C}_t}(\alpha_i(u^*_j)))$(see Definition~\ref{def:VTD}). Observe that $\widehat{\mathsf{N}}_{\gr(\mathsf{CC}_i)}$ $(\alpha_i(u^*_j))=\mathsf{NeiSub}$, thus $\mathsf{NeiSub}\in \mathsf{NeiSubsets}(\alpha_i(u^*_j))$. In addition, from Condition~\ref{def:GbarsubmultigraphAc2} of Definition~\ref{def:GbarsubmultigraphAc}, $\alpha_i(u^*_j)\in u^*$, so $ z^*= u^*$. Also, from Lemma~\ref{lem:VerDer}, $\mathsf{DerVerTyp}(\alpha_i(u^*_j))\in \mathsf{VerTypS}$. Therefore, $\mathsf{Alloc}_{G_i'}$ is a vertex allocation of $\mathsf{NeiOfInd}(G_i')=\mathsf{NeiOfInd}(\gr(\mathsf{CC}_i'))$, so Condition~\ref{def:RT3} of Definition~\ref{def:RT} holds.
 
 Now, since $(\mathsf{CC}_i,{\cal C}_i)$ is an $\widehat{E}_i$-valid pair, from Condition~\ref{lemCC:con1} of Definition~\ref{def:ValPa}, at most $2|\vc|^2$ cycles in ${\cal C}_i$ have length other than $4$. Therefore, for every $2\leq j\leq 2|\vc|$, $0\leq N_j\leq 2|\vc|^2$. So, Condition~\ref{def:RT4} of Definition~\ref{def:RT} holds.
 
 We proved that all of the conditions of Definition~\ref{def:RT} hold, so  $\mathsf{DerVerTyp}(i)$ is a robot type. This ends the proof.
\end{proof}                                                                                                  

\subsection{Cycle Type}\label{sec:CycTyp}
Lastly, we have a variable for each {\em cycle type}. For every $1\leq i\leq k$, let $(\mathsf{CC}_i,{\cal C}_i)$ be an $\widehat{E}_i$-valid pair. First, we show how we derive a cycle type for each $C\in {\cal C}_i$, for every $1\leq i\leq k$, and then we will present the definition.
 An intuition for a cycle type is as follows. In Lemma~\ref{def:VTD}, where we derive vertex types, we saw how we derive $\mathsf{NeiSubsets}(u)$ for each $u\in V(G)\setminus \vc$. Some of the multisets in $\mathsf{NeiSubsets}(u)$ are derived from cycles in $\mathsf{CC}_i$, for some $1\leq i\leq k$. Now, we look on the ``puzzle'' from the perspective of the cycles. For each $\{\{u,v\},\{u,v'\}\}\in \mathsf{EdgePairs}(C)$, the multiset $\{v,v'\}$ is ``allocated'' for the vertex type of $u$.



Similarly to the definition of a vertex allocation of $\mathsf{NeiOfInd}(G')$ (Definition~\ref{def:SOPP}), we have the following definition, for allocation of $\mathsf{EdgePairs}(C)$:


\begin{definition} [{\bf Vertex Allocation of $\mathsf{EdgePairs}(C)$}]\label{def:PPOC}
	Let $G$ be a connected graph, let $\vc$ be a vertex cover of $G$ and let $C\in \mathsf{Cyc}_{G}\cup \mathsf{Cyc}_{G^*}$. A {\em vertex allocation} of $\mathsf{EdgePairs}(C)$ is a function $\mathsf{PaAlloc}_C:\mathsf{EdgePairs}(C)\rightarrow \mathsf{VerType}$, such that for every $\{\{v_{i-1},v_i\},\{v_i,v_{i+1}\}\}\in \mathsf{EdgePairs}(C)$, $\{v_{i-1},v_{i+1}\}\in \mathsf{NeiSubsets}$ and $v_i\in u^*$ or $v_i=u^*$, where $\mathsf{PaAlloc}_C(\{\{v_{i-1},v_i\},$ $\{v_i,v_{i+1}\}\})=(u^*,\mathsf{NeiSubsets})$.
\end{definition}

Let $C=(v_0,\ldots,v_\ell=v_0)\in \mathsf{Cyc}_G$. Recall that we denote by $\ind(C)$ the cycle in $G^*$ obtained from $C$ by replacing each $v_i\in V(C)\cap (V(G)\setminus \vc)$ by $u^*\in \ind$, where $v_i\in u^*$, for every $1\leq i\leq \ell$. Similarly, let $\mathsf{PaAlloc}_C$ be a vertex allocation of $\mathsf{EdgePairs}(C)$. We denote by $\ind(\mathsf{PaAlloc}_C)$ the function $\ind(\mathsf{PaAlloc}_C):\mathsf{EdgePairs}(\ind(C))\rightarrow \mathsf{VerType}$ defined as follows: for every $\{\{v_{i-1},u^*\},\{u^*,v_{i+1}\}\}\in \mathsf{EdgePairs}(\ind(C))$, $\ind(\mathsf{PaAlloc}_C)(\{\{v_{i-1},u^*\},\{u^*,v_{i+1}\}$ $\})=\mathsf{Alloc}_{C}(\{\{v_{i-1},v_i\},\{v_i,v_{i+1}\}\})$. Observe that $\ind(\mathsf{PaAlloc}_C)$ is a vertex allocation of $\mathsf{EdgePairs}(\ind(C))$:

\begin{observation}\label{obs:indAll}
Let $C\in \mathsf{Cyc}_G$ and let $\mathsf{PaAlloc}_C$ be a vertex allocation of $\mathsf{EdgePairs}(C)$. Then, $\ind(\mathsf{PaAlloc}_C)$ is a vertex allocation of $\mathsf{EdgePairs}(\ind(C))$.
\end{observation}

\begin{definition} [{\bf Deriving Cycle Types From a Solution}]\label{def:CTD}
	Let $G$ be a connected graph, let $\vi\in V(G)$, let $k,B\in \mathbb{N}$, let $\vc$ be a vertex cover of $G$ and let $\widehat{E}_1,\ldots,\widehat{E}_k$ such that conditions of Lemma~\ref{obs:equivsol} hold. For every $1\leq i\leq k$, let $(\mathsf{CC}_i,{\cal C}_i)$ be an $\widehat{E}_i$-valid pair. For every $1\leq i\leq k$, $C\in {\cal C}_i$ and $\{\{u,v\},\{u,v'\}\}\in\mathsf{EdgePairs}(C)$, let $\mathsf{PaAlloc}_{i,C}(\{\{v,u\},\{u,v'\}\})=\mathsf{DerVerTyp}(u)$. Then, for every $1\leq i\leq k$ and $C\in {\cal C}_i$, let 
	
	\noindent $\mathsf{DerCycTyp}(\{(\mathsf{CC}_j,{\cal C}_j)\}_{1\leq j\leq k},i,C)=(\ind(C),\ind(\mathsf{PaAlloc}_{i,C}),\mathsf{DerRobTyp}(i))$. 
\end{definition}

Whenever $\{(\mathsf{CC}_i,{\cal C}_i)\}_{1\leq i\leq k}$ is clear from context, we refer to $\mathsf{DerCycTyp}(\{(\mathsf{CC}_j,{\cal C}_j)\}_{1\leq j\leq k},$ $i,C)$ as $\mathsf{DerCycTyp}(i,C)$.


Now, we define the term cycle type. In addition to the vertex allocation of $\mathsf{EdgePairs}(C)$, we have the robot type $\mathsf{RobTyp}=(\mathsf{CC},\mathsf{Alloc}_{\gr(\mathsf{CC})},\mathsf{NumOfCyc})$ associated with the cycle type. In order to maintain the connectivity of $\gr(\widehat{E})$, we demand that $V(\gr(\mathsf{CC}))\cap V(C)\cap \vc\neq \emptyset$ (see the discussion before Definition~\ref{def:RT}).





\begin{definition} [{\bf Cycle Type}]\label{def:CT}
	Let $G$ be a connected graph, let $\vi\in V(G)$ and let $\vc$ be a vertex cover of $G$. Let $C\in \mathsf{Cyc}_{G^*}$, let $\mathsf{PaAlloc}_C$ be a vertex allocation of $\mathsf{EdgePairs}(C)$ and let $\mathsf{RobTyp}=(\mathsf{CC},\mathsf{Alloc}_{\gr(\mathsf{CC})},\mathsf{NumOfCyc})$ be a robot type. Then, $\mathsf{CycTyp}=(C,\mathsf{PaAlloc}_C,\mathsf{RobTyp})$ is a {\em cycle type} if $V(\gr(\mathsf{CC}))\cap V(C)\cap \vc\neq \emptyset$.
\end{definition}


We denote the set of cycle types by $\mathsf{CycTypS}$.

In the following lemma we show the ``correctness'' of Definition~\ref{def:CTD}, that is, for every $1\leq i\leq k$, $C\in {\cal C}_i$, $\mathsf{DerCycTyp}(i,C)$ is indeed a cycle type.

\begin{lemma}\label{lem:CycDer}
	Let $G$ be a connected graph, let $\vi\in V(G)$, let $k,B\in \mathbb{N}$ and let $\vc$ be a vertex cover of $G$. Let $\widehat{E}_1,\ldots,\widehat{E}_k$ such that conditions of Lemma~\ref{obs:equivsol} hold and for every $1\leq i\leq k$ let $(\mathsf{CC}_i,{\cal C}_i)$ be an $\widehat{E}_i$-valid pair. Then, for every $1\leq i\leq k$, $C\in {\cal C}_i$, $\mathsf{DerCycTyp}(i,C)$ is a cycle type. 
\end{lemma}

\begin{proof}
Let $1\leq i\leq k$ and $C\in {\cal C}_i$. We show that $\mathsf{DerCycTyp}(i,C)$ is a cycle type. First, from Condition~\ref{lemCC:con4} of Definition~\ref{def:ValPa}, every cycle in ${\cal C}_i$ of length other than $4$ is simple. So, the length of $C$ is at most $\mathsf{max}\{4,2|\vc|\}$, and thus, $\ind(C)\in \mathsf{Cyc}_{G^*}$. 

Second, we show that $\mathsf{PaAlloc}_{i,C}$, defined in Definition~\ref{def:CTD}, is a vertex allocation of $\mathsf{EdgePairs}(C)$. Let $\{\{u,v\},\{u,v'\}\}\in\mathsf{EdgePairs}(C)$. Then, $\mathsf{PaAlloc}_{i,C}(\{\{v,u\},\{u,v'\}\})=\mathsf{DerVerTyp}(u)$. First, from Lemma~\ref{lem:VerDer}, $\mathsf{DerVerTyp}(u)$ is a vertex type. Now, by Definition~\ref{def:VTD}, $\mathsf{DerVerTyp}(u)=(u^*,\mathsf{NeiSubsets}(u))$, where $u\in u^*$, $\mathsf{NeiSubsets}(u)=\bigcup_{1\leq j\leq k}($ $\{\widehat{\mathsf{N}}_{\gr(\mathsf{CC}_j)}(u)\}\cup \mathsf{NeiPairs}_{{\cal C}_j}(u))$ and  $\mathsf{NeiPairs}_{{\cal C}_j}(u)=\{\{v,v'\}~|~C\in {\cal C}_j, \{\{u,v\},\{u,v'\}\}\in \mathsf{EdgePairs}(C)\}$. So, since $\{\{u,v\},\{u,v'\}\}\in\mathsf{EdgePairs}(C)$, then $\{v,v'\}\in \mathsf{NeiPairs}_{{\cal C}_i}(u)\subseteq \mathsf{NeiSubsets}(u)$. Therefore, $\mathsf{PaAlloc}_{i,C}$ is a vertex allocation of $\mathsf{EdgePairs}(C)$, and from Observation~\ref{obs:indAll}, $\ind(\mathsf{PaAlloc}_{i,C})$ is a vertex allocation of $\mathsf{EdgePairs}(\ind(C))$.

Now, since $(\mathsf{CC}_i,{\cal C}_i)$ is an $\widehat{E}_i$-valid pair, from Condition~\ref{lemCC:con3} of Definition~\ref{def:ValPa}, $V(\gr(\mathsf{CC}_i))$ $\cap \vc=V(\gr(\widehat{E}_i))\cap\vc$. From Definition~\ref{def:RTD}, $\mathsf{DerRobTyp}(i)=((\mathsf{CC}'_i),\mathsf{Alloc}_{G_i'},$\\$\mathsf{NumOfCyc}_i)$, where $\mathsf{CC}'_i=E(G_i')$ and $G'_i=\overline{G}(\gr(\mathsf{CC}_i))$. Now, Condition~\ref{def:GbarsubmultigraphAc1} of Definition~\ref{def:GbarsubmultigraphAc} implies that $V(G'_i)\cap \vc=V(\gr(\widehat{E}_i))\cap\vc$, so $V(\gr(\mathsf{CC}'))\cap \vc=V(\gr(\widehat{E}_i))$ $\cap\vc$. From Observation~\ref{obs:cycleVc}, $V(\gr(\widehat{E}_i))\cap\vc\cap V(C)\neq \emptyset$, thus $V(\gr(\mathsf{CC}'))\cap\vc\cap V(C)\neq \emptyset$. In addition, from Lemma~\ref{lem:RobDer}, $\mathsf{DerRobTyp}(i)$ is a robot type. Overall, $\mathsf{DerCycTyp}(i,C)$ is a cycle type.
\end{proof}


\subsection{The Instance $\mathsf{Reduction}(G,\vi,k,B)$ of the ILP Problem}

Now, we are ready to present our reduction to the ILP problem. We have the following variables:
\begin{itemize}
	\item For every $\mathsf{VerTyp}\in \mathsf{VerTypS}$, we have the variable $x_\mathsf{VerTyp}$.
	\item For every $\mathsf{RobTyp}\in \mathsf{RobTypS}$, we have the variable $x_\mathsf{RobTyp}$.
	\item For every $\mathsf{CycTyp}\in \mathsf{CycTypS}$, we have the variable $x_\mathsf{CycTyp}$.
\end{itemize}

Each variable stands for the number of elements of each type.
We give intuition for each of the following equations: 

\smallskip\noindent{\bf Equation~\ref{ilp:1}: Robot Type for Each Robot.} In this equation, we make sure that the total sum of robot types is exactly $k$, that is, there is exactly one robot type for each robot:

1. $\ds{\sum_{\mathsf{RobTyp}\in \mathsf{RobTypS}}x_\mathsf{RobTyp}=k}$.

\smallskip\noindent{\bf Equation~\ref{ilp:2}: Vertex Type for Each Vertex.} For every $u^*\in \ind$, we denote the set of $\mathsf{VerTyp}=(u^*,\U)\in \mathsf{VerTypS}$ by $\mathsf{VerTypS}_{u^*}$. In this equation, we make sure, for every $u^*\in \ind$, that the total sum of vertex types in $\mathsf{VerTypS}_{u^*}$ is exactly $|u^*|$. That is, there is exactly one vertex type $\mathsf{VerTyp}=(u^*,\U)$ for each $u\in u^*$:

2. For every $u^*\in \ind$, $\ds{\sum_{\mathsf{VerTyp}\in \mathsf{VerTypS}_{u^*}} x_\mathsf{VerTyp}=|u^*|}$

\smallskip\noindent{\bf Equation~\ref{ilp:3}: Assigning Enough Subsets to Each Vertex Type.} We have the following notations:
\begin{itemize}
	\item For every $\mathsf{VerTyp}=(u^*,\U)\in \mathsf{VerTypS}$, every $\W=\{v,v'\}\in \U$ and $1\leq j\leq 2|\vc|$, we denote
	
	 \noindent$\mathsf{CycTypS}(\mathsf{VerTyp},\W,j)=\{\mathsf{CycTyp}=(C,\mathsf{PaAlloc}_C,\mathsf{RobTyp})\in \mathsf{CycTypS}~|~$
	 
	 \noindent$|\{\{\{u^*,v\},\{u^*,v'\}\}\in \mathsf{EdgePairs}(C)~|~$ $\mathsf{PaAlloc}_C(\{\{u^*,v\},\{u^*,v'\}\})=\mathsf{VerTyp}\}|=j\}$. That is, $\mathsf{CycTypS}(\mathsf{VerTyp},\W,j)$ is the set of cycle types that assign $\W$ to $\mathsf{VerTyp}$ exactly $j$ times.
	\item For every $\mathsf{VerTyp}=(u^*,\U)\in \mathsf{VerTypS}$, every $\W\in \U$ and $1\leq j\leq 2^{|\vc|}+|\vc|^2$, 
	
	\noindent$\mathsf{RobTypS}(\mathsf{VerTyp},\W,j)=\{\mathsf{RobTyp}=(\mathsf{CC},\mathsf{Alloc}_{\gr(\mathsf{CC})},\mathsf{NumOfCyc})\in \mathsf{RobTypS}~|~$
	
	\noindent$|\{(u^*_i,\W)\in \mathsf{NeiOfInd}(\gr(\mathsf{CC}))~|~\mathsf{Alloc}_{\gr(\mathsf{CC})}((u^*_i,\W))=\mathsf{VerTyp}\}|=j\}$. That is, $\mathsf{RobTypS}(\mathsf{VerTyp},\W,j)$ is the set of robot types that assign $\W$ to $\mathsf{VerTyp}$ exactly $j$ times.
	\end{itemize}
 Recall that, for a vertex $u\in u^*\in\ind$ of vertex type $\mathsf{VerTyp}=(u^*,\U)$, $\U$ encodes ``how'' the edges incident to $u$ are covered. That is, for every $\W\in \U$, there exists a robot $i\in [k]$ which covers the multiset of edges with one endpoint in $u$ and the other being a vertex in $\W$. A robot is able to cover this exact multiset of edges if $(u^*_i,\W)\in \mathsf{NeiOfInd}(\gr(\mathsf{CC}_i)$ or $\W=\{v,v'\}$ and $\{\{u^*,v\},\{u^*,v'\}\}\in \mathsf{EdgePairs}(C)$ for $C\in {\cal C}_i$, where $\mathsf{CC}_i$ and ${\cal C}_i$ satisfy the conditions of Lemma~\ref{lem:4cycandCC}. Therefore, for every $\mathsf{VerTyp}=(u^*,\U)\in \mathsf{VerTypS}$ and every $\W\in \U$, we ensure we have at least one $\W$ assigned to each vertex of $\mathsf{VerTyp}$:
 
 3. For every $\mathsf{VerTyp}=(u^*,\U)\in \mathsf{VerTypS}$, and every $\W\in \U$, 
 
 \noindent$\ds{\sum_{j=1}^{2{|\vc|}}\sum_{{\mathsf{CycTyp}\in \mathsf{CycTypS}(\mathsf{VerTyp},\W,j)}}j\cdot x_\mathsf{CycTyp}+}$\\ $\ds{\sum_{j=1}^{2^{|\vc|}+|\vc|^2}\sum_{{\mathsf{RobTyp}\in \mathsf{RobTypS}(\mathsf{VerTyp},\W,j)}}j\cdot x_\mathsf{RobTyp}\geq x_\mathsf{VerTyp}}$.

Observe that Equations~\ref{ilp:1}--\ref{ilp:3} ensure that every edge with one endpoint in $V\setminus \vc$ is covered.


\smallskip\noindent{\bf Equation~\ref{ilp:4}: Covering Each Edge with Both Endpoints in $\vc$.} We have the following notations:
\begin{itemize}
	\item  For every $\{u,v\}\in E$ such that $u,v\in \vc$,
	
	 \noindent$\mathsf{CycTypS}(\{u,v\})=\{\mathsf{CycTyp}=(C,\mathsf{PaAlloc}_C,\mathsf{RobTyp})\in \mathsf{CycTypS}~|~\{u,v\}\in E(C)\}$. That is, $\mathsf{CycTypS}(\{u,v\})$ is the set of cycle types $\mathsf{CycTyp}=(C,\mathsf{PaAlloc}_C,\mathsf{RobTyp})$ where $C$ covers $\{u,v\}$.
	\item For every $\{u,v\}\in E$ such that $u,v\in \vc$, 
	
	 \noindent$\mathsf{RobTypS}(\{u,v\})=\{\mathsf{RobTyp}=(\mathsf{CC},\mathsf{Alloc}_{\gr(\mathsf{CC})},$ $\mathsf{NumOfCyc})\in \mathsf{RobTypS}~|~\{u,v\}\in \mathsf{CC}\}$. That is, $\mathsf{RobTypS}(\{u,v\})$ is the set of robot types $\mathsf{RobTyp}=(\mathsf{CC},\mathsf{Alloc}_{\gr(\mathsf{CC})})$ where $\mathsf{CC}$ covers $\{u,v\}$.
\end{itemize}
In this equation we ensure that each $\{u,v\}\in E$ with both endpoints in $\vc$ is covered at least once:

4.  For every $\{u,v\}\in E$ such that $u,v\in \vc$, 

\noindent$\ds{\sum_{{\mathsf{CycTyp}\in \mathsf{CycTypS}(\{u,v\})}}x_\mathsf{CycTyp}+\sum_{{\mathsf{RobTyp}\in \mathsf{RobTypS}(\{u,v\})}}x_\mathsf{RobTyp}\geq 1}$.

 Observe that Equations~\ref{ilp:1}--\ref{ilp:4} ensure that every edge in $G$ is covered.

\smallskip\noindent{\bf Equation~\ref{ilp:5}: Assigning the Exact Amount of Cycles with Length Other Than $4$ to Each Robot Type.} We have the following notation: 
\begin{itemize}
	\item For every $\mathsf{RobTyp}\in \mathsf{RobTypS}$ and for every $2\leq j\leq 2|\vc|$, 
	
	\noindent$\mathsf{CycTypS}(\mathsf{RobTyp},j)=\{\mathsf{CycTyp}=(C,\mathsf{PaAlloc}_C,\mathsf{RobTyp})\in \mathsf{CycTypS}~|~|C|=j\}$. That is, $\mathsf{CycTypS}(\mathsf{RobTyp},j)$ is the set of cycle types, stand for a cycle of length $j$ and assigned to a robot with robot type $\mathsf{RobTyp}$.  
\end{itemize}
Let $\mathsf{RobTyp}=(\mathsf{CC},\mathsf{Alloc}_{\gr(\mathsf{CC})},\mathsf{NumOfCyc})\in \mathsf{RobTypS}$ be a robot type. In this equation, we verify that the number of cycles of length other than $4$ assigned to robots with robot type $\mathsf{RobTyp}$ is exactly as determined in $\mathsf{NumOfCyc}$:

5. For every $\mathsf{RobTyp}=(\mathsf{CC},\mathsf{Alloc}_{\gr(\mathsf{CC})},\mathsf{NumOfCyc})\in \mathsf{RobTypS}$ and for every 

\noindent$2\leq j\leq 2|\vc|$, $j\neq 4$, $\ds{\sum_{{\mathsf{CycTyp}\in \mathsf{CycTypS}(\mathsf{RobTyp},j)}}x_\mathsf{CycTyp}=N_j\cdot x_\mathsf{RobTyp}}$, 

\noindent where $\mathsf{NumOfCyc}=(N_2,N_3,N_5,N_6,\ldots,N_{2|\vc|})$.

\smallskip\noindent{\bf Equation~\ref{ilp:6}: Verifying the Budget Limit.} 
Let $\mathsf{RobTyp}=(\mathsf{CC},\mathsf{Alloc}_{\gr(\mathsf{CC})},\mathsf{NumOfCyc})$ $\in \mathsf{RobTypS}$ be a robot type, and let $i\in [k]$ be robot with robot type $\mathsf{RobTyp}$. From Lemma~\ref{lem:4cycandCC}, there exist $\mathsf{CC}_i\subseteq \widehat{E}_i$ and a multiset of cycles ${\cal C}_i$ in $\gr(\widehat{E}_i)$, such that $\mathsf{CC}_i\cup \bigcup_{C\in {\cal C}_i}E(C)=\widehat{E}_i$. Now, as we determine the robot type of the $i$-th robot to be $\mathsf{RobTyp}$, $|\mathsf{CC}_i|=|\mathsf{CC}|$, and the number of cycles in ${\cal C}_i$ of length other than $4$ is fixed by $\mathsf{NumOfCyc}$. We have the following notation:
\begin{itemize}
	\item $\ds{\mathsf{Bud}(\mathsf{RobTyp})=|\mathsf{CC}|+\sum_{2\leq j\leq 2|\vc|,j\neq 4}N_j\cdot j}$, where $\mathsf{NumOfCyc}=(N_2,N_3,N_5,N_6,\ldots,$\\$N_{2|\vc|})$.
\end{itemize}
That is, $\mathsf{Bud}(\mathsf{RobTyp})$ is the number of edges in $\mathsf{CC}_i\cup \bigcup_{C\in {\cal C}_i}E(C)$ excluding the number of edges in cycles of length $4$ in ${\cal C}_i$.
Therefore, $B-\mathsf{Bud}(\mathsf{RobTyp})$ is the budget left for the robot for the cycles in ${\cal C}_i$ of length $4$. Now, we take the largest number which is a multiple of $4$ and less or equal to $B-\mathsf{Bud}(\mathsf{RobTyp})$, to be the budget left for cycles of length $4$. So, we have the following notation:
\begin{itemize}
	\item For every $\mathsf{RobTyp}\in \mathsf{RobTypS}$,  $\mathsf{CycBud}(\mathsf{RobTyp})=\lfloor(B-\mathsf{Bud}(\mathsf{RobTyp}))\cdot \frac{1}{4}\rfloor\cdot 4$.
\end{itemize}

\noindent6. For every $\mathsf{RobTyp}\in \mathsf{RobTypS}$,\\ $\ds{\sum_{{\mathsf{CycTyp}\in \mathsf{CycTypS}(\mathsf{RobTyp},4)}}4\cdot x_\mathsf{CycTyp}\leq x_\mathsf{RobTyp}\cdot \mathsf{CycBud}(\mathsf{RobTyp})}$.

\medskip\noindent{\bf Summary of Equations.} In conclusion, given an instance $(G,\vi,k,B)$ of \cg problem, $\mathsf{Reduction}(G,\vi,k,B)$ is the instance of the ILP problem described by the following equations:


\begin{enumerate}
	\item $\ds{\sum_{\mathsf{RobTyp}\in \mathsf{RobTypS}}x_\mathsf{RobTyp}=k}$. \label{ilp:1}
	\item For every $u^*\in \ind$, $\ds{\sum_{\mathsf{VerTyp}\in \mathsf{VerTypS}_{u^*}} x_\mathsf{VerTyp}=|u^*|}$.\label{ilp:2}
	\item For every $\mathsf{VerTyp}=(u^*,\U)\in \mathsf{VerTypS}$, and every $\W\in \U$, 
	
	\noindent $\ds{\sum_{j=1}^{2{|\vc|}}\sum_{{\mathsf{CycTyp}\in \mathsf{CycTypS}(\mathsf{VerTyp},\W,j)}}j\cdot x_\mathsf{CycTyp}+}$\\
		$\ds{\sum_{j=1}^{2^{|\vc|}+|\vc|^2}\sum_{{\mathsf{RobTyp}\in \mathsf{RobTypS}(\mathsf{VerTyp},\W,j)}}j\cdot x_\mathsf{RobTyp}\geq x_\mathsf{VerTyp}}$.\label{ilp:3}
	\item For every $\{u,v\}\in E$ such that $u,v\in \vc$,\\ $\ds{\sum_{{\mathsf{CycTyp}\in \mathsf{CycTypS}(\{u,v\})}}x_\mathsf{CycTyp}+\sum_{{\mathsf{RobTyp}\in \mathsf{RobTypS}(\{u,v\})}}x_\mathsf{RobTyp}\geq 1}$.\label{ilp:4}
	\item For every $\mathsf{RobTyp}=(\mathsf{CC},\mathsf{Alloc}_{\gr(\mathsf{CC})},\mathsf{NumOfCyc})\in \mathsf{RobTypS}$ and for every 
	
	\noindent$2\leq j\leq 2|\vc|$, $j\neq 4$, $\ds{\sum_{{\mathsf{CycTyp}\in \mathsf{CycTypS}(\mathsf{RobTyp},j)}}x_\mathsf{CycTyp}=N_j\cdot x_\mathsf{RobTyp}}$, 
	
	\noindent where $\mathsf{NumOfCyc}=(N_2,N_3,N_5,N_6,\ldots,N_{2|\vc|})$.\label{ilp:5}
	\item For every $\mathsf{RobTyp}\in \mathsf{RobTypS}$,\\ $\ds{\sum_{{\mathsf{CycTyp}\in \mathsf{CycTypS}(\mathsf{RobTyp},4)}}4\cdot x_\mathsf{CycTyp}\leq x_\mathsf{RobTyp}\cdot \mathsf{CycBud}(\mathsf{RobTyp})}$.\label{ilp:6}
\end{enumerate}

\subsection{Correctness: Forward Direction}\label{sec:forDir}
Now, we turn to prove the correctness of the reduction. In particular, we have the following lemma:

\begin{lemma}\label{lem:fpt}
	Let $G$ be a connected graph, let $\vi\in V(G)$ and let $k,B\in \mathbb{N}$. Then, $(G,\vi,k,B)$ is a yes-instance of the \cg, if and only if $\mathsf{Reduction}(G,\vi,k,B)$ is a yes-instance of the {\sc Integer Linear Programming}.
\end{lemma}

We split the proof of the correctness of Lemma~\ref{lem:fpt} to two lemmas. We begin with the proof of the first direction: 

\begin{lemma}\label{lem:fpt1}
	Let $G$ be a connected graph, let $\vi\in V(G)$ and let $k,B\in \mathbb{N}$. If $(G,\vi,k,B)$ is a yes-instance of the \cg problem, then $\mathsf{Reduction}(G,\vi,k,B)$ is a yes-instance of the {\sc Integer Linear Programming}.
\end{lemma}

Towards the proof of Lemma~\ref{lem:fpt1}, we present the function $\mathsf{RobExpToILP}$. This function gets as input, for every $1\leq i\leq k$, $(\mathsf{CC}_i,{\cal C}_i)$ that is an $\widehat{E}_i$-valid pair. Then, for every $z\in \mathsf{VerTypS}\cup \mathsf{RobTypS}\cup \mathsf{CycTypS}$, $\mathsf{RobExpToILP}(z)$ is the number of elements of type $z$ derived by Definitions~\ref{def:VTD}, \ref{def:RTD} and~\ref{def:CTD}:


\begin{definition} [{\bf $\mathsf{RobExpToILP}$}]\label{def:REToILP}
	Let $G$ be a connected graph, let $\vi\in V(G)$, let $\vc$ be a vertex cover of $G$ and let $\widehat{E}_1,\ldots,\widehat{E}_k$ such that conditions of Lemma~\ref{obs:equivsol} hold. For every $1\leq i\leq k$, let $(\mathsf{CC}_i,{\cal C}_i)$ be an $\widehat{E}_i$-valid pair. Then:
	\begin{enumerate}
			\item For every $\mathsf{VerTyp}\in \mathsf{VerTypS}$, $\mathsf{RobExpToILP}(\{(\mathsf{CC}_j,{\cal C}_j)\}_{1\leq j\leq k},\mathsf{VerTyp})=|\{u\in V(G)\setminus \vc~|~\mathsf{DerVerTyp}(u)=\mathsf{VerTyp}\}|$.
				\item For every $\mathsf{RobTyp}\in \mathsf{RobTypS}$, $\mathsf{RobExpToILP}(\{(\mathsf{CC}_j,{\cal C}_j)\}_{1\leq j\leq k},\mathsf{RobTyp})=|\{1\leq i\leq k~|~\mathsf{DerRobTyp}(i)=\mathsf{RobTyp}\}|$.
		\item For every $\mathsf{CycTyp}\in \mathsf{CycTypS}$, $\mathsf{RobExpToILP}(\{(\mathsf{CC}_j,{\cal C}_j)\}_{1\leq j\leq k},\mathsf{CycTyp})=|\{(i,C)~|~1\leq i\leq k,C\in {\cal C}_i, \mathsf{DerCycTyp}(i,C)=\mathsf{CycTyp}\}|$.
	\end{enumerate} 
\end{definition}

Whenever $\{(\mathsf{CC}_i,{\cal C}_i)\}_{1\leq i\leq k}$ is clear from context, we refer to\\ $\mathsf{RobExpToILP}(\{(\mathsf{CC}_j,{\cal C}_j)\}_{1\leq j\leq k},\mathsf{VerTyp})$, $\mathsf{RobExpToILP}(\{(\mathsf{CC}_j,{\cal C}_j)\}_{1\leq j\leq k},\mathsf{RobTyp})$ and\\ $\mathsf{RobExpToILP}(\{(\mathsf{CC}_j,{\cal C}_j)\}_{1\leq j\leq k},\mathsf{CycTyp})$ as $\mathsf{RobExpToILP}(\mathsf{VerTyp})$, $\mathsf{RobExpToILP}(\mathsf{RobTyp})$\\ and $\mathsf{RobExpToILP}(\mathsf{CycTyp})$, respectively.

In the next lemma, we prove that the values given by Definition~\ref{def:REToILP} satisfy the inequalities of $\mathsf{Reduction}(G,\vi,k,B)$:

\begin{lemma}\label{lem:ForDir}
		Let $G$ be a connected graph, let $\vi\in V(G)$, let $\vc$ be a vertex cover of $G$ and let $\widehat{E}_1,\ldots,\widehat{E}_k$ such that conditions of Lemma~\ref{obs:equivsol} hold. For every $1\leq i\leq k$, let $(\mathsf{CC}_i,{\cal C}_i)$ be an $\widehat{E}_i$-valid pair. Then, the values $x_z=\mathsf{RobExpToILP}(z)$, for every $z\in \mathsf{VerTypS}\cup \mathsf{RobTypS}\cup \mathsf{CycTypS}$, satisfy the inequalities of $\mathsf{Reduction}(G,\vi,k,B)$.
\end{lemma}

For the sake of readability, we split Lemma~\ref{lem:ForDir} and its proof into two lemmas: In Lemma~\ref{lem:ForDir1} we prove that the values $x_z=\mathsf{RobExpToILP}(z)$, for every $z\in \mathsf{VerTypS}\cup \mathsf{RobTypS}\cup \mathsf{CycTypS}$, satisfy inequalities \ref{ilp:1}--\ref{ilp:3} of $\mathsf{Reduction}(G,\vi,k,B)$, and in Lemma~\ref{lem:ForDir2} we prove that these values satisfy inequalities \ref{ilp:4}--\ref{ilp:6}.

\begin{lemma}\label{lem:ForDir1}
	Let $G$ be a connected graph, let $\vi\in V(G)$, let $\vc$ be a vertex cover of $G$ and let $\widehat{E}_1,\ldots,\widehat{E}_k$ such that conditions of Lemma~\ref{obs:equivsol} hold. For every $1\leq i\leq k$, let $(\mathsf{CC}_i,{\cal C}_i)$ be an $\widehat{E}_i$-valid pair. Then, the values $x_z=\mathsf{RobExpToILP}(z)$, for every $z\in \mathsf{VerTypS}\cup \mathsf{RobTypS}\cup \mathsf{CycTypS}$, satisfy inequalities \ref{ilp:1}--\ref{ilp:3} of $\mathsf{Reduction}(G,\vi,k,B)$.
\end{lemma}


\begin{proof}
	By Lemma~\ref{lem:RobDer}, for every $1\leq i\leq k$, $\mathsf{DerVerTyp}(i)$ is a robot type, and thus, Equation~\ref{ilp:1} is satisfied. 
	
	Similarly, from Lemma~\ref{lem:VerDer}, for every $u\in u^*\in \ind$, $\mathsf{DerVerTyp}(u)$ is a vertex type, where $\mathsf{DerVerTyp}(u)=(u^*,\U)$, for some $\mathsf{NeiSubsets}\subseteq 2^{\mathsf{N}_{G^*}(u^*)\times \mathsf{2}}$. Therefore, Equations~\ref{ilp:2} are satisfied. 
	
	Now, let $\mathsf{VerTyp}=(u^*,\U)\in \mathsf{VerTypS}$, let $u\in u^*$ such that $\mathsf{DerVerTyp}(u)=\mathsf{VerTyp}$, and let $\W\in \U$. So, by Definition~\ref{def:VTD}, $\mathsf{NeiSubsets}=\mathsf{NeiSubsets}(u)=\bigcup_{1\leq i\leq k}(\{\widehat{\mathsf{N}}_{\gr(\mathsf{CC}_i)}(u)\}\cup \mathsf{NeiPairs}_{{\cal C}_i}(u))$, where $\mathsf{NeiPairs}_{{\cal C}_i}(u)=\{\{v,v'\}~|~C\in {\cal C}_i, \{\{u,v\},$ $\{u,v'\}\}\in \mathsf{EdgePairs}(C)\}$. Thus, there exists $1\leq i\leq k$ such that at least one among the following conditions holds:
	\begin{enumerate}
		\item $\W\in \mathsf{NeiPairs}_{{\cal C}_i}(u)$. Therefore, there exists $C\in {\cal C}_i$ such that $\{\{u,v\},\{u,v'\}\}\in $ $\mathsf{EdgePairs}(C)$ and $\W=\{v,v'\}$. Thus, by Definition~\ref{def:CTD}, $\mathsf{DerCycTyp}(i,C)=(\ind(C),\ind(\mathsf{PaAlloc}_{i,C})$, $\mathsf{DerRobTyp}(i))$, where $\mathsf{PaAlloc}_{i,C}$ is a vertex allocation of\\ $\mathsf{EdgePairs}(C)$, and $\mathsf{PaAlloc}_{i,C}(\{\{v,u\},\{u,v'\}\})=\mathsf{DerVerTyp}(u)=\mathsf{VerTyp}$. So, $u$ contributes at least one to $|\{\{\{u,v\},\{u,v'\}\}\in \mathsf{EdgePairs}(C)~|~$ $\mathsf{PaAlloc}_{i,C}(\{\{u,v\},\{u,v'\}\})=\mathsf{VerTyp}\}|$, thus, $u$ contributes at least one to $|\{\{\{u^*,v\},\{u^*,v'\}\}\in \mathsf{EdgePairs}(\ind(C))~|~$ $\ind(\mathsf{PaAlloc}_{i,C})(\{\{u^*,v\},\{u^*,v'\}\})=\mathsf{VerTyp}\}|$. Now, since $C\in \mathsf{Cyc}_G$, $\ind(C)\in \mathsf{Cyc}_{G^*}$, so the length of $\ind(C)$ is bounded by $\mathsf{max}\{4,2|\vc|\}$, and thus, $|\mathsf{EdgePairs}(\ind(C))|\leq 2|\vc|$. In addition, by Lemma~\ref{lem:CycDer}, $\mathsf{DerCycTyp}(i,C)=\mathsf{CycTyp}$ is a cycle type. Therefore, there exists $1\leq t\leq 2|\vc|$ such that $\mathsf{CycTyp}\in \mathsf{CycTypS}(\mathsf{VerTyp},\W,t)$. Thus, $u$ contributes at least one to $\ds{\sum_{r=1}^{2{|\vc|}}\sum_{{\mathsf{CycTyp}\in \mathsf{CycTypS}(\mathsf{VerTyp},\W,r)}}r\cdot x_\mathsf{CycTyp}}$.
		\item $\W\in \widehat{\mathsf{N}}_{\gr(\mathsf{CC}_i)}(u)$. By Definition~\ref{def:RTD}, $\mathsf{DerCycTyp}(i)=(\mathsf{CC}'_i,\mathsf{Alloc}_{G_i'},\mathsf{NumOfCyc}_i)$, where (i) $\mathsf{CC}'_i=E(G_i')$, (ii) $G'_i=\overline{G}(\gr(\mathsf{CC}_i))$ with an isomorphism $\alpha_i:V(G_i')\rightarrow V(\gr(\mathsf{CC}_i))$ and (iii) for every $(u^*_j,\W)\in \mathsf{NeiOfInd}(G'_i)$, $\mathsf{Alloc}_{G_i'}((u^*_j,\W))=\mathsf{DerVerTyp}(\alpha_i(u^*_j))$. Let $j\in \mathbb{N}$ such that $\alpha_i(u^*_j)=u$, so $\mathsf{Alloc}_{G_i'}((u^*_j,\W))=\mathsf{DerVerTyp}(\alpha_i(u^*_j))=\mathsf{VerTyp}$. Therefore, $u$ contributes at least one to $|\{(u^*_t,\W)\in \mathsf{NeiOfInd}(\gr(\mathsf{CC}_i))~|~$ $\mathsf{Alloc}_{\gr(\mathsf{CC}_i)}((u^*_t,\W))=\mathsf{VerTyp}\}|$. In addition, by Lemma \ref{lem:RobDer}, $\mathsf{DerRobTyp}(i)=\mathsf{RobTyp}$ is a robot type. Now, since $\gr(\mathsf{CC}_i)$ is a $\overline{G}$-submultigraph (see Definition~\ref{def:GbarGraph}), $|V(\gr(\mathsf{CC}_i))\cap u^*|\leq 2^{|\mathsf{N}_{G^*}({u^*})|}+|\vc|^2\leq 2^{|\vc|}+|\vc|^2$, so $|\{(u^*_t,\W)\in \mathsf{NeiOfInd}(\gr(\mathsf{CC}_i))|\leq 2^{|\vc|}+|\vc|^2$. Therefore, there exists $1\leq t\leq 2^{|\vc|}+|\vc|^2$ such that $\mathsf{RobTyp}\in \mathsf{RobTypS}(\mathsf{VerTyp},\W,t)$. Thus, $u$ contributes at least one to $\ds{\sum_{r=1}^{2^{|\vc|}+|\vc|^2}\sum_{{\mathsf{RobTyp}\in \mathsf{RobTypS}(\mathsf{VerTyp},\W,r)}}r\cdot x_\mathsf{RobTyp}}$. 
	\end{enumerate}

Therefore, each $u\in V(G)\setminus \vc$ such that $\mathsf{DerVerTyp}(u)=\mathsf{VerTyp}=(u^*,\U)$, contributes at least one to 

\noindent$\ds{\sum_{j=1}^{2{|\vc|}}\sum_{{\mathsf{CycTyp}\in \mathsf{CycTypS}(\mathsf{VerTyp},\W,j)}}j\cdot x_\mathsf{CycTyp}+\sum_{j=1}^{2^{|\vc|}+|\vc|^2}}$ 
$\ds{\sum_{{\mathsf{RobTyp}\in \mathsf{RobTypS}(\mathsf{VerTyp},\W,j)}}j\cdot x_\mathsf{RobTyp}}$. 

\noindent Thus, since $\mathsf{RobExpToILP}(\mathsf{VerTyp})=|\{u\in V(G)\setminus \vc~|~\mathsf{DerVerTyp}(u)=\mathsf{VerTyp}\}|$ and $x_\mathsf{VerTyp}=\mathsf{RobExpToILP}(\mathsf{VerTyp})$, we get that Equations~\ref{ilp:3} are satisfied.
\end{proof}

\begin{lemma}\label{lem:ForDir2}
	Let $G$ be a connected graph, let $\vi\in V(G)$, let $\vc$ be a vertex cover of $G$ and let $\widehat{E}_1,\ldots,\widehat{E}_k$ such that conditions of Lemma~\ref{obs:equivsol} hold. For every $1\leq i\leq k$, let $(\mathsf{CC}_i,{\cal C}_i)$ be an $\widehat{E}_i$-valid pair. Then, the values $x_z=\mathsf{RobExpToILP}(z)$, for every $z\in \mathsf{VerTypS}\cup \mathsf{RobTypS}\cup \mathsf{CycTypS}$, satisfy inequalities \ref{ilp:4}--\ref{ilp:6} of $\mathsf{Reduction}(G,\vi,k,B)$.
\end{lemma}

\begin{proof}
	Let $\{u,v\}\in E$ such that $u,v\in \vc$. By Condition~\ref{obs:equivsol3} of Lemma~\ref{obs:equivsol}, $E(G)\subseteq \widehat{E}_1\cup \ldots \cup\widehat{E}_k$. So, there exists $1\leq i\leq k$ such that $\{u,v\}\in \widehat{E}_i$. Since $(\mathsf{CC}_i,{\cal C}_i)$ is an $\widehat{E}_i$-valid pair, by Condition~\ref{lemCC:con6} of Definition~\ref{def:ValPa}, $\mathsf{CC}_i\cup \bigcup_{C\in {\cal C}_i}E(C)=\widehat{E}_i$. So, at least one among the following two cases holds:
	\begin{enumerate}
		\item $\{u,v\}\in \mathsf{CC}_i$. By Definition~\ref{def:RTD}, $\mathsf{DerRobTyp}(i)=(\mathsf{CC}'_i,\mathsf{Alloc}_{G_i'},\mathsf{NumOfCyc}_i)$, where $\mathsf{CC}'_i=E(G_i')$ and $G'_i=\overline{G}(\gr(\mathsf{CC}_i))$. So, $\{u,v\}\in \mathsf{CC}'_i$. Now, by Lemma~\ref{lem:RobDer}, $\mathsf{DerRobTyp}(i)=\mathsf{RobTyp}$ is a robot type. Therefore, $\mathsf{RobTyp}\in\mathsf{RobTypS}(\{u,v\})$, thus $\ds{\sum_{{\mathsf{RobTyp}\in \mathsf{RobTypS}(\{u,v\})}}x_\mathsf{RobTyp}\geq 1}$.
		\item There exists $C\in {\cal C}_i$ such that $\{u,v\}\in E(C)$. By Definition~\ref{def:CTD}, $\mathsf{DerCycTyp}(i,C)=(\ind(C),\ind(\mathsf{PaAlloc}_{i,C})$, and observe that $\{u,v\}\in \ind(C)$. In addition, by Lemma~\ref{lem:CycDer}, $\mathsf{DerCycTyp}(i,C)=\mathsf{CycTyp}$ is a cycle type. Therefore, $\mathsf{CycTyp}\in \mathsf{CycTypS}(\{u,v\})$, so $\ds{\sum_{{\mathsf{CycTyp}\in \mathsf{CycTypS}(\{u,v\})}}x_\mathsf{CycTyp}\geq 1}$.
	\end{enumerate}
	Either way, we get that $\ds{\sum_{{\mathsf{CycTyp}\in \mathsf{CycTypS}(\{u,v\})}}x_\mathsf{CycTyp}+\sum_{{\mathsf{RobTyp}\in \mathsf{RobTypS}(\{u,v\})}}}$ $x_\mathsf{RobTyp}\geq 1$, so Equations~\ref{ilp:4} are satisfied.
	
	Now, let $\mathsf{RobTyp}=(\mathsf{CC},\mathsf{Alloc}_{\gr(\mathsf{CC})},\mathsf{NumOfCyc})\in \mathsf{RobTypS}$, let $2\leq j\leq 2|\vc|$, $j\neq 4$, and let $1\leq i\leq k$ such that $\mathsf{DerRobTyp}(i)=\mathsf{RobTyp}$. By Definition~\ref{def:RTD}, $N_j$ is the number of cycles of length $j$ in ${\cal C}_i$. Now, for every $C\in {\cal C}_i$, by Definition~\ref{def:CTD}, $\mathsf{DerRobTyp}(i,C)=(\ind(C),\ind(\mathsf{PaAlloc}_{i,C})$, and in particular, $|C|=|\ind(C)|$.  In addition, by Lemma~\ref{lem:CycDer}, for every $C\in {\cal C}_i$, $\mathsf{DerCycTyp}(i,C)=\mathsf{CycTyp}$ is a cycle type. Therefore, $|\{C\in {\cal C}_i~|~\mathsf{DerRobTyp}(i,C)=(\ind(C),\ind(\mathsf{PaAlloc}_{i,C}), |\ind(C)|=j\}|=N_j$. So, Equations~\ref{ilp:5} are satisfied.
	
	
	Now, let $\mathsf{RobTyp}=(\mathsf{CC},\mathsf{Alloc}_{\gr(\mathsf{CC})},\mathsf{NumOfCyc})\in \mathsf{RobTypS}$, and let $1\leq i\leq k$ such that $\mathsf{DerRobTyp}(i)=\mathsf{RobTyp}$. Since $(\mathsf{CC}_i,{\cal C}_i)$ is an $\widehat{E}_i$-valid pair, by Condition~\ref{lemCC:con6} of Definition~\ref{def:ValPa}, $\mathsf{CC}_i\cup \bigcup_{C\in {\cal C}_i}E(C)=\widehat{E}_i$. So, $|\widehat{E}_i|=|\mathsf{CC}_i|+\sum_{C\in{\cal C}_i}|C|$. By Definition~\ref{def:RTD}, $\mathsf{DerRobTyp}(i)=(\mathsf{CC}'_i,\mathsf{Alloc}_{G_i'},\mathsf{NumOfCyc}_i)$, where $\mathsf{CC}'_i=E(G_i')$ and $G'_i=\overline{G}(\gr(\mathsf{CC}_i))$. Thus, $|\mathsf{CC}|=|\mathsf{CC}'_i|=|\mathsf{CC}_i|$. In addition, by Definition~\ref{def:RTD}, for every $2\leq j\leq 2|\vc|$, $j\neq 4$, $N_{i,j}$ is the number of cycles of size $j$ in ${\cal C}_i$, where $\mathsf{NumOfCyc}=(N_{2},N_{3},N_{5},N_{6},\ldots,N_{2|\vc|})$. Moreover, by Condition~\ref{lemCC:con5} of Definition~\ref{def:ValPa}, the cycles in ${\cal C}_i$ of length other than $4$ are simple. So, for every $C\in {\cal C}_i$, $|C|=4$ or $2\leq|C|\leq {2|\vc|}$.
	Therefore, 
	
	\noindent$\ds{|\widehat{E}_i|=|\mathsf{CC}_i|+\sum_{C\in{\cal C}_i}|C|=|\mathsf{CC}|+\sum_{2\leq j\leq 2|\vc|,j\neq 4}\sum_{C\in{\cal C}_i,|C|=j}j+\sum_{C\in{\cal C}_i,|C|=4}4=}$
		
		$\ds{|\mathsf{CC}|+\sum_{2\leq j\leq 2|\vc|,j\neq 4}N_j+\sum_{C\in{\cal C}_i,|C|=4}4}$. 
	
	Recall that $\ds{\mathsf{Bud}(\mathsf{RobTyp})=|\mathsf{CC}|+\sum_{2\leq j\leq 2|\vc|,j\neq 4}N_j}$. Thus,
	
	 \noindent$\ds{|\widehat{E}_i|=\mathsf{Bud}(\mathsf{RobTyp})+\sum_{C\in{\cal C}_i,|C|=4}4=\mathsf{Bud}(\mathsf{RobTyp})+|\{C\in{\cal C}_i~|~|C|=4\}|\cdot 4}$. Now, by Condition~\ref{obs:equivsol4} of Lemma~\ref{obs:equivsol}, $|\widehat{E}_i|\leq B$. So, $\mathsf{Bud}(\mathsf{RobTyp})+|\{C\in{\cal C}_i~|~|C|=4\}|\cdot 4\leq B$, implies that $|\{C\in{\cal C}_i~|~|C|=4\}|\leq (B-\mathsf{Bud}(\mathsf{RobTyp}))\cdot  \frac{1}{4}$. Observe that $|\{C\in{\cal C}_i~|~|C|=4\}|\in \mathbb{N}$, thus $|\{C\in{\cal C}_i~|~|C|=4\}|\leq \lfloor(B-\mathsf{Bud}(\mathsf{RobTyp}))\cdot  \frac{1}{4}\rfloor$, so $|\{C\in{\cal C}_i~|~|C|=4\}|\cdot 4\leq \lfloor(B-\mathsf{Bud}(\mathsf{RobTyp}))\cdot  \frac{1}{4}\rfloor\cdot 4=\mathsf{CycBud}(\mathsf{RobTyp})$. Therefore, the number of cycles of length $4$ in ${\cal C}_i$ is bounded by $\mathsf{CycBud}(\mathsf{RobTyp})\cdot \frac{1}{4}$. 
	
	Now, for every $C\in {\cal C}_i$, by Definition~\ref{def:CTD}, $\mathsf{DerCycTyp}(i,C)=(\ind(C),\ind(\mathsf{PaAlloc}_{i,C})$, $\mathsf{DerRobTyp}(i))=(\ind(C),\ind(\mathsf{PaAlloc}_{i,C})$, $\mathsf{RobTyp})$. In addition, by Lemma~\ref{lem:CycDer}, $\mathsf{DerCycTyp}($ $i,C)$ is a cycle type. So, for every $C\in {\cal C}_i$, $\mathsf{DerCycTyp}(i,C)\in \mathsf{CycTypS}(\mathsf{RobTyp},4)$ if and only if $|C|=4$. Therefore, $|\{C\in {\cal C}_i~|~\mathsf{DerCycTyp}(i,C)\in \mathsf{CycTypS}(\mathsf{RobTyp},4)\}|=|\{C\in{\cal C}_i~|~|C|=4\}|\leq \mathsf{CycBud}(\mathsf{RobTyp})\cdot \frac{1}{4}$. Now, for every $\mathsf{CycTyp}\in \mathsf{CycTypS}$, $x_\mathsf{CycTyp}=\mathsf{RobExpToILP}(\mathsf{CycTyp})=|\{(i,C)~|~1\leq i\leq k,C\in {\cal C}_i, \mathsf{DerCycTyp}(i,C)=\mathsf{CycTyp}\}|$. Thus, $i$ contributes at most $\mathsf{CycBud}(\mathsf{RobTyp})\cdot \frac{1}{4}$ to $\ds{\sum_{{\mathsf{CycTyp}\in \mathsf{CycTypS}(\mathsf{RobTyp},4)}}x_\mathsf{CycTyp}}$. So, 
	
	\noindent$\ds{\sum_{{\mathsf{CycTyp}\in \mathsf{CycTypS}(\mathsf{RobTyp},4)}}4\cdot x_\mathsf{CycTyp}\leq x_\mathsf{RobTyp}\cdot \mathsf{CycBud}(\mathsf{RobTyp})}$.
	Therefore, Equations~\ref{ilp:6} are satisfied. This completes the proof.
\end{proof}

Lemmas~\ref{lem:ForDir1} and~\ref{lem:ForDir2} conclude the correctness of Lemma~\ref{lem:ForDir}.

Now, we invoke Lemma~\ref{lem:ForDir} in order to prove the correctness of Lemma~\ref{lem:fpt1}:

\begin{proof}
Assume that $(G,\vi,k,B)$ is a yes-instance of the \cg problem. By Lemma~\ref{obs:equivsol}, there exist $k$ multisets $\widehat{E}'_1,\ldots,\widehat{E}'_k$ such that the conditions of Lemma~\ref{obs:equivsol} hold. So, by Observation~\ref{obs:equivso}, there exist $\widehat{E}_1,\ldots,\widehat{E}_k$ such that the conditions of Lemma~\ref{obs:equivsol} hold and for every $1\leq i\leq k$, each $\{u,v\}\in \widehat{E}_i$ appears at most twice in $\widehat{E}_i$. Thus, for every $1\leq i\leq k$, by Lemma~\ref{lem:4cycandCC}, there exists $(\mathsf{CC}_i,{\cal C}_i)$ that is an $\widehat{E}_i$-valid pair. By Lemma~\ref{lem:ForDir}, the values $x_z=\mathsf{RobExpToILP}(z)$, for every $z\in \mathsf{VerTypS}\cup \mathsf{RobTypS}\cup \mathsf{CycTypS}$, satisfy the inequalities of $\mathsf{Reduction}(G,\vi,k,B)$. Therefore, $\mathsf{Reduction}(G,\vi,k,B)$ is a yes-instance of the {\sc Integer Linear Programming}. 
\end{proof}


\subsection{Correctness: Reverse Direction}\label{sec:revDir}
In the next lemma, we state the revers direction of Lemma~\ref{lem:fpt1}:

\begin{lemma}\label{lem:fpt2}
	Let $G$ be a connected graph, let $\vi\in V(G)$ and let $k,B\in \mathbb{N}$. If $\mathsf{Reduction}(G,\vi,k,B)$ is a yes-instance of {\sc Integer Linear Programming}, then $(G,\vi,$ $k,B)$ is a yes-instance of the \cg problem.
\end{lemma}

Towards the proof of Lemma~\ref{lem:fpt2}, given values $x_z$, for every $z\in \mathsf{VerTypS}\cup \mathsf{RobTypS}\cup \mathsf{CycTypS}$, that satisfy the inequalities of $\mathsf{Reduction}(G,\vi,k,B)$, we show how to construct multisets $\widehat{E}_1,\ldots,\widehat{E}_k$ satisfy the conditions of Lemma~\ref{obs:equivsol}. Recall that intuitively, for every $z\in \mathsf{VerTypS}\cup \mathsf{RobTypS}\cup \mathsf{CycTypS}$, $x_z$ stands for the number of elements of type $z$ in a solution. First, we define a vertex type for each vertex in $V(G)\setminus \vc$. To this end, for every $u\in u^*\in \ind$, we arbitrary picks a vertex type $(u^*,\U)\in\mathsf{VerTypS}$ such that the total number of vertices we chose for a vertex type $\mathsf{VerTyp}$ is exactly $x_\mathsf{VerTyp}$:

\begin{definition} [{\bf Deriving Vertex Types from an ILP Solution}]\label{def:DVTILP}
	Let $G$ be a connected graph, let $\vi\in V(G)$, let $k,B\in \mathbb{N}$ and let $\vc$ be a vertex cover of $G$. Let $x_z$, for every $z\in \mathsf{VerTypS}\cup \mathsf{RobTypS}\cup \mathsf{CycTypS}$, be values that satisfy the inequalities of $\mathsf{Reduction}(G,\vi,k,B)$. For every $u^*\in \ind$, let $\mathsf{VerTypS}_{u^*}=\{(u^*,\U)\in\mathsf{VerTypS}\}$ and $\mathsf{VerTypFun}_{u^*}:u^*\rightarrow \mathsf{VerTypS}_{u^*}$ be such that for every $\mathsf{VerTyp}\in \mathsf{VerTypS}_{u^*}$, $|\{u\in u^*~|~\mathsf{VerTypS}_{u^*}(u)=\mathsf{VerTyp}\}|=x_\mathsf{VerTyp}$. Then, for every $u\in u^*\in \ind$, $\mathsf{ILPDerVerTyp}(\{x_z~|~z\in \mathsf{VerTypS}\cup \mathsf{RobTypS}\cup \mathsf{CycTypS}\},u)=\mathsf{VerTypS}_{u^*}(u)$.
\end{definition}

Whenever $\{x_z~|~z\in \mathsf{VerTypS}\cup \mathsf{RobTypS}\cup \mathsf{CycTypS}\}$ is clear from context, we refer to $\mathsf{ILPDerVerTyp}(\{x_z~|~z\in \mathsf{VerTypS}\cup \mathsf{RobTypS}\cup \mathsf{CycTypS}\},u)$ as $\mathsf{ILPDerVerTyp}(u)$. 

Observe that since $x_z$, for every $z\in \mathsf{VerTypS}\cup \mathsf{RobTypS}\cup \mathsf{CycTypS}$, satisfy the inequalities of $\mathsf{Reduction}(G,\vi,k,B)$, then, in particular, Equations~\ref{ilp:2} are satisfied. That is, for every $u^*\in \ind$, $\sum_{\mathsf{VerTyp}\in \mathsf{VerTypS}_{u^*}} x_\mathsf{VerTyp}=|u^*|$. Therefore, for every $u^*\in \ind$, there exists a function $\mathsf{VerTypS}_{u^*}$ such that for every $\mathsf{VerTyp}\in \mathsf{VerTypS}_{u^*}$ $|\{u\in u^*~|~\mathsf{VerTypS}_{u^*}(u)=\mathsf{VerTyp}\}|=x_\mathsf{VerTyp}$, as defined in Definition~\ref{def:DVTILP}. Thus, $\mathsf{ILPDerVerTyp}$ is well defined.

Next, we define a robot type for each robot. For every $i\in [k]$, we arbitrary determine a robot type for the $i$-th robot such that the total number of robots we choose for a robot type $\mathsf{RobTyp}$ is exactly $x_\mathsf{RobTyp}$:

\begin{definition} [{\bf Deriving Robot Types from an ILP Solution}]\label{def:DRTILP}
	Let $G$ be a connected graph, let $\vi\in V(G)$, let $k,B\in \mathbb{N}$ and let $\vc$ be a vertex cover of $G$. Let $x_z$, for every $z\in \mathsf{VerTypS}\cup \mathsf{RobTypS}\cup \mathsf{CycTypS}$, be values that satisfy the inequalities of $\mathsf{Reduction}(G,\vi,k,B)$. Let $\mathsf{RobTypFun}:[k]\rightarrow \mathsf{RobTyps}$ be such that for every $\mathsf{RobTyp}\in \mathsf{RobTypS}$, $|\{i\in[k]~|~\mathsf{RobTypFun}(i)=\mathsf{RobTyp}\}|=x_\mathsf{RobTyp}$. Then, for every $1\leq i\leq k$, $\mathsf{ILPDerRobTyp}(\{x_z~|~z\in \mathsf{VerTypS}\cup \mathsf{RobTypS}\cup \mathsf{CycTypS}\},i)=\mathsf{RobTypFun}(i)$.
\end{definition}

 Whenever $\{x_z~|~z\in \mathsf{VerTypS}\cup \mathsf{RobTypS}\cup \mathsf{CycTypS}\}$ is clear from context, we refer to $\mathsf{ILPDerRobTyp}(\{x_z~|~z\in \mathsf{VerTypS}\cup \mathsf{RobTypS}\cup \mathsf{CycTypS}\},i)$ as $\mathsf{ILPDerRobTyp}(i)$. 

Observe that since $x_z$, for every $z\in \mathsf{VerTypS}\cup \mathsf{RobTypS}\cup \mathsf{CycTypS}$, satisfy Equation~\ref{ilp:1}, so $\sum_{\mathsf{RobTyp}\in \mathsf{RobTypS}}x_\mathsf{RobTyp}=k$. Thus, there exists a function $\mathsf{RobTypFun}:[k]\rightarrow \mathsf{RobTyps}$ such that for every $\mathsf{RobTyp}\in \mathsf{RobTypS}$, $|\{i\in[k]~|~\mathsf{RobTypFun}(i)=\mathsf{RobTyp}\}|=x_\mathsf{RobTyp}$, as defined in Definition~\ref{def:DRTILP}. Therefore, $\mathsf{ILPDerRobTyp}$ is well defined.

Let $\mathsf{VerTyp}=(u^*,\U)\in\mathsf{VerTypS}$. Recall that in order to cover the edges adjacent to a vertex $u$ of type $\mathsf{VerTyp}$, we need to ``allocate'' each $\W\in\U$ to $u$. An allocation of $\W$ can be derived from a robot of a robot type $\mathsf{RobTyp}=(\mathsf{CC},\mathsf{Alloc}_{\gr(\mathsf{CC})},\mathsf{NumOfCyc})\in \mathsf{RobTypS}$, if $(u^*_j,\W)\in \mathsf{NeiOfInd}(\gr(\mathsf{CC})$ and 

\noindent$\mathsf{Alloc}_{\gr(\mathsf{CC})}(u^*_j,\W)=\mathsf{VerTyp}$. In addition, allocation of $\W$ can be derived from a cycle of a cycle type $\mathsf{CycTyp}=(C,\mathsf{PaAlloc}_C,\mathsf{RobTyp})\in \mathsf{CycTypS}$, if $\W=\{v,v'\}$, $\{\{u^*,v\},\{u^*,v'\}\}\in \mathsf{EdgePairs}(C)$ and $\mathsf{PaAlloc}_C(\{\{u^*,v\},\{u^*,v'\}\})=\mathsf{VerTyp}$. Observe that a robot or a cycle might allocate $\W$ to a vertex type $\mathsf{VerTyp}$ multiple times. In the following definition, we derive for every $\mathsf{VerTyp}=(u^*,\U)\in\mathsf{VerTypS}$ and $\W\in\U$ the set of total allocations of $\W$ to $\mathsf{VerTyp}$ from both robots and cycles. 

\begin{definition} [{\bf Deriving Subsets from an ILP Solution}]\label{def:DSILP}
	Let $G$ be a connected graph, let $\vi\in V(G)$, let $k,B\in \mathbb{N}$ and let $\vc$ be a vertex cover of $G$. Let $x_z$, for every $z\in \mathsf{VerTypS}\cup \mathsf{RobTypS}\cup \mathsf{CycTypS}$, be values that satisfy the inequalities of $\mathsf{Reduction}(G,\vi,k,B)$. Then, for every $\mathsf{VerTyp}=(u^*,\U)\in \mathsf{VerTypS}$, and every $\W\in \U$, $\mathsf{ILPDerSubToAlloc}(\{x_z~|~z\in \mathsf{VerTypS}\cup \mathsf{RobTypS}\cup \mathsf{CycTypS}\},\mathsf{VerTyp},\W)=\{(\mathsf{CycTyp},i,t)$ $|~1\leq j\leq 2|\vc|, \mathsf{CycTyp}\in \mathsf{CycTypS}(\mathsf{VerTyp},\W,j), 1\leq i\leq x_\mathsf{CycTyp}, 1\leq t\leq j\}\cup \{(i,t)~|~1\leq j\leq 2^{|\vc|}+|\vc|^2, 1\leq i\leq k,\mathsf{ILPDerRobTyp}(i)\in \mathsf{RobTypS}(\mathsf{VerTyp},\W,j), 1\leq t\leq j\}$.
\end{definition}

 Whenever $\{x_z~|~z\in \mathsf{VerTypS}\cup \mathsf{RobTypS}\cup \mathsf{CycTypS}\}$ is clear from context, we refer to $\mathsf{ILPDerSubToAlloc}(\{x_z~|~z\in \mathsf{VerTypS}\cup \mathsf{RobTypS}\cup \mathsf{CycTypS}\},\mathsf{VerTyp},\W)$ as 

 \noindent$\mathsf{ILPDerSubToAlloc}(\mathsf{VerTyp},\W)$.

Now, we allocate the subsets we derived in Definition~\ref{def:DSILP} to vertices. For every $\mathsf{VerTyp}=(u^*,\U)\in\mathsf{VerTypS}$ and $\W\in\U$, we arbitrary allocate each $\mathsf{Sub}\in \mathsf{ILPDerSubToAlloc}$ $(\mathsf{VerTyp},\W)$ to a vertex in $u^*$, while ensuring that each vertex of a vertex type $\mathsf{VerTyp}$ gets at least one item allocated.

\begin{definition} [{\bf Deriving Vertex Allocation of $\mathsf{ILPDerSubToAlloc}$ from an ILP Solution}]\label{def:DVAILP}
	Let $G$ be a connected graph, let $\vi\in V(G)$, let $k,B\in \mathbb{N}$ and let $\vc$ be a vertex cover of $G$. Let $x_z$, for every $z\in \mathsf{VerTypS}\cup \mathsf{RobTypS}\cup \mathsf{CycTypS}$, be values that satisfy the inequalities of $\mathsf{Reduction}(G,\vi,k,B)$. For every $\mathsf{VerTyp}=(u^*,\U)\in \mathsf{VerTypS}$ and every $\W\in \U$, let $\mathsf{SubAlloc}_{\mathsf{VerTyp},\W}: \mathsf{ILPDerSubToAlloc}(\mathsf{VerTyp},\W)\rightarrow u^*$ be a function such that:
	 \begin{itemize}
	 	\item For every $u\in u^*$ such that $\mathsf{ILPDerVerTyp}(u)=\mathsf{VerTyp}$, there exists\\ $\mathsf{Sub}\in \mathsf{ILPDerSubToAlloc}$ $(\mathsf{VerTyp},\W)$ such that $\mathsf{SubAlloc}_{\mathsf{VerTyp},\W}(\mathsf{Sub})=u$. 
	 \end{itemize}
Then, for every $\mathsf{VerTyp}=(u^*,\U)\in \mathsf{VerTypS}$, $\W\in \U$ and $\mathsf{Sub}\in \mathsf{ILPDerSubToAlloc}$ $(\mathsf{VerTyp},\W)$, 

\noindent$\mathsf{ILPDerSubAlloc}(\{x_z~|~z\in \mathsf{VerTypS}\cup \mathsf{RobTypS}\cup \mathsf{CycTypS}\},\mathsf{VerTyp},\W,\mathsf{Sub})=$\\ $\mathsf{SubAlloc}_{\mathsf{VerTyp},\W}(\mathsf{Sub})$.
\end{definition}

 Whenever $\{x_z~|~z\in \mathsf{VerTypS}\cup \mathsf{RobTypS}\cup \mathsf{CycTypS}\}$ is clear from context, we refer to $\mathsf{ILPDerSubAlloc}(\{x_z~|~z\in \mathsf{VerTypS}\cup \mathsf{RobTypS}\cup \mathsf{CycTypS}\},\mathsf{VerTyp},\W,\mathsf{Sub})$ as 
 
 \noindent$\mathsf{ILPDerSubAlloc}(\mathsf{VerTyp},\W,\mathsf{Sub})$.

Recall that, by Equation~\ref{ilp:3}, for every $\mathsf{VerTyp}=(u^*,\U)\in \mathsf{VerTypS}$, and every $\W\in \U$, 

\noindent$\displaystyle{\sum_{j=1}^{2{|\vc|}}\sum_{{\mathsf{CycTyp}\in \mathsf{CycTypS}(\mathsf{VerTyp},\W,j)}}j\cdot x_\mathsf{CycTyp}+}$\\$\ds{\sum_{j=1}^{2^{|\vc|}+|\vc|^2}}$ $\displaystyle{\sum_{{\mathsf{RobTyp}\in \mathsf{RobTypS}(\mathsf{VerTyp},\W,j)}}j\cdot x_\mathsf{RobTyp}\geq x_\mathsf{VerTyp}}$. 

Observe that the left member of the equation equals to the number of allocations of $\W$ we have for vertices of vertex type $\mathsf{VerTyp}$. That is,
 
  \noindent$\displaystyle{|\mathsf{ILPDerSubAlloc}(\mathsf{VerTyp},\W,\mathsf{Sub})|=\sum_{j=1}^{2{|\vc|}}\sum_{{\mathsf{CycTyp}\in \mathsf{CycTypS}(\mathsf{VerTyp},\W,j)}}j\cdot x_\mathsf{CycTyp}+}$\\$\ds{\sum_{j=1}^{2^{|\vc|}+|\vc|^2}\sum_{{\mathsf{RobTyp}\in \mathsf{RobTypS}(\mathsf{VerTyp},\W,j)}}j\cdot x_\mathsf{RobTyp}}$. So, since $x_z$, for every $z\in \mathsf{VerTypS}\cup \mathsf{RobTypS}\cup \mathsf{CycTypS}$, satisfy Equations~\ref{ilp:3}, we have enough subsets to allocate. In addition, $u^*\neq \emptyset$. So, there exists a function $\mathsf{SubAlloc}_{\mathsf{VerTyp},\W}$ as defined in Definition~\ref{def:DVAILP}. Therefore, $\mathsf{ILPDerSubAlloc}$ is well defined.

 Now, we look on the vertex allocation of $\mathsf{ILPDerSubToAlloc}$, defined in Definition~\ref{def:DVAILP}, from the perspective of the robots. Let $1\leq i\leq k$, and let $\mathsf{RobTyp}=(\mathsf{CC},\mathsf{Alloc}_{\gr(\mathsf{CC})},\mathsf{NumOfCyc})$ $\in \mathsf{RobTypS}$ such that $\mathsf{ILPDerRobTyp}(i)=\mathsf{RobTyp}$. In addition, $\W\in 2^{\mathsf{N}_{G^*}(u^*)\times \mathsf{2}}$, $\mathsf{VerTyp}\in \mathsf{VerTypS}$ and $1\leq r\leq 2^{|\vc|}+|\vc|^2$ such that $\mathsf{RobTyp}\in \mathsf{RobTypS}(\mathsf{VerTyp},\W,r)$. Notice that this implies $\mathsf{SetOfSub}(u^*,\W,\mathsf{VerTyp})=\{(u^*_j,\W)\in \mathsf{NeiOfInd}(\gr(\mathsf{CC})$ $|~\mathsf{Alloc}_{\gr(\mathsf{CC})}$ $((u^*_j,\W))=\mathsf{VerTyp}\}$ has exactly $r$ elements. So, there exists exactly $r$ allocations of $\W$ associated with the $i$-th robot in $\mathsf{ILPDerSubToAlloc}(\mathsf{VerTyp},\W)$. In the following definition, we ``execute'' the allocations associated with the $i$-th robot, that is, we replace each $u_j^*\in V(\gr(\mathsf{CC}))$ by a vertex $u\in u^*$, derived by $\mathsf{ILPDerSubAlloc}$. Observe that the allocations of the $i$-th robot in $\mathsf{ILPDerSubToAlloc}(\mathsf{VerTyp},\W)$ are labeled $1$ to $r$. That is, $(i,1),\ldots,(i,r)\in \mathsf{ILPDerSubToAlloc}(\mathsf{VerTyp},\W)$. So, first we arbitrary label each element in $\mathsf{SetOfSub}(u^*,\W,\mathsf{VerTyp})$ by some unique $t\in[r]$, and then we replace the vertex by $\mathsf{ILPDerSubAlloc}(\mathsf{VerTyp},\W,(i,t))$.
 
 Recall that a permutation of a multiset $A$ is a bijection $\mathsf{Permut}_A:A\rightarrow [A]$.

 \begin{definition} [{\bf Deriving a Transformation of $\mathsf{CC}$ from an ILP Solution}]\label{def:DTCCILP}
 	Let $G$ be a connected graph, let $\vi\in V(G)$, let $k,B\in \mathbb{N}$ and let $\vc$ be a vertex cover of $G$. Let $x_z$, for every $z\in \mathsf{VerTypS}\cup \mathsf{RobTypS}\cup \mathsf{CycTypS}$, be values that satisfy the inequalities of $\mathsf{Reduction}(G,\vi,k,B)$. Let $1\leq i\leq k$, and let $\mathsf{RobTyp}=(\mathsf{CC},\mathsf{Alloc}_{\gr(\mathsf{CC})},\mathsf{NumOfCyc})\in \mathsf{RobTypS}$ such that $\mathsf{ILPDerRobTyp}(i)=\mathsf{RobTyp}$. For every $\W\in 2^{\mathsf{N}_{G^*}(u^*)\times \mathsf{2}}$, $\mathsf{VerTyp}\in \mathsf{VerTypS}$ and $1\leq r\leq 2^{|\vc|}+|\vc|^2$ such that $\mathsf{RobTyp}\in \mathsf{RobTypS}(\mathsf{VerTyp},\W,r)$, let $\mathsf{SetOfSub}(u^*,\W,\mathsf{VerTyp})=\{(u^*_j,\W)\in \mathsf{NeiOfInd}(\gr(\mathsf{CC})~|~\mathsf{Alloc}_{\gr(\mathsf{CC})}((u^*_j,$ $\W))=\mathsf{VerTyp}\}$ and let $\mathsf{Permut}_{u^*,\W,\mathsf{VerTyp}}:\mathsf{SetOfSub}(u^*,\W,\mathsf{VerTyp})\rightarrow [r]$ be a permutation. Now, let $G'$ be the multigraph obtained from $\gr(\mathsf{CC})$ by replacing each $u^*_j$ by $\mathsf{ILPDerSubAlloc}$ $(\mathsf{VerTyp},\W,(i,\mathsf{Permut}_{u^*,\W,\mathsf{VerTyp}}(u^*_j,\W))$, where $(u^*_j,\W)\in \mathsf{NeiOfInd}(\gr(\mathsf{CC}))$ and $\mathsf{Alloc}_{\gr(\mathsf{CC})}(u^*_j,\W)=\mathsf{VerTyp}$. Then, $\mathsf{ILPDerCCTransf}(\{x_z~|$ $z\in \mathsf{VerTypS}\cup \mathsf{RobTypS}\cup \mathsf{CycTypS}\},i)=E(G')$.
 \end{definition}

  Whenever $\{x_z~|~z\in \mathsf{VerTypS}\cup \mathsf{RobTypS}\cup \mathsf{CycTypS}\}$ is clear from context, we refer to $\mathsf{ILPDerCCTransf}(\{x_z~|~z\in \mathsf{VerTypS}\cup \mathsf{RobTypS}\cup \mathsf{CycTypS}\},i)$ as $\mathsf{ILPDerCCTransf}(i)$.
  
  Let $1\leq i\leq k$, and let $\mathsf{RobTyp}=(\mathsf{CC},\mathsf{Alloc}_{\gr(\mathsf{CC})},\mathsf{NumOfCyc})\in \mathsf{RobTypS}$ such that $\mathsf{ILPDerRobTyp}(i)=\mathsf{RobTyp}$. In the next lemma, we show that $\mathsf{ILPDerCCTransf}(i)$ has the properties which will be useful later: (i) $\mathsf{ILPDerCCTransf}(i)$ is a multiset with elements from $E(G)$, (ii) $|\mathsf{ILPDerCCTransf}(i)|=|\mathsf{CC}|$, (iii) vertices from $\vc$ are not replaced, and (iv) the properties of $\mathsf{CC}$ given by Condition~\ref{def:RT2} of Definition~\ref{def:RT} are maintained also in $\mathsf{ILPDerCCTransf}(i)$.

  \begin{lemma}\label{lem:TranCCPro}
  Let $G$ be a connected graph, let $\vi\in V(G)$, let $k,B\in \mathbb{N}$ and let $\vc$ be a vertex cover of $G$. Let $x_z$, for every $z\in \mathsf{VerTypS}\cup \mathsf{RobTypS}\cup \mathsf{CycTypS}$, be values that satisfy the inequalities of $\mathsf{Reduction}(G,\vi,k,B)$. Let $1\leq i\leq k$, and let $\mathsf{RobTyp}=(\mathsf{CC},\mathsf{Alloc}_{\gr(\mathsf{CC})},\mathsf{NumOfCyc})\in \mathsf{RobTypS}$ such that $\mathsf{ILPDerRobTyp}(i)=\mathsf{RobTyp}$. Then, the following conditions hold:
  \begin{enumerate}
  	\item $\mathsf{ILPDerCCTransf}(i)$ is a multiset with elements from $E(G)$.\label{lem:TranCCPro1}
  	\item $|\mathsf{ILPDerCCTransf}(i)|=|\mathsf{CC}|$.\label{lem:TranCCPro2}
  	\item $V(\gr(\mathsf{CC}))\cap \vc= V(\gr(\mathsf{ILPToRobExp}(i))\cap \vc$. \label{lem:TranCCPro3}
  	\item $\vi\in V(\gr(\mathsf{ILPDerCCTransf}(i)))$, $\gr(\mathsf{ILPDerCCTransf}(i))$ is connected and every vertex in $\gr(\mathsf{ILPDerCCTransf}(i))$ has even degree in it.\label{lem:TranCCPro4}
  \end{enumerate}
  \end{lemma}

  \begin{proof}
  We prove that $\mathsf{ILPDerCCTransf}(i)$ is a multiset with elements from $E(G)$. By Condition~\ref{def:RT1} of Definition~\ref{def:RT}, $\mathsf{CC}\subseteq E(\overline{G})$. Therefore, it is enough to show that each $u_j^*\in V(\gr(\mathsf{CC}))$ is replaced by some $u\in u^*$. Let $u_j^*\in V(\gr(\mathsf{CC}))$. By Definition~\ref{def:SOP}, $(u^*_i,\widehat{\mathsf{N}}_{\gr(\mathsf{CC})}(u^*_i))\in \mathsf{NeiOfInd}(\gr(\mathsf{CC}))$. Since $\mathsf{Alloc}_{\gr(\mathsf{CC})}$ is a vertex allocation of $\mathsf{NeiOfInd}(\gr(\mathsf{CC}))$, by Definition~\ref{def:SOPP}, there exists $\mathsf{VerTyp}=(u^*,\U)\in \mathsf{VerTypS}$ such that $\mathsf{Alloc}_{\gr(\mathsf{CC})}((u^*_j,\W))=\mathsf{VerTyp}$, where $\W=\widehat{\mathsf{N}}_{\gr(\mathsf{CC})}(u^*_i)\in \U$. Thus, there exists $1\leq r\leq 2^{|\vc|}+|\vc|^2$ such that $\mathsf{RobTyp}\in \mathsf{RobTypS}(\mathsf{VerTyp},\W,r)$. 
  
  Now, let $\mathsf{SetOfSub}(u^*,\W,\mathsf{VerTyp})=\{(u^*_j,\W)\in \mathsf{NeiOfInd}(\gr(\mathsf{CC})~|$\\$\mathsf{Alloc}_{\gr(\mathsf{CC})}$ $((u^*_j,\W))=\mathsf{VerTyp}\}$ and let $\mathsf{Permut}_{u^*,\W,\mathsf{VerTyp}}:\mathsf{SetOfSub}(u^*,\W,$ $\mathsf{VerTyp})\rightarrow [r]$ be the permutation defined in Definition~\ref{def:DTCCILP}. Let $t\in[r]$ be such that $\mathsf{Permut}_{u^*,\W,\mathsf{VerTyp}}(u^*_j,\W)=t$. Since $\mathsf{RobTyp}\in \mathsf{RobTypS}(\mathsf{VerTyp},\W,r)$, by Definition~\ref{def:DSILP}, $(i,1),\ldots,(i,r)\in\mathsf{ILPDerSubToAlloc}$ $(\mathsf{VerTyp},\W)$. By Definition~\ref{def:DTCCILP}, $u_j^*$ is replaced by $\mathsf{ILPDerSubAlloc}(\mathsf{VerTyp},\W,(i,t))$, and by Definition~\ref{def:DVAILP}, $\mathsf{ILPDerSubAlloc}($ $\mathsf{VerTyp},\W,(i,t))\in u^*$. Therefore, Condition~\ref{lem:TranCCPro1} holds.
  
  Observe that $|\mathsf{ILPDerCCTransf}(i)|=|\mathsf{CC}|$, so Condition~\ref{lem:TranCCPro2} holds. In addition, it is easy to see, by Definition~\ref{def:DTCCILP}, that only vertices from $V(G)\setminus \vc$ are replaced, so Condition~\ref{lem:TranCCPro3} holds.
  
  By Condition~\ref{def:RT2} of Definition~\ref{def:RT}, $\vi\in V(\gr(\mathsf{CC})$. Recall that we assume $\vi\in \vc$, so $\vi\in V(\gr(\mathsf{ILPDerCCTransf}(i)))$. In addition, by Condition~\ref{def:RT2} of Definition~\ref{def:RT}, $\gr(\mathsf{CC})$ is connected. Thus, observe that $\gr(\mathsf{ILPDerCCTransf}(i))$ is also connected. Now, we show that every vertex in $\gr(\mathsf{ILPDerCCTransf}(i))$ has even degree in it. By Condition~\ref{def:RT2} of Definition~\ref{def:RT}, every vertex in $\gr(\mathsf{CC})$ has even degree in it. Let $u\in V(\gr(\mathsf{ILPDerCCTransf}(i)))$. If $u\in \vc$, observe that its degree in\\ $\gr(\mathsf{ILPDerCCTransf}(i))$ is equal to its degree in $\gr(\mathsf{CC})$, so it is even. Otherwise $u\in u^*\in \ind$, and its degree in $\gr(\mathsf{ILPDerCCTransf}(i))$ equals to the sum of degrees of the vertices $u^*_j$ in $\gr(\mathsf{CC})$, which are replaced by $u$, by Definition~\ref{def:DTCCILP}. The degree of each such $u^*_j$ is even in $\gr(\mathsf{CC})$, so the degree of $u$ in $\gr(\mathsf{ILPDerCCTransf}(i))$ is also even. Therefore, Condition~\ref{lem:TranCCPro4} holds.
  \end{proof}

Now, we look at the vertex allocation of $\mathsf{ILPDerSubToAlloc}$, defined in Definition~\ref{def:DVAILP}, from the perspective of the cycles. In the next definition we execute a processes similar to Definition~\ref{def:DTCCILP}. Let $\mathsf{CycTyp}=(C,\mathsf{PaAlloc}_C,\mathsf{RobTyp})\in \mathsf{CycTypS}$ and let $1\leq i\leq x_\mathsf{CycTyp}$. In addition, let $\W=\{v,v'\}\in 2^{\mathsf{N}_{G^*}(u^*)\times \mathsf{2}}$, $\mathsf{VerTyp}\in \mathsf{VerTypS}$ and let $1\leq r\leq 2{|\vc|}$ such that $\mathsf{CycTyp}\in \mathsf{CycTypS}(\mathsf{VerTyp},\W,r)$. Notice that this implies $\mathsf{SetOfSub}(u^*,\W,\mathsf{VerTyp})=\{\{\{u^*,v\},\{u^*,v'\}\}\in \mathsf{EdgePairs}(C)~|~\mathsf{PaAlloc}_{C}=\mathsf{VerTyp}\}$ has exactly $r$ elements. So, there exists exactly $r$ allocations of $\W$ associated with the $i$-th cycle of $\mathsf{CycTyp}$ in $\mathsf{ILPDerSubToAlloc}(\mathsf{VerTyp},\W)$. In the following definition, we ``execute'' the allocations associated with the $i$-th cycle of $\mathsf{CycTyp}$, that is, we replace each $u^*\in V(C)$ by a vertex $u\in u^*$, derived by $\mathsf{ILPDerSubAlloc}$. Observe that the allocations of the $i$-th cycle of $\mathsf{CycTyp}$ in $\mathsf{ILPDerSubToAlloc}(\mathsf{VerTyp},\W)$ are labeled $1$ to $r$. That is, $(\mathsf{CycTyp},i,1),\ldots,(\mathsf{CycTyp},i,r)\in \mathsf{ILPDerSubToAlloc}(\mathsf{VerTyp},\W)$. So, first we arbitrary label each element in $\mathsf{SetOfSub}(u^*,\W,\mathsf{VerTyp})$ by some unique $t\in[r]$, and then we replace the vertex by $\mathsf{ILPDerSubAlloc}(\mathsf{VerTyp},\W,(\mathsf{CycTyp},i,t))$.

 \begin{definition} [{\bf Deriving a Transformation of a Cycle from an ILP Solution}]\label{def:DTCILP}
	Let $G$ be a connected graph, let $\vi\in V(G)$, let $k,B\in \mathbb{N}$ and let $\vc$ be a vertex cover of $G$. Let $x_z$, for every $z\in \mathsf{VerTypS}\cup \mathsf{RobTypS}\cup \mathsf{CycTypS}$, be values that satisfy the inequalities of $\mathsf{Reduction}(G,\vi,k,B)$. Let $\mathsf{CycTyp}=(C,\mathsf{PaAlloc}_C,\mathsf{RobTyp})\in \mathsf{CycTypS}$ and let $1\leq i\leq x_\mathsf{CycTyp}$. For every $\W=\{v,v'\}\in 2^{\mathsf{N}_{G^*}(u^*)\times \mathsf{2}}$, $\mathsf{VerTyp}\in \mathsf{VerTypS}$ and $1\leq r\leq 2{|\vc|}$ such that $\mathsf{CycTyp}\in \mathsf{CycTypS}(\mathsf{VerTyp},\W,r)$, let $\mathsf{SetOfSub}(u^*,\W,\mathsf{VerTyp})=\{\{\{u^*,v\},\{u^*,v'\}\}\in \mathsf{EdgePairs}(C)~|~\mathsf{PaAlloc}_{C}=\mathsf{VerTyp}\}$ and let $\mathsf{Permut}_{u^*,\W,\mathsf{VerTyp}}:\mathsf{SetOfSub}(u^*,\W,\mathsf{VerTyp})\rightarrow [r]$ be a permutation. Then, $\mathsf{ILPDerCycTransf}(\{x_z~|~z\in \mathsf{VerTypS}\cup \mathsf{RobTypS}\cup \mathsf{CycTypS}\},\mathsf{CycTyp},i)$ is the cycle obtained from $C$ by replacing each $u^*\in V(C)$ by $\mathsf{ILPDerSubAlloc}(\mathsf{VerTyp},\{v,v'\},(\mathsf{CycTyp},i,$ $\mathsf{Permut}_{u^*,\{v,v'\},\mathsf{VerTyp}}(u^*,\{\{u^*,v\},\{u^*,$\\$v'\}\})$, where $v$ and $v'$ are the two vertices that appear before and after $u^*$ in $C$, respectively. 
\end{definition}


 Whenever $\{x_z~|~z\in \mathsf{VerTypS}\cup \mathsf{RobTypS}\cup \mathsf{CycTypS}\}$ is clear from context, we refer to $\mathsf{ILPDerCycTransf}(\{x_z~|~z\in \mathsf{VerTypS}\cup \mathsf{RobTypS}\cup \mathsf{CycTypS}\},\mathsf{CycTyp},i)$ as $\mathsf{ILPDerCycTransf}($ $\mathsf{CycTyp},i)$.
 
 Let $\mathsf{CycTyp}=(C,\mathsf{PaAlloc}_C,\mathsf{RobTyp})\in \mathsf{CycTypS}$ and let $1\leq i\leq x_\mathsf{CycTyp}$. In the next lemma, we show that $\mathsf{ILPDerCycTransf}(\mathsf{CycTyp},i)$ has some properties that will be useful later.
 
 \begin{lemma}\label{lem:TranCPro}
 	Let $G$ be a connected graph, let $\vi\in V(G)$, let $k,B\in \mathbb{N}$ and let $\vc$ be a vertex cover of $G$. Let $x_z$, for every $z\in \mathsf{VerTypS}\cup \mathsf{RobTypS}\cup \mathsf{CycTypS}$, be values that satisfy the inequalities of $\mathsf{Reduction}(G,\vi,k,B)$. Let $\mathsf{CycTyp}=(C,\mathsf{PaAlloc}_C,\mathsf{RobTyp})\in \mathsf{CycTypS}$ and let $1\leq i\leq x_\mathsf{CycTyp}$. Then, the following conditions hold:
 	\begin{enumerate}
 		\item $\mathsf{ILPDerCycTransf}(\mathsf{CycTyp},i)$ is a cycle in $G$.\label{lem:TranCPro1}
 			\item $|C|=|\mathsf{ILPDerCycTransf}(\mathsf{CycTyp},i)|$.\label{lem:TranCPro2}
 		\item $V(C)\cap \vc=V(\mathsf{ILPDerCycTransf}(\mathsf{CycTyp},i))\cap \vc$.\label{lem:TranCPro3}
 	\end{enumerate}
 \end{lemma}

\begin{proof}
We prove that $\mathsf{ILPDerCycTransf}(\mathsf{CycTyp},i)$ is a cycle in $G$. By Definition~\ref{def:CT}, $C\in \mathsf{Cyc}_{G^*}$, so it is enough to show that every vertex $u^*\in V(C)$ is replaced by some $u\in u^*$. Let $u^*\in V(C)$. Let $v$ and $v'$ be the two vertices that appear before and after $u^*$ in $C$, respectively, and let $\W=\{v,v'\}$.
By Definition~\ref{def:POC}, $\{\{u^*,v\},\{u^*,v'\}\}\in \mathsf{EdgePairs}(C)$. Since $\mathsf{PaAlloc}_C$ is a vertex allocation of $\mathsf{EdgePairs}(C)$, by Definition~\ref{def:PPOC}, there exists $\mathsf{VerTyp}=(u^*,\U)\in \mathsf{VerTypS}$ such that $\W\in\U$ and $\mathsf{PaAlloc}_C(\{\{u^*,v\},\{u^*,v'\}\})=\mathsf{VerTyp}$. Thus, there exists $1\leq r\leq |\vc|^2$ such that $\mathsf{CycTyp}\in \mathsf{CycTypS}(\mathsf{VerTyp},\W,r)$. 

Now, let $\mathsf{SetOfSub}(u^*,\W,\mathsf{VerTyp})=\{\{\{u^*,v\},\{u^*,v'\}\}\in \mathsf{EdgePairs}(C)~|~\mathsf{PaAlloc}_{C}$ $=\mathsf{VerTyp}\}$ and let $\mathsf{Permut}_{u^*,\W,\mathsf{VerTyp}}:\mathsf{SetOfSub}(u^*,\W,\mathsf{VerTyp})\rightarrow [r]$ be the permutation defined in Definition~\ref{def:DTCILP}. Let $t\in[r]$ be such that $\mathsf{Permut}_{u^*,\W,\mathsf{VerTyp}}(u^*_j,\W)=t$. 
Since $\mathsf{CycTyp}\in \mathsf{CycTypS}(\mathsf{VerTyp},\W,r)$, by Definition~\ref{def:DSILP}, $(\mathsf{CycTyp},i,1),\ldots,(\mathsf{CycTyp},i,$ $r)\in\mathsf{ILPDerSubToAlloc}$ $(\mathsf{VerTyp},\W)$. By Definition~\ref{def:DTCILP}, $u_j^*$ is replaced by\\ $\mathsf{ILPDerSubAlloc}(\mathsf{VerTyp},\W,(\mathsf{CycTyp},i,t))$, and by Definition~\ref{def:DVAILP},\\ $\mathsf{ILPDerSubAlloc}(\mathsf{VerTyp},\W,(\mathsf{CycTyp},i,t))\in u^*$. Therefore, Condition~\ref{lem:TranCPro1} holds.

  Observe that $|C|=|\mathsf{ILPDerCycTransf}(\mathsf{CycTyp},i)|$, so Condition~\ref{lem:TranCCPro2} holds. In addition, it is easy to see, by Definition~\ref{def:DTCILP}, that only vertices from $V(G)\setminus \vc$ are replaced, so Condition~\ref{lem:TranCPro3} holds.
\end{proof}

 Next, we allocate the cycles to the robots. We allocate cycles of type $\mathsf{CycTyp}=(C,\mathsf{PaAlloc}_C,\mathsf{RobTyp})\in \mathsf{CycTypS}$ to robots of type $\mathsf{RobTyp}$, while preserving the budget limitation of the robots:
  Each robot of type $\mathsf{RobTyp}=(\mathsf{CC},\mathsf{Alloc}_{\gr(\mathsf{CC})},\mathsf{NumOfCyc})\in \mathsf{RobTypS}$ gets exactly $N_j$ cycles of length $j$, for every $2\leq j\leq 2|\vc|$, $j\neq 4$, where $\mathsf{NumOfCyc}=(N_2,N_3,N_5,N_6,\ldots,N_{2|\vc|})$. Then, the cycles of type $\mathsf{CycTyp}\in \mathsf{CycTypS}(\mathsf{RobTyp},4)$ are allocated to robots of type $\mathsf{RobTyp}$ ``equally''. That is, the number of cycles of length $4$ allocated to a robot $i$ of type $\mathsf{RobTyp}$ is larger by at most $1$ than the number of cycles of length $4$ allocated to other robot $i'$ of type $\mathsf{RobTyp}$.

\begin{definition} [{\bf Deriving Robot Allocation of Cycles from an ILP Solution}]\label{def:DRAILP}
	Let $G$ be a connected graph, let $\vi\in V(G)$, let $k,B\in \mathbb{N}$ and let $\vc$ be a vertex cover of $G$. Let $x_z$, for every $z\in \mathsf{VerTypS}\cup \mathsf{RobTypS}\cup \mathsf{CycTypS}$, be values that satisfy the inequalities of $\mathsf{Reduction}(G,\vi,k,B)$. For every $\mathsf{CycTyp}=(C,\mathsf{PaAlloc}_C,\mathsf{RobTyp})\in \mathsf{CycTypS}$, let  $\mathsf{CycAlloc}_{\mathsf{CycTyp}}:[x_\mathsf{CycTyp}]\rightarrow [k]$ be a function such that the following conditions hold:
	\begin{enumerate}
		\item For every $i\in [x_\mathsf{CycTyp}], \mathsf{ILPDerRobTyp}(\mathsf{CycAlloc}_{\mathsf{CycTyp}}(i))=\mathsf{RobTyp}$. \label{def:DRAILPCon1}
		\item For every $\mathsf{RobTyp}=(\mathsf{CC},\mathsf{Alloc}_{\gr(\mathsf{CC})},\mathsf{NumOfCyc})\in \mathsf{RobTypS}$, $2\leq j\leq 2|\vc|$, $j\neq 4$, and $i\in [k]$ such that $\mathsf{ILPDerRobTyp}(i)=\mathsf{RobTyp}$, it holds that 
		
		\noindent$\displaystyle{\sum_{{\mathsf{CycTyp}\in \mathsf{CycTypS}(\mathsf{RobTyp},j)}}|\{t\in[x_\mathsf{CycTyp}]~|~\mathsf{CycTypFun}_{\mathsf{CycTyp}}(t)=i\}|=N_j}$,  
		
		\noindent where $\mathsf{NumOfCyc}=(N_2,N_3,N_5,N_6,\ldots,N_{2|\vc|})$.\label{def:DRAILPCon2}
		\item For every $i,i'\in [k]$ such that $\mathsf{ILPDerRobTyp}(i)=\mathsf{ILPDerRobTyp}(i')=\mathsf{RobTyp}$, it holds that
		
		 \noindent$\displaystyle{|\sum_{{\mathsf{CycTyp}\in \mathsf{CycTypS}(\mathsf{RobTyp},4)}}}$ $\displaystyle{|\{t\in[x_\mathsf{CycTyp}]~|~\mathsf{CycAlloc}_{\mathsf{CycTyp}}(t)=i\}|-}$ 
		 
		 $\displaystyle{\sum_{{\mathsf{CycTyp}\in \mathsf{CycTypS}(\mathsf{RobTyp},4)}}|\{t\in[x_\mathsf{CycTyp}]~|~\mathsf{CycAlloc}_{\mathsf{CycTyp}}(t)=i'\}||\leq 1}$.\label{def:DRAILPCon3}
	\end{enumerate}
Then, for every $\mathsf{CycTyp}\in\mathsf{CycTypS}$ and $i\in [x_\mathsf{CycTyp}]$, 

\noindent$\mathsf{ILPDerCycAlloc}(\{x_z~|~z\in \mathsf{VerTypS}\cup \mathsf{RobTypS}\cup \mathsf{CycTypS}\},\mathsf{CycTyp},i)=\mathsf{CycAlloc}_{\mathsf{CycTyp}}(i)$. 
\end{definition}

Whenever $\{x_z~|~z\in \mathsf{VerTypS}\cup \mathsf{RobTypS}\cup \mathsf{CycTypS}\}$ is clear from context, we refer to $\mathsf{ILPDerCycAlloc}(\{x_z~|~z\in \mathsf{VerTypS}\cup \mathsf{RobTypS}\cup \mathsf{CycTypS}\},\mathsf{CycTyp},i)$ as $\mathsf{ILPDerCycAlloc}($ $\mathsf{CycTyp},i)$.

Observe that we assume $x_z$, for every $z\in \mathsf{VerTypS}\cup \mathsf{RobTypS}\cup \mathsf{CycTypS}$, satisfy Equations~\ref{ilp:5} and~\ref{ilp:6}. So, for every $\mathsf{CycTyp}=(C,\mathsf{PaAlloc}_C,\mathsf{RobTyp})\in \mathsf{CycTypS}$ such that $x_\mathsf{CycTyp}\geq 1$, there exists $j\in[k]$ such that $\mathsf{ILPDerRobTyp}(j)=\mathsf{RobTyp}$. In addition, by Equations~\ref{ilp:5}, for every $\mathsf{RobTyp}=(\mathsf{CC},\mathsf{Alloc}_{\gr(\mathsf{CC})},\mathsf{NumOfCyc})\in \mathsf{RobTypS}$ and for every $2\leq j\leq 2|\vc|$, $j\neq 4$, $\displaystyle{\sum_{{\mathsf{CycTyp}\in \mathsf{CycTypS}(\mathsf{RobTyp},j)}}x_\mathsf{CycTyp}=N_j\cdot x_\mathsf{RobTyp}}$,

\noindent where $\mathsf{NumOfCyc}=(N_2,N_3,N_5,N_6,\ldots,N_{2|\vc|})$.
Thus, there exist some functions\\ $\mathsf{CycAlloc}_{\mathsf{CycTyp}}:[x_\mathsf{CycTyp}]\rightarrow [k]$, for every $\mathsf{CycTyp}\in \mathsf{CycTypS}$, as defined in Definition~\ref{def:DRAILP}. Therefore, $\mathsf{ILPDerCycAlloc}$ is well defined.

Lastly, we present the function $\mathsf{ILPToRobExp}$. This function gets $i\in[k]$ as input, and returns the multiset of edges, derived by the functions defined in this section, for the $i$-th robot. In particular, each robot gets the ``transformed'' $\mathsf{CC}$, where $\mathsf{ILPDerRobTyp}(i)=(\mathsf{CC},\mathsf{Alloc}_{\gr(\mathsf{CC})},\mathsf{NumOfCyc})$, and the ``transformed'' cycles, allocated to the $i$-th robot by Definition~\ref{def:DRAILP}.

\begin{definition} [{\bf $\mathsf{ILPToRobExp}$}]\label{def:ILPToREP}
	Let $G$ be a connected graph, let $\vi\in V(G)$, let $k,B\in \mathbb{N}$ and let $\vc$ be a vertex cover of $G$. Let $x_z$, for every $z\in \mathsf{VerTypS}\cup \mathsf{RobTypS}\cup \mathsf{CycTypS}$, be values that satisfy the inequalities of $\mathsf{Reduction}(G,\vi,k,B)$. 
Then, for every $1\leq i\leq k$, $\mathsf{ILPToRobExp}(\{x_z~|~z\in \mathsf{VerTypS}\cup \mathsf{RobTypS}\cup \mathsf{CycTypS}\},i)=\mathsf{ILPDerCCTransf}(i)\cup \{E(\mathsf{ILPDerCycTransf}(\mathsf{CycTyp},j))~|~\mathsf{CycTyp}\in\mathsf{CycTypS}, j\in[x_\mathsf{CycTyp}]$ such that $\mathsf{ILPDerCycAlloc}($ $\mathsf{CycTyp},j)=i\}$.
\end{definition}

Whenever $\{x_z~|~z\in \mathsf{VerTypS}\cup \mathsf{RobTypS}\cup \mathsf{CycTypS}\}$ is clear from context, we refer to $\mathsf{ILPToRobExp}(\{x_z~|~z\in \mathsf{VerTypS}\cup \mathsf{RobTypS}\cup \mathsf{CycTypS}\},i)$ as $\mathsf{ILPToRobExp}(i)$.

Towards the proof of Lemma~\ref{lem:fpt2}, we prove that the multisets defined in Definition~\ref{def:ILPToREP} satisfy the conditions of Lemma~\ref{obs:equivsol}:

\begin{lemma}\label{lem:ILPCon}
	Let $G$ be a connected graph, let $\vi\in V(G)$, let $k,B\in \mathbb{N}$ and let $\vc$ be a vertex cover of $G$. Let $x_z$, for every $z\in \mathsf{VerTypS}\cup \mathsf{RobTypS}\cup \mathsf{CycTypS}$, be values that satisfy the inequalities of $\mathsf{Reduction}(G,\vi,k,B)$. 
	Then, $\mathsf{ILPToRobExp}(1),\ldots, \mathsf{ILPToRobExp}(k)$ are multisets which satisfy the conditions of Lemma~\ref{obs:equivsol}.
\end{lemma}

For the sake of readability, we split Lemma~\ref{lem:ILPCon} and its proof into an observation and three lemmas. In Observation~\ref{obs:ilpCon} we prove that $\mathsf{ILPToRobExp}(i)$ is a multiset with elements from $E(G)$. In Lemmas~\ref{lem:ILPCon1} we prove that  Conditions~\ref{obs:equivsol1} and~\ref{obs:equivsol2} of Lemma~\ref{obs:equivsol} hold. In Lemmas~\ref{lem:ILPCon2} and~\ref{lem:ILPCon3} we prove that Conditions~\ref{obs:equivsol3} and~\ref{obs:equivsol4} of Lemma~\ref{obs:equivsol} hold, respectively.

\begin{observation}\label{obs:ilpCon}
Let $1\leq i\leq k$. Then, $\mathsf{ILPToRobExp}(i)$ is a multiset with elements from $E(G)$. 
\end{observation}

\begin{proof}
By Condition~\ref{lem:TranCCPro1} of Definition~\ref{lem:TranCCPro}, $\mathsf{ILPDerCCTransf}(i)$ is a multiset with elements from $E(G)$.
In addition, for every $\mathsf{CycTyp}\in\mathsf{CycTypS}$ and $j\in [x_\mathsf{CycTyp}]$, by Condition~\ref{lem:TranCPro1} of Definition~\ref{lem:TranCPro}, $\mathsf{ILPDerCycTransf}(\mathsf{CycTyp},j)$ is a cycle in $G$, so, indeed, $E(\mathsf{ILPDerCycTransf}(\mathsf{CycTyp},$ $j))$ is a multiset with elements from $E(G)$. 
\end{proof}

\begin{lemma}\label{lem:ILPCon1}
	Let $G$ be a connected graph, let $\vi\in V(G)$, let $k,B\in \mathbb{N}$ and let $\vc$ be a vertex cover of $G$. Let $x_z$, for every $z\in \mathsf{VerTypS}\cup \mathsf{RobTypS}\cup \mathsf{CycTypS}$, be values that satisfy the inequalities of $\mathsf{Reduction}(G,\vi,k,B)$. 
	Then, $\mathsf{ILPToRobExp}(1),\ldots, \mathsf{ILPToRobExp}(k)$ are multisets that satisfy Conditions~\ref{obs:equivsol1} and~\ref{obs:equivsol2} of Lemma~\ref{obs:equivsol}.
\end{lemma}

\begin{proof}
Let $1\leq i\leq k$. By Condition~\ref{lem:TranCCPro4} of Lemma~\ref{lem:TranCCPro}, $\vi\in V(\gr(\mathsf{ILPDerCCTransf}(i)))$, so $\vi\in V(\gr(\mathsf{ILPToRobExp}(i))$. Therefore, Condition~\ref{obs:equivsol1} of Lemma~\ref{obs:equivsol} is satisfied.

Let $\mathsf{ILPDerRobTyp}(i)=\mathsf{RobTyp}=(\mathsf{CC},\mathsf{Alloc}_{\gr(\mathsf{CC})},\mathsf{NumOfCyc})$.
We prove that $\gr($ $\mathsf{ILPToRobExp}(i))$ is connected.  By Condition~\ref{lem:TranCCPro4} of Lemma~\ref{lem:TranCCPro}, $\gr$ $(\mathsf{ILPDerCCTransf}(i))$ is connected; so, it is enough to show that, for every $v\in V(\gr(\mathsf{ILPToRobExp}(i)))\setminus V(\gr(\mathsf{ILPDerCCTransf}(i)))$, there exists a path from $v$ to some\\ $u\in V(\gr(\mathsf{ILPDerCCTransf}(i)))$. Since $\mathsf{ILPToRobExp}(i)=\mathsf{ILPDerCCTransf}(i)\cup$\\$ \{E(\mathsf{ILPDerCycTransf}(\mathsf{CycTyp},j))~|~\mathsf{CycTyp}\in\mathsf{CycTypS}, j\in[x_\mathsf{CycTyp}]$ such that\\ $\mathsf{ILPDerCycAlloc}(\mathsf{CycTyp},j)=i\}$, there exists $\mathsf{CycTyp}=(C,\mathsf{PaAlloc}_C,\mathsf{RobTyp})\in\mathsf{CycTypS}, j\in[x_\mathsf{CycTyp}]$ such that $\mathsf{ILPDerCycAlloc}(\mathsf{CycTyp},j)=i$ and $v\in V(\mathsf{ILPDerCycTransf}(\mathsf{CycTyp},j))$. Now, since $\mathsf{ILPDerCycAlloc}(\mathsf{CycTyp},j)=i$, from Condition~\ref{def:DRAILPCon1} of Definition~\ref{def:DRAILP},\\ $\mathsf{ILPToRobExp}(i)=\mathsf{RobTyp}$. So, from Definition~\ref{def:CT}, $V(\gr(\mathsf{CC}))\cap V(C)\cap \vc\neq \emptyset$. 

Now, by Condition~\ref{lem:TranCPro3} of Lemma~\ref{lem:TranCPro}, $V(C)\cap \vc=V(\mathsf{ILPDerCycTransf}(\mathsf{CycTyp},j))\cap \vc$, and from Condition~\ref{lem:TranCCPro3} of Lemma~\ref{lem:TranCCPro}, $V(\gr(\mathsf{CC}))\cap \vc= V(\gr(\mathsf{ILPToRobExp}(i))\cap \vc$. This implies that $V(\gr(\mathsf{ILPToRobExp}(i))\cap \vc\cap V(\gr(\mathsf{CC}))\neq \emptyset$, so let $u\in V(\gr(\mathsf{ILPToRobExp}(i))\cap \vc\cap V(\gr(\mathsf{CC}))$. In addition, from Condition~\ref{lem:TranCPro1} of Lemma~\ref{lem:TranCPro}, $\mathsf{ILPDerCycTransf}$ is a cycle in $G$.\\ Thus, there exists a path from $v$ to some $u\in V(\gr(\mathsf{ILPDerCCTransf}(i)))$, and therefore, $\gr(\mathsf{ILPDerCCTransf}(i))$ is connected.

Now, we show that every $u\in V(\gr(\mathsf{ILPDerCCTransf}(i)))$ has even degree in $\gr($ $\mathsf{ILPDerCCTransf}(i))$.
From Condition~\ref{lem:TranCCPro4} of Lemma~\ref{lem:TranCCPro}, every vertex in\\ $\gr(\mathsf{ILPDerCCTransf}(i))$ has even degree, and from Condition~\ref{lem:TranCPro1} of Lemma~\ref{lem:TranCPro},\\ $\mathsf{ILPDerCycTransf}(\mathsf{CycTyp},i)$ is a cycle in $G$.\\ Therefore, every $u\in V(\gr(\mathsf{ILPDerCCTransf}(i)))$ has even degree in\\ $\gr(\mathsf{ILPDerCCTransf}(i))$. Thus, Condition~\ref{obs:equivsol2} of Lemma~\ref{obs:equivsol} is satisfied.
\end{proof}

\begin{lemma}\label{lem:ILPCon2}
	Let $G$ be a connected graph, let $\vi\in V(G)$, let $k,B\in \mathbb{N}$ and let $\vc$ be a vertex cover of $G$. Let $x_z$, for every $z\in \mathsf{VerTypS}\cup \mathsf{RobTypS}\cup \mathsf{CycTypS}$, be values that satisfy the inequalities of $\mathsf{Reduction}(G,\vi,k,B)$. 
	Then, $\mathsf{ILPToRobExp}(1),\ldots, \mathsf{ILPToRobExp}(k)$ are multisets that satisfy Condition~\ref{obs:equivsol3} of Lemma~\ref{obs:equivsol}.
\end{lemma}

\begin{proof}
We show that for every $\{u,v\}\in E(G)$, there exists $1\leq i\leq k$ such that $\{u,v\}\in \mathsf{ILPToRobExp}(i)$. Let $\{u,v\}\in E$. We have the following two cases:

\smallskip\noindent{\bf Case 1: $u,v\in \vc$.}
By Equation~\ref{ilp:4}, 

\noindent$\ds{\sum_{{\mathsf{CycTyp}\in \mathsf{CycTypS}(\{u,v\})}}x_\mathsf{CycTyp}+\sum_{{\mathsf{RobTyp}\in \mathsf{RobTypS}(\{u,v\})}}}$ $x_\mathsf{RobTyp}\geq 1$.\\ If $\ds{\sum_{{\mathsf{CycTyp}\in \mathsf{CycTypS}(\{u,v\})}}x_\mathsf{CycTyp}\geq 1}$, then there exists $\mathsf{CycTyp}\in \mathsf{CycTypS}$ such that $x_\mathsf{CycTyp}\geq 1$ and $\mathsf{CycTyp}\in \mathsf{CycTypS}(\{u,v\})$.\\ Therefore, there exists $1\leq i\leq k$ such that $\mathsf{ILPDerCycAlloc}(\mathsf{CycTyp},1)=i$. Thus, $\{u,v\}\in E(\mathsf{ILPDerCycTransf}(\mathsf{CycTyp},1))\subseteq \mathsf{ILPToRobExp}(i)$.

Otherwise, $\ds{\sum_{{\mathsf{CycTyp}\in \mathsf{CycTypS}(\{u,v\})}}x_\mathsf{CycTyp}=0}$, so $\ds{\sum_{{\mathsf{RobTyp}\in \mathsf{RobTypS}(\{u,v\})}}}$ $\ds{x_\mathsf{RobTyp}\geq 1}$. Therefore, there exists $1\leq i\leq k$ and $\mathsf{RobTyp}=(\mathsf{CC},\mathsf{Alloc}_{\gr(\mathsf{CC})},\mathsf{NumOfCyc})\in \mathsf{RobTypS}(\{u,v\})$, such that $\mathsf{ILPDerRobTyp}(i)=\mathsf{RobTyp}$. Thus, $\{u,v\}\in \mathsf{CC}$, so, $\{u,v\}\in \mathsf{ILPDerCCTransf}(i)\subseteq \mathsf{ILPToRobExp}(i)$, and therefore $\{u,v\}\in \mathsf{ILPToRobExp}(i)$.  

\smallskip\noindent{\bf Case 2: $u\in u^*\in\ind, v\in \vc$.} Therefore, there exists $\mathsf{VerTyp}=(u^*,\U)\in\mathsf{VerTypS}_{u^*}$ such that $\mathsf{ILPDerVerTyp}(u)=\mathsf{VerTyp}$. By Definition~\ref{def:VTD}, there exists $\W\in \U$ such that $v\in \W$. In addition, there exists $\mathsf{Sub}\in \mathsf{ILPDerSubToAlloc}(\mathsf{VerTyp},$ $\W)$ such that $\mathsf{ILPDerSubAlloc}(\mathsf{VerTyp},\W,\mathsf{Sub})=u$. We have the following two subcases:

\smallskip\noindent{\bf Case 2.1: $\mathsf{Sub}=(\mathsf{CycTyp},j,t)$ for some $\mathsf{CycTyp}=(C,\mathsf{PaAlloc}_C,\mathsf{RobTyp})\in \mathsf{CycTypS}$, where $\mathsf{CycTyp}\in \mathsf{CycTypS}(\mathsf{VerTyp},\W,r)$, $1\leq j\leq x_\mathsf{CycTyp}, 1\leq t\leq r$.} Observe that, in this case $\W=\{v,v'\}$ for some $v\in V(G)$, and $\{\{u^*,v\},\{u^*,v'\}\}\in \mathsf{EdgePairs}(C)$. So, let $\mathsf{SetOfSub}(u^*,\W,\mathsf{VerTyp})=\{\{\{u^*,v\},\{u^*,v'\}\}\in \mathsf{EdgePairs}(C)~|~\mathsf{PaAlloc}_{C}=\mathsf{VerTyp}\}$ and let $\mathsf{Permut}_{u^*,\W,\mathsf{VerTyp}}:\mathsf{SetOfSub}(u^*,\W,\mathsf{VerTyp})\rightarrow [|\mathsf{SetOfSub}(u^*,\W,$\\$\mathsf{VerTyp})|]$ be the permutation defined in Definition~\ref{def:DTCILP}. Notice that $|\mathsf{SetOfSub}(u^*,\W,$\\$\mathsf{VerTyp})|=r$. Thus, there exists $(\{u^*,v\},\{u^*,v'\})\in \mathsf{SetOfSub}(u^*,\W,\mathsf{VerTyp})$ such that $\mathsf{Permut}_{u^*,\W,\mathsf{VerTyp}}$ $((\{u^*,v\},\{u^*,v'\}))=t$. Then, by Definition~\ref{def:DTCILP}, there exists $u^*\in V(C)$, where $v$ and $v'$ are the vertices come before and after $u^*$ in $C$, respectively, that is replaced by $\mathsf{ILPDerSubAlloc}(\mathsf{VerTyp},\W,$ $(\mathsf{CycTyp},j,$ $\mathsf{Permut}_{u^*,\W,\mathsf{VerTyp}}(u^*,(\{u^*,v\},$ $\{u^*,v'\}))=\mathsf{ILPDerSubAlloc}(\mathsf{VerTyp},\W,$ $(\mathsf{CycTyp},j,t))=u$.\\ So, $\{u,v\}\in E(\mathsf{ILPDerCycTransf}(\mathsf{CycTyp},j))$. Now, by Definition~\ref{def:DRAILP}, there exists $1\leq i\leq k$ such that $\mathsf{ILPDerCycAlloc}(\mathsf{CycTyp},j)=i$. Thus, by Definition~\ref{def:ILPToREP},\\ $E(\mathsf{ILPDerCycTransf}(\mathsf{CycTyp},j))\subseteq \mathsf{ILPToRobExp}(i)$, and therefore, $\{u,v\}\in \mathsf{ILPToRobExp}(i)$.  

\smallskip\noindent{\bf Case 2.2: $\mathsf{Sub}=(i,t)$, where $\mathsf{ILPDerRobTyp}(i)\in \mathsf{RobTypS}(\mathsf{VerTyp},\W,r)$, $1\leq r\leq 2^{|\vc|}+|\vc|^2, 1\leq i\leq k$ and $1\leq t\leq r$.}  Let $\mathsf{SetOfSub}(u^*,\W,\mathsf{VerTyp})=\{(u^*_j,\W)\in \mathsf{NeiOfInd}(\gr(\mathsf{CC})~|~\mathsf{Alloc}_{\gr(\mathsf{CC})}$ $((u^*_j,\W))=\mathsf{VerTyp}\}$ and let $\mathsf{Permut}_{u^*,\W,\mathsf{VerTyp}}:\mathsf{SetOfSub}(u^*,\W,\mathsf{VerTyp})\rightarrow [|\mathsf{SetOfSub}(u^*,\W,\mathsf{VerTyp})|]$ be the permutation defined in Definition~\ref{def:DTCCILP}. Notice that $|\mathsf{SetOfSub}(u^*,\W,\mathsf{VerTyp})|=r$. Thus, there exists $(u^*_j,\W)\in \mathsf{SetOfSub}(u^*,\W,\mathsf{VerTyp})$ such that $\mathsf{Permut}_{u^*,\W',\mathsf{VerTyp}}(u^*_j,\W)=t$. Therefore, by Definition~\ref{def:DTCCILP}, $u^*_j$ is replaced by $\mathsf{ILPDerSubAlloc}(\mathsf{VerTyp},\W,$\\$(i,\mathsf{Permut}_{u^*,\W,\mathsf{VerTyp}}(u^*_j,\W))=\mathsf{ILPDerSubAlloc}(\mathsf{VerTyp},\W,(i,t))=u$. Observe that since $(u^*_j,\W)\in \mathsf{NeiOfInd}(\gr(\mathsf{CC})$, by Definition~\ref{def:SOP}, it follows that $\widehat{\mathsf{N}}_{\gr(\mathsf{CC})}($ $u^*_j)=\W$. So, $\{u,v\}\in \mathsf{ILPDerCCTransf}(i)$. Therefore, by Definition~\ref{def:ILPToREP},\\ $\mathsf{ILPDerCCTransf}(i)\subseteq \mathsf{ILPToRobExp}(i)$, and therefore, $\{u,v\}\in \mathsf{ILPToRobExp}(i)$.  

Therefore, we proved that for every $\{u,v\}\in E(G)$, there exists $1\leq i\leq k$ such that $\{u,v\}\in \mathsf{ILPToRobExp}(i)$, so, Condition~\ref{obs:equivsol3} of Lemma~\ref{obs:equivsol} is satisfied.
\end{proof}

\begin{lemma}\label{lem:ILPCon3}
	Let $G$ be a connected graph, let $\vi\in V(G)$, let $k,B\in \mathbb{N}$ and let $\vc$ be a vertex cover of $G$. Let $x_z$, for every $z\in \mathsf{VerTypS}\cup \mathsf{RobTypS}\cup \mathsf{CycTypS}$, be values that satisfy the inequalities of $\mathsf{Reduction}(G,\vi,k,B)$. 
	Then, $\mathsf{ILPToRobExp}(1),\ldots, \mathsf{ILPToRobExp}(k)$ are multisets that satisfy Condition~\ref{obs:equivsol4} of Lemma~\ref{obs:equivsol}.
\end{lemma}

\begin{proof}
We show that for every $i\in [k]$, $|\mathsf{ILPToRobExp}(i)|\leq B$. Let $i\in [k]$ and let $\mathsf{ILPDerRobTyp}(i)=(\mathsf{CC},\mathsf{Alloc}_{\gr(\mathsf{CC})},\mathsf{NumOfCyc})$. By Definition~\ref{def:ILPToREP}, $\mathsf{ILPToRobExp}(i)=\mathsf{ILPDerCCTransf}(i)\cup \{E(\mathsf{ILPDerCycTransf}(\mathsf{CycTyp},j))~|~\mathsf{CycTyp}\in\mathsf{CycTypS}, j\in[x_\mathsf{CycTyp}]$ such that $\mathsf{ILPDerCycAlloc}(\mathsf{CycTyp},j)=i\}$. By Condition~\ref{lem:TranCCPro2} of Lemma~\ref{lem:TranCCPro}, $|\mathsf{ILPDerCCTransf}(i)|=|\mathsf{CC}|$ and by Condition~\ref{lem:TranCPro2} of Lemma~\ref{lem:TranCPro}, for every $\mathsf{CycTyp}=(C,\mathsf{PaAlloc}_C,\mathsf{RobTyp})\in \mathsf{CycTypS}$ and $j\in[x_\mathsf{CycTyp}]$, $|\mathsf{ILPDerCycTransf}(\mathsf{CycTyp},j)|=|C|$. Thus,

\noindent$\ds{|\mathsf{ILPToRobExp}(i)|=|\mathsf{CC}|+\sum_{(C,\mathsf{PaAlloc}_C,\mathsf{RobTyp})\in \mathsf{CycTypS},j\in[x_\mathsf{CycTyp}]}\sum_{\mathsf{ILPDerCycAlloc}(\mathsf{CycTyp},j)=i}|C|}$. 

Recall that, for every $\mathsf{RobTyp}\in \mathsf{RobTypS}$ and for every $2\leq j\leq 2|\vc|$, $\mathsf{CycTypS}(\mathsf{RobTyp},j)$ $=\{\mathsf{CycTyp}=(C,\mathsf{PaAlloc}_C,\mathsf{RobTyp})\in \mathsf{CycTypS}~|~|C|=j\}$. Now, by Definition~\ref{def:CT}, for every $(C,\mathsf{PaAlloc}_C,\mathsf{RobTyp})\in \mathsf{CycTypS}$, $|C|\leq 2|\vc|$. Moreover, from Condition~\ref{def:DRAILPCon1} of Definition~\ref{def:DRAILP}, for every $(C,\mathsf{PaAlloc}_C,\mathsf{RobTyp})\in \mathsf{CycTypS}$ and $j\in[x_\mathsf{CycTyp}]$ such that $\mathsf{ILPDerCycAlloc}(\mathsf{CycTyp},j)=i$, $\mathsf{ILPDerRobTyp}(i)=\mathsf{RobTyp}$. Therefore,

 \noindent$\ds{|\mathsf{ILPToRobExp}(i)|=|\mathsf{CC}|+}$\\
 	$\ds{\sum_{2\leq j\leq 2|\vc|}\sum_{\mathsf{CycTyp}\in \mathsf{CycTypS}(\mathsf{RobTyp},j)}}$ $\ds{|\{1\leq t\leq x_\mathsf{CycTyp}~|~\mathsf{ILPDerCycAlloc}(\mathsf{CycTyp},t)=i\}|\cdot j}$.\\ So, $\ds{|\mathsf{ILPToRobExp}(i)|=|\mathsf{CC}|+}$
 
 \noindent$\ds{\sum_{2\leq j\leq 2|\vc|,j\neq4}\sum_{\mathsf{CycTyp}\in \mathsf{CycTypS}(\mathsf{RobTyp},j)}|\{1\leq t\leq x_\mathsf{CycTyp}~|~\mathsf{ILPDerCycAlloc}(\mathsf{CycTyp},t)=i\}|\cdot j+}$
 	
 	\noindent$\ds{\sum_{\mathsf{CycTyp}\in \mathsf{CycTypS}(\mathsf{RobTyp},4)}|\{1\leq t\leq x_\mathsf{CycTyp}~|~\mathsf{ILPDerCycAlloc}(\mathsf{CycTyp},t)=i\}|\cdot 4}$. 

Now, by Condition~\ref{def:DRAILPCon2} of Definition~\ref{def:DRAILP}, for every $2\leq j\leq 2|\vc|$, $j\neq 4$, 

\noindent$\ds{\sum_{{\mathsf{CycTyp}\in \mathsf{CycTypS}(\mathsf{RobTyp},j)}}|\{t\in[x_\mathsf{CycTyp}]~|~\mathsf{ILPDerCycAlloc}(\mathsf{CycTyp},t)=i\}|=N_j}$. Therefore, 

\noindent$\ds{|\mathsf{ILPToRobExp}(i)|=|\mathsf{CC}|+\sum_{2\leq j\leq 2|\vc|,j\neq4}N_j\cdot j +}$ 

\noindent$\ds{\sum_{\mathsf{CycTyp}\in \mathsf{CycTypS}(\mathsf{RobTyp},4)}|\{1\leq t\leq x_\mathsf{CycTyp}~|~\mathsf{ILPDerCycAlloc}(\mathsf{CycTyp},t)=i\}|\cdot 4}$. Recall that

\noindent$\ds{\mathsf{Bud}(\mathsf{RobTyp})=|\mathsf{CC}|+\sum_{2\leq j\leq 2|\vc|,j\neq 4}N_j\cdot j}$, so

\noindent $\ds{|\mathsf{ILPToRobExp}(i)|=\mathsf{Bud}(\mathsf{RobTyp})+}$\\ $\ds{\sum_{\mathsf{CycTyp}\in \mathsf{CycTypS}(\mathsf{RobTyp},4)}|\{1\leq t\leq x_\mathsf{CycTyp}~|~\mathsf{ILPDerCycAlloc}(\mathsf{CycTyp},t)=i\}|\cdot 4}$. 

Now, by Condition~\ref{def:DRAILPCon3} of Definition~\ref{def:DRAILP}, for every $a,b\in [k]$ such that $\mathsf{ILPDerRobTyp}(a)=\mathsf{ILPDerRobTyp}(b)=\mathsf{RobTyp}$, 

\noindent$\ds{|\sum_{{\mathsf{CycTyp}\in \mathsf{CycTypS}(\mathsf{RobTyp},4)}}}$ $\ds{|\{t\in[x_\mathsf{CycTyp}]~|~\mathsf{ILPDerCycAlloc}(\mathsf{CycTyp},t)=a\}|-}$
	
	$\ds{\sum_{{\mathsf{CycTyp}\in \mathsf{CycTypS}(\mathsf{RobTyp},4)}}|\{t\in[x_\mathsf{CycTyp}]~|~\mathsf{ILPDerCycAlloc}(\mathsf{CycTyp},t)=b\}||\leq 1}$. This implies that $\ds{\sum_{\mathsf{CycTyp}\in \mathsf{CycTypS}(\mathsf{RobTyp},4)}|\{1\leq t\leq x_\mathsf{CycTyp}~|~\mathsf{ILPDerCycAlloc}(\mathsf{CycTyp},t)=i\}|\leq}$ 
	
	\noindent$\ds{\lceil\frac{\sum_{\mathsf{CycTyp}\in \mathsf{CycTypS}(\mathsf{RobTyp},4)}x_{\mathsf{CycTyp}}}{|\{r\in [k]~|~\mathsf{ILPDerRobTyp}(r)=\mathsf{RobTyp}\}|}\rceil}$. By Definition~\ref{def:DRTILP}, 
	
	\noindent$|\{r\in [k]~|~\mathsf{ILPDerRobTyp}(r)=\mathsf{RobTyp}\}|=x_\mathsf{RobTyp}$, 
	
	\noindent so $\ds{\sum_{\mathsf{CycTyp}\in \mathsf{CycTypS}(\mathsf{RobTyp},4)}|\{1\leq t\leq x_\mathsf{CycTyp}~|~}$ $\ds{\mathsf{ILPDerCycAlloc}(\mathsf{CycTyp},t)=i\}|\leq}$\\ 
	$\ds{\lceil\frac{\sum_{\mathsf{CycTyp}\in \mathsf{CycTypS}(\mathsf{RobTyp},4)}x_{\mathsf{CycTyp}}}{x_\mathsf{RobTyp}}\rceil}$. Therefore, $\ds{|\mathsf{ILPToRobExp}(i)|\leq \mathsf{Bud}(\mathsf{RobTyp})+}$\\ $\ds{\lceil\frac{\sum_{\mathsf{CycTyp}\in \mathsf{CycTypS}(\mathsf{RobTyp},4)}x_{\mathsf{CycTyp}}}{x_\mathsf{RobTyp}}\rceil\cdot 4}$. 

Now, by Equation~\ref{ilp:6}, $\ds{\sum_{{\mathsf{CycTyp}\in \mathsf{CycTypS}(\mathsf{RobTyp},4)}}4\cdot x_\mathsf{CycTyp}\leq x_\mathsf{RobTyp}\cdot \mathsf{CycBud}(\mathsf{RobTyp})}$. Recall that $\mathsf{CycBud}(\mathsf{RobTyp})=\lfloor(B-\mathsf{Bud}(\mathsf{RobTyp}))\cdot \frac{1}{4}\rfloor\cdot 4$. \\So, $\ds{|\mathsf{ILPToRobExp}(i)|\leq \mathsf{Bud}(\mathsf{RobTyp})+\lceil\frac{\frac{1}{4}\cdot x_\mathsf{RobTyp}\cdot \lfloor(B-\mathsf{Bud}(\mathsf{RobTyp}))\cdot \frac{1}{4}\rfloor\cdot 4}{x_\mathsf{RobTyp}}\rceil\cdot 4=}$
 	
 	\noindent$\ds{\mathsf{Bud}(\mathsf{RobTyp})+\lceil\lfloor(B-\mathsf{Bud}(\mathsf{RobTyp}))\cdot \frac{1}{4}\rfloor\rceil\cdot 4=}$\\ $\ds{\mathsf{Bud}(\mathsf{RobTyp})+\lfloor(B-\mathsf{Bud}(\mathsf{RobTyp}))\cdot \frac{1}{4}\rfloor\cdot 4}$. Thus, $|\mathsf{ILPToRobExp}(i)|\leq \mathsf{Bud}(\mathsf{RobTyp})+(B-\mathsf{Bud}(\mathsf{RobTyp}))\cdot \frac{1}{4}\cdot 4=B$. Therefore, Condition~\ref{obs:equivsol4} of Lemma~\ref{obs:equivsol} is satisfied.
\end{proof}

Observation~\ref{obs:ilpCon} and Lemmas~\ref{lem:ILPCon1}--\ref{lem:ILPCon3} conclude the correctness of Lemma~\ref{lem:ILPCon}.
Now, we invoke Lemma~\ref{lem:ILPCon} in order to prove Lemma~\ref{lem:fpt2}:

\begin{proof}
Assume that $\mathsf{Reduction}(G,\vi,k,B)$ is a yes-instance of the {\sc Integer Linear Programming}. Let $x_z$, for every $z\in \mathsf{VerTypS}\cup \mathsf{RobTypS}\cup \mathsf{CycTypS}$, be values that satisfy the inequalities of $\mathsf{Reduction}(G,\vi,k,B)$. By Lemma~\ref{lem:ILPCon}, $\mathsf{ILPToRobExp}(1),\ldots, \mathsf{ILPToRobExp}(k)$ are multisets which satisfy the conditions of Lemma~\ref{obs:equivsol}. Therefore, by Lemma~\ref{obs:equivsol}, $(G,\vi,k,B)$ is a yes-instance of the \cg problem.
\end{proof}


Lemmas~\ref{lem:fpt1} and~\ref{lem:fpt2} conclude the correctness of Lemma~\ref{lem:fpt}. 

\subsection{Runtime Analysis}\label{sec:runtime}

In this section, we analyze the runtime of the reduction and also the size of the input for ILP. First, we give an upper bound for $|\mathsf{VerTypS}|$, $|\mathsf{RobTypS}|$ and $|\mathsf{CycTypS}|$, and analyze the runtime of computing these sets.


\begin{lemma}\label{lem:fptr1}
	Let $G$ be a connected graph, let $\vi\in V(G)$ and let $k,B\in \mathbb{N}$. Then:
	\begin{enumerate}
		\item $|\mathsf{VerTypS}|\leq 2^{2^{\OO(|\vc|)}}$.
		\item $|\mathsf{RobTypS}|\leq 2^{2^{\OO(|\vc|)}}$.
		\item $|\mathsf{CycTypS}|\leq 2^{2^{\OO(|\vc|)}}$.
	\end{enumerate}
In addition, creating $\mathsf{VerTypS}$, $\mathsf{RobTypS}$ and $\mathsf{CycTypS}$ take $2^{2^{\OO(|\vc|)}}+\OO(|\vc|\cdot|V(G)|+|E(G)|)$ time.
\end{lemma}

\begin{proof}
Recall that a vertex type is a pair $(u^*,\U)$ satisfies the conditions of Definition~\ref{def:VT}. In particular, for every vertex type $(u^*,\U)$, $u^*\in \ind$. We do not need to know $\ind$ explicitly, it is enough to determine, for every subset $U$ of $\vc$, if there exists $u\in \indd=V(G)\setminus \vc$ such that $\mathsf{N}_G(u)=U$. For every such $U$, there exists $u^*\in \ind$ such that every $u\in u^*$, $\mathsf{N}_G(u)=U$. In particular, for every subset $U$ of $\vc$, we count the number of vertices $u\in V(G)\setminus \vc$ (which is potentially zero) such that $\mathsf{N}_G(u)=U$ in $\OO(2^{|\vc|}+|\vc|\cdot|V(G)|+|E(G)|)$ time: we create an array of size $2^{|\vc|}$, where each cell corresponds to some subset $U$ of $\vc$. Initially, we put $0$ in each cell. For each vertex $u\in \indd$, we calculate $\mathsf{N}_G(u)$, and search for its corresponding cell, using binary search, in $|\vc|$ time. Then, we add one to this cell. 
In addition, for every $u^*\in \ind$, $|2^{\mathsf{N}_{G^*}(u^*)\times \mathsf{2}}|\leq 2^{2|\vc|}$, so $\U$ has at most $2^{2^{2|\vc|}}$ options. Thus, $|\mathsf{VerTypS}|\leq 2^{|\vc|}\cdot 2^{2^{2|\vc|}}=2^{2^{\OO(|\vc|)}}$, and creating $\mathsf{VerTypS}$ takes $2^{2^{\OO(|\vc|)}}+\OO(|\vc|\cdot|V(G)|+|E(G)|)$ time.

Second, recall that a robot type is a triple $(\mathsf{CC},\mathsf{Alloc}_{\gr(\mathsf{CC})},\mathsf{NumOfCyc})$ satisfies the conditions of Definition~\ref{def:RT}. Notice that $\overline{G}$ (see Definition~\ref{def:GbarGraph}), has at most $|\vc|+2^{|\vc|}\cdot(2^{|\vc|}+|\vc|^2)$ vertices and at most $2|\vc|^2+2|\vc|\cdot(2^{|\vc|}+|\vc|)=2^{\OO(|\vc|)}$ edges. Thus, the number of options for choosing $\mathsf{CC}\subset E(\overline{G})$ is bounded by $2^{2^{\OO(|\vc|)}}$. Now, for each such $\mathsf{CC}\subset E(\overline{G})$, $|\mathsf{NeiOfInd}(\gr(\mathsf{CC}))|\leq 2^{|\vc|}\cdot(2^{|\vc|}+|\vc|^2)=2^{\OO(|\vc|)}$ (see Definition~\ref{def:SOP}). Therefore, since $|\mathsf{VerTypS}|\leq 2^{2^{\OO(|\vc|)}}$, the number of options for $\mathsf{Alloc}_{\gr(\mathsf{CC})}$ is bounded by $(2^{2^{\OO(|\vc|)}})^{2^{\OO(|\vc|)}}=2^{2^{\OO(|\vc|)}\cdot {2^{\OO(|\vc|)}}}={2^{2^{\OO(|\vc|)}}}$ (see Definition~\ref{def:SOPP}). The number of options for $\mathsf{NumOfCyc}$ (see Definition~\ref{def:RT}) is bounded by $(2|\vc|^2)^{2|\vc|}=|\vc|^{\OO(|\vc|)}$. Thus, $|\mathsf{RobTypS}|\leq 2^{2^{\OO(|\vc|)}}\cdot 2^{2^{\OO(|\vc|)}}\cdot |\vc|^{\OO(|\vc|)}=2^{2^{\OO(|\vc|)}}$.In addition, creating $\mathsf{RobTypS}$ takes $2^{2^{\OO(|\vc|)}}$ time, once $\mathsf{VerTypS}$ is computed.

Next, recall that a cycle type is a triple $(C,\mathsf{PaAlloc}_C,\mathsf{RobTyp})$ satisfies the conditions of Definition~\ref{def:CT}. Observe that the number of vertices in $G^*$ is bounded by $|\vc|+2^{|\vc|}$, and the number of cycles of length at most $2|\vc|$ in $G^*$ is bounded by $\sum_{2\leq i\leq 2|\vc|}(|\vc|+2^{|\vc|})^i\leq 2|\vc|(|\vc|+2^{|\vc|})^{2|\vc|}=2^{\OO|\vc|^2}$. Now, for each $C\in \mathsf{Cyc}_{G^*}$, $|\mathsf{EdgePairs}(C)|\leq |\vc|$ (see Definition~\ref{def:POC}). Thus, the number of options for $\mathsf{PaAlloc}_C$ is bounded by $(2^{2^{\OO(|\vc|)}})^{|\vc|}=2^{2^{\OO(|\vc|)}}$ (see Definition~\ref{def:PPOC}). Therefore, $|\mathsf{CycTypS}|\leq 2^{\OO|\vc|^2}\cdot 2^{2^{\OO(|\vc|)}}\cdot 2^{2^{\OO(|\vc|)}}=2^{2^{\OO(|\vc|)}}$. In addition, creating $\mathsf{CycTypS}$ takes $2^{2^{\OO(|\vc|)}}$ time, once $\mathsf{VerTypS}$ and $\mathsf{RobTypS}$ are computed.
\end{proof}

In the next lemma, we analyze the runtime and the size of the reduction.

\begin{lemma}\label{lem:fptr2}
	Let $G$ be a connected graph, let $\vi\in V(G)$ and let $k,B\in \mathbb{N}$. Then, $\mathsf{Reduction}(G,\vi,k,B)$ runs in time $2^{2^{\OO(|\vc|)}}+\OO(|\vc|\cdot|V(G)|+|E(G)|)$, and returns an instance for the ILP problem of size at most $2^{2^{\OO(|\vc|)}}\cdot \log(k+B+|V(G)|)$ with $2^{2^{\OO(|\vc|)}}$ variables.
\end{lemma}

\begin{proof}
First, we create $\mathsf{VerTypS}$, $\mathsf{RobTypS}$ and $\mathsf{CycTypS}$ in $2^{2^{\OO(|\vc|)}}+\OO(|\vc|\cdot|V(G)|+|E(G)|)$ time (see Lemma~\ref{lem:fptr1}).
Now, we give an upper bound for the number of equations in $\mathsf{Reduction}(G,\vi,k,B)$, and the runtime of creating the equations. There is one equation in Equation~\ref{ilp:1}, that takes $\OO(|\mathsf{RobTypS}|)$ time to create, so, by Lemma~\ref{lem:fptr1}, $2^{2^{\OO(|\vc|)}}$ time.

There is one equation for each $u^*\in \ind$ in Equation~\ref{ilp:2}, so there are at most $2^{|\vc|}$ equations. Computing $|u^*|$ takes $\OO(2^{|\vc|}+|\vc|\cdot|V(G)|+|E(G)|)$ time (described in the proof of Lemma~\ref{lem:fptr1}). The rest of the computation of creating these equations, takes $\OO(|\mathsf{VerTypS}|)$ time, so, by Lemma~\ref{lem:fptr1}, it takes $2^{2^{\OO(|\vc|)}}$ time. 

There is one equation for every $\mathsf{VerTyp}=(u^*,\U)\in \mathsf{VerTypS}$, and every $\W\in \U$ in Equation~\ref{ilp:3}, so there are at most $|\mathsf{VerTypS}|\cdot 2^{2^{|\vc|}}= 2^{2^{\OO(|\vc|)}}\cdot 2^{2^{|\vc|}}= 2^{2^{\OO(|\vc|)}}$ equations. Now, we can create  the sets $\mathsf{CycTypS}(\mathsf{VerTyp},\W,j)$, for every $\mathsf{VerTyp}=(u^*,\U)\in \mathsf{VerTypS}$, every $\W\in \U$ and $1\leq j\leq |\vc|$ as follows. For every $\mathsf{CycTyp}=(C,\mathsf{PaAlloc}_C,\mathsf{RobTyp})\in \mathsf{CycTypS}$ we initial a table $\{\W\in 2^{\vc\times \mathsf{2}}~|~|\W|=2\}\times \mathsf{VerTyp}$ with $0$ in every cell. Then, we go over each $\{\{v_{i-1},v_{i}\},\{v_{i},v_{i+1}\}\}\in\mathsf{EdgePairs}(C)$, and add $1$ to $(\{v_{i-1},v_{i+1}\},\mathsf{PaAlloc}_C(\{v_{i-1},v_{i+1}\}))$. Then, for every cell $(\W,\mathsf{VerTyp})$ different than $0$, we add $\mathsf{CycTyp}$ to $\mathsf{CycTypS}(\mathsf{VerTyp},\W,j)$. It is easy to see the correctness of this process. Observe that it takes at most $|\mathsf{CycTypS}|\cdot |\mathsf{VerTypS}|\cdot 2^{2|\vc|}\cdot |\vc|=2^{2^{\OO(|\vc|)}}\cdot 2^{2^{\OO(|\vc|)}}\cdot 2^{2|\vc|}\cdot |\vc|=2^{2^{\OO(|\vc|)}}$ time.
 Similarly, we can create $\mathsf{RobTypS}(\mathsf{VerTyp},\W,j)$, for every $\mathsf{VerTyp}=(u^*,\U)\in \mathsf{VerTypS}$, every $\W\in \U$ and $1\leq j\leq 2^{|\vc|}+|\vc|^2$, in at most $|\mathsf{RobTypS}|\cdot |\mathsf{VerTypS}|\cdot (2^{|\vc|}+|\vc|^2)\cdot 2^{2|\vc|}=2^{2^{\OO(|\vc|)}}\cdot 2^{2^{\OO(|\vc|)}}\cdot (2^{|\vc|}+|\vc|^2)\cdot 2^{2|\vc|}=2^{2^{\OO(|\vc|)}}$. Now, each equation takes at most $|\mathsf{CycTypS}|\cdot |\mathsf{RobTypS}|=2^{2^{\OO(|\vc|)}}\cdot 2^{2^{\OO(|\vc|)}}$ time. So, to sum up, there are at most $2^{2^{\OO(|\vc|)}}$ equations in Equation~\ref{ilp:3}, and they take at most $2^{2^{\OO(|\vc|)}}+2^{2^{\OO(|\vc|)}}+2^{2^{\OO(|\vc|)}}\cdot 2^{2^{\OO(|\vc|)}}=2^{2^{\OO(|\vc|)}}$ time.

There are $|\vc|$ equations in Equation~\ref{ilp:4}. Observe that we can create $\mathsf{CycTypS}(\{u,v\})$ for every $\{u,v\}\in E$ such that $u,v\in \vc$, in $|\mathsf{CycTypS}|\cdot |\vc|\cdot |\vc|=2^{2^{\OO(|\vc|)}}\cdot  |\vc|\cdot |\vc|=2^{2^{\OO(|\vc|)}}$ time. Similarly, we can create $\mathsf{RobTypS}(\{u,v\})$ for every $\{u,v\}\in E$ such that $u,v\in \vc$, in $|\mathsf{RobTypS}|\cdot |\vc|\cdot |\vc|=2^{2^{\OO(|\vc|)}}\cdot  |\vc|\cdot |\vc|=2^{2^{\OO(|\vc|)}}$ time. Now, observe that each equation takes at most $\OO (|\mathsf{CycTypS}|\cdot |\mathsf{RobTypS}|)=2^{2^{\OO(|\vc|)}}$ time. So, there are $|\vc|$ equations in Equation~\ref{ilp:4}, and they take $2^{2^{\OO(|\vc|)}}+2^{2^{\OO(|\vc|)}}+2^{2^{\OO(|\vc|)}}\cdot|\vc|=2^{2^{\OO(|\vc|)}}$ time.

There are at most $|\mathsf{RobTypS}|\cdot 2|\vc|=2^{2^{\OO(|\vc|)}}$ equations in Equation~\ref{ilp:5}. Observe that we can create $\mathsf{CycTypS}(\mathsf{RobTyp},j)$ for every $\mathsf{RobTyp}\in \mathsf{RobTypS}$ and for every $2\leq j\leq 2|\vc|$, in $|\mathsf{CycTypS}|\cdot \OO(|\vc|)\cdot |\mathsf{RobTypS}|=2^{2^{\OO(|\vc|)}}$ time. Now, observe that each equation takes $\OO |\mathsf{CycTypS}|=2^{2^{\OO(|\vc|)}}$ time. So, there are $2^{2^{\OO(|\vc|)}}$ equations in Equation~\ref{ilp:5}, and they take $2^{2^{\OO(|\vc|)}}+2^{2^{\OO(|\vc|)}}\cdot 2^{2^{\OO(|\vc|)}}=2^{2^{\OO(|\vc|)}}$ time.

Finally, there are $|\mathsf{RobTypS}|=2^{2^{\OO(|\vc|)}}$ equations in Equation~\ref{ilp:6}. Computing\\ $\mathsf{CycBud}(\mathsf{RobTyp})$ for every $\mathsf{RobTyp}\in \mathsf{RobTypS}$ takes at most $\OO(|\mathsf{RobTypS}|\cdot 2^{\OO(|\vc|)})=2^{2^{\OO(|\vc|)}}$ time. Each equation in Equation~\ref{ilp:6} takes $\OO(|\mathsf{CycTypS}|)=2^{2^{\OO(|\vc|)}}$ time. So, there are $2^{2^{\OO(|\vc|)}}$ equations in Equation~\ref{ilp:6}, and they take $2^{2^{\OO(|\vc|)}}+2^{2^{\OO(|\vc|)}}\cdot2^{2^{\OO(|\vc|)}}=2^{2^{\OO(|\vc|)}}$ time.

Therefore, we get that there are at most $2^{2^{\OO(|\vc|)}}$ variables and $2^{2^{\OO(|\vc|)}}$ equations in $\mathsf{Reduction}(G,\vi,k,B)$. In addition, notice that the coefficients of the equations are bounded by $\mathsf{max}\{B,k,|V(G)|\}$.

In summary, $\mathsf{Reduction}(G,\vi,k,B)$ works in time $2^{2^{\OO(|\vc|)}}+\OO(|\vc|\cdot|V(G)|+|E(G)|)$ and returns an instance for the ILP problem of size at most $2^{2^{\OO(|\vc|)}}\cdot 2^{2^{\OO(|\vc|)}}\cdot \log(k+B+|V(G)|)=2^{2^{\OO(|\vc|)}}\cdot \log(k+B+|V(G)|)$ and at most $|\mathsf{VerTypS}|+|\mathsf{RobTypS}|+|\mathsf{CycTypS}|\leq 2^{2^{\OO(|\vc|)}}$ variables.
\end{proof}

Now, we invoke Lemmas~\ref{lem:fpt} and~\ref{lem:fptr2}, in order to prove the following corollary:

\begin{corollary}\label{cor:fpt}
There exists an algorithm that solves \cg in $2^{2^{2^{\OO(\vn)}}}\cdot (\log(k+B+|V(G)|))^{\OO(1)}+\OO(\vn\cdot|V(G)|+|E(G)|)$ time.
\end{corollary}

\begin{proof}
Let $(G,\vi,k,B)$ be an instance of \cg. By~\cite{papadimitriou1998combinatorial}, there exists a $2$-approximation for computing a minimal vertex cover of an input in linear time. Then, we activate $\mathsf{Reduction}(G,\vi,k,B)$ with the vertex cover $\vc$ computed by this algorithm. Observe that Lemma~\ref{lem:fpt} concludes the correctness of the algorithm. By  Lemma~\ref{lem:fptr2}, $\mathsf{Reduction}(G,\vi,k,$ $B)$ runs in time $2^{2^{\OO(|\vc|)}}+\OO(|\vc|\cdot|V(G)|+|E(G)|)$, and returns an instance for the ILP problem of size at most $2^{2^{\OO(|\vc|)}}\cdot \log(k+B+|V(G)|)$ and with $2^{2^{\OO(|\vc|)}}$ variables. So, by Theorem~\ref{the:runningTimeILP}, $\mathsf{Reduction}(G,\vi,k,B)$ can be solved in time $(2{^{2^{\OO(|\vc|)}}})^{2^{2^{\OO(|\vc|)}}}\cdot 2^{2^{\OO(|\vc|)}}\cdot \log(k+B+|V(G)|)^{\OO(1)}=2^{2^{2^{\OO(|\vc|)}}}\cdot \log(k+B+|V(G)|)^{\OO(1)}$. Therefore, the total runtime of the algorithm is $2^{2^{\OO(|\vc|)}}+\OO(|\vc|\cdot|V(G)|+|E(G)|)+2^{2^{2^{\OO(|\vc|)}}}\cdot (\log(k+B+|V(G)|))^{\OO(1)}=2^{2^{2^{\OO(|\vc|)}}}\cdot (\log(k+B+|V(G)|))^{\OO(1)}+\OO(|\vc|\cdot|V(G)|+|E(G)|)$. Now, since $|\vc|\leq 2\cdot\vn$, we conclude the correctness of the corollary.
\end{proof}
	
Corollary~\ref{cor:fpt} concludes the correctness of Theorem~\ref{th:fptVc}.

\section{Approximation Algorithm with Additive Error of $\OO(\vn)$}\label{sec:AppAlg}

In this section, we prove the following result.

\approxTheorem*

%

To prove the above theorem, we first give an approximation algorithm that takes a vertex cover $\vc$ of the input graph as an additional parameter and returns a solution with an additive approximation of $4|\vc|$. 

\subsection{The Algorithm}
We now describe the stages of our approximation algorithm, whose pseudocode is given as Algorithm~\ref{alg:AppAlg}. Here $\vc$ is the input vertex cover.

\begin{algorithm}[!t]
	\SetKwInOut{Input}{Input}
	\SetKwInOut{Output}{Output}
	\medskip
	{\textbf{function} $\mathsf{ApproxAlg}$}$(\langle G=(V,E),k,\vi,\mathsf{VC}\rangle)$\;
	$\vc'\gets \vc \cup \{\vi\}$\; \label{alg1:L2}
		Make $\vc'$ connected\; \label{alg1:L3}
	$I\gets V(G)\setminus \vc'$\; \label{alg1:L4}
	$\widehat{E}_\mathsf{IND}\gets \{\{u,v\}\in E(G)~|~u\in I\}$\; \label{alg1:L5}
	\For {every $u\in I$ with odd degree in $G$}
	{Let $\{u,v\}\in E$\;
		$\widehat{E}_\mathsf{IND}\gets \widehat{E}_\mathsf{IND}\cup\{\{u,v\}\}$\;	
	}\label{alg1:L9}
	\For {every $1\leq i\leq k$\label{alg1:L10}}
	{$\widehat{E}_i\gets \emptyset$\;
	}
	\label{alg1:L13}
	\For {every $u\in I$}
	{
		\While{\label{alg1:L144}There exists $\{u,v\}\in \widehat{E}_\mathsf{IND}$}
		{
			Let $\{u,v\},\{u,v'\}\in \widehat{E}_\mathsf{IND}$\;
			Let $1\leq i\leq k$ with minimum $|\widehat{E}_i|$\; \label{alg1:L16}
			$\widehat{E}_i\gets \widehat{E}_i\cup \{\{u,v\},\{u,v'\}\}$\;
			$\widehat{E}_\mathsf{IND}\gets \widehat{E}_\mathsf{IND}\setminus \{\{u,v\},\{u,v'\}\}$\;
		}
	}\label{alg1:L20}
	\For {\label{alg1:L211} every $\{u,v\}\in E$ such that $u,v\in \vc'$}
	{
		Let $1\leq i\leq k$ with minimum $|\widehat{E}_i|$\; \label{alg1:L21}
		$\widehat{E}_i\gets \widehat{E}_i\cup \{\{u,v\}\}$\;
	}
	\label{alg1:L24}
	Let $T$ be a spanning tree of $G[\vc']$\; \label{alg1:L25}
	\For {every $1\leq i\leq k$}
	{
		$\widehat{E}_i\gets \widehat{E}_i\cup E(T)$\;  \label{alg1:L27}
	} \label{alg1:L28}
	\For  {\label{alg1:L29} every $1\leq i\leq k$}
	{
		$\mathsf{MakeVCEvenDeg}(T,\widehat{E}_i,\vc')$\;\label{alg1:L30}
	}\label{alg1:L31}
	\For {\label{alg1:L32} every $1\leq i\leq k$}
	{
		Find an Eulerian cycle $\mathsf{RC}_i$ in $\widehat{G}[\widehat{E}_i]$\;\label{alg1:L335}
	}\label{alg1:L33}
	\Return $\{\mathsf{RC}_i\}_{i=1}^k$\;
	\caption{$B+\OO(|\vc|)$ Approximation}
	\label{alg:AppAlg}
\end{algorithm}

\smallskip\noindent{\bf Lines \ref{alg1:L2}--\ref{alg1:L3}: Making the vertex cover connected and adding $\vi$.}
We make $G[\vc]$ connected by adding at most $|\vc|-1$ vertices to $\vc$: We begin with a partition $V_1,\ldots V_t$ of $\vc$ into connected components of $G[\vc]$. Then, until $V_1,\ldots V_t$ unite into one set we do the following. For every $v\in V(G)\setminus \vc$, if $v$ has neighbors in two different sets among $V_1,\ldots V_t$, then: (i) we add $v$ to $\vc$, and (ii) we unite the sets that have at least one neighbor of $v$. Since $G$ is connected it is easy to see that by the end of the processes $G[\vc]$ is connected. Moreover, observe that this takes $\OO(|V(G)|+|E(G)|)$ runtime, and we add at most $|\vc|-1$ vertices to the initial vertex cover. In addition, we also add $\vi$. The new vertex cover we obtained is denoted by $\vc'$. For an example, see Figure~\ref{fig:approx1B}.

\smallskip\noindent{\bf Lines \ref{alg1:L4}--\ref{alg1:L9}: Making the degree of each vertex in the independent set even.}
We ensure that every vertex $u$ from the independent set $I=V\setminus \vc'$ has even degree as follows. Initially, $\widehat{E}_\mathsf{IND}=\{\{u,v\}\in E~|~u\in I\}$. For every $u\in I$ with odd degree in $G$, we simply duplicate an arbitrary edge $\{u,v\}$ incident to $u$, that is, we add $\{u,v\}$ to the multiset $\widehat{E}_\mathsf{IND}$ (e.g., see green edges in Figure~\ref{fig:approx1C}). Since $I$ is an independent set, adding $\{u,v\}$ does not change the degree of other vertices in $I$. 

\smallskip\noindent{\bf Lines \ref{alg1:L10}--\ref{alg1:L24}: Balanced partition of the edge set to the $k$ robots.}
First, we partition the edges incident to vertices in the independent set (the multiset $\widehat{E}_\mathsf{IND}$) to $k$ multisets ($\widehat{E}_i$ for every $1\leq i\leq k$), representing the $k$ robot cycles under construction. As stated in Lemma~\ref{obs:equivsol}, every such $\widehat{E}_i$ should represent an edge multiset of a multigraph that has an Eulerian cycle, thus representing the edge multiset of a robot cycle-graph. Therefore, we repeatedly take two edges incident to the same vertex in $I$, to preserve its even degree in the respective $\widehat{E}_i$ (e.g., see Figures~\ref{fig:approx1D}--\ref{fig:approx1G}). Next, we partition the edges with both endpoints in $\vc'$. As we seek a partition as ``balanced'' as possible, we repeatedly add an edge to $\widehat{E}_i$ with minimum number of edges. Observe that after this stage, every edge of $G$ is covered. Intuitively, we added edges that we ``must'' to add, and the partition is as balanced as it can be (e.g., see Figures~\ref{fig:approx2A}--\ref{fig:approx2D}). So, we get that $\mathsf{max}\{|\widehat{E}_1|,|\widehat{E}_2|,\ldots,$ $|\widehat{E}_k|\}\leq B$, where $B$ is the optimal budget (this is proved formally in Section~\ref{subsection:corrApp}). Still, we have three issues: (i) A multiset $\widehat{E}_i$ might not induce a connected graph. (ii) Vertices from $\vc'$ might have odd degree in the graphs induced by some of the $\widehat{E}_i$'s. (iii) Graphs induced by some of the $\widehat{E}_i$'s might not contain $\vi$.

\smallskip\noindent{\bf Lines \ref{alg1:L25}--\ref{alg1:L28}: Making the robot cycle-graphs connected and adding $\vi$ to it by adding a spanning tree of $G[\vc']$.}
We add to each $\widehat{E}_i$ the edges of a spanning tree of $G[\vc']$. After this stage, we get that $\gr(\widehat{E}_i)$ connected and $\vi\in V(\gr(\widehat{E}_i))$, for every $1\leq i\leq k$ (e.g., see Figures~\ref{fig:approx2F}--\ref{fig:approx2H}). 

\smallskip\noindent{\bf Lines \ref{alg1:L29}--\ref{alg1:L31}: Making the degree of each vertex in $\vc'$ even.}
Here, we use Algorithm \ref{alg:MakeVCInAEvenDeg} to get even degree for every vertex in each robot cycle-graph (e.g., see Figures~\ref{fig:approx2I}--\ref{fig:approx2K}).

\smallskip\noindent{\bf Lines \ref{alg1:L32}--\ref{alg1:L33}: Finding Eulerian cycle in each robot cycle-graph.}
Observe that at the beginning of this stage, for every $1\leq i\leq k$, $\gr(\widehat{E}_i)$ is connected, and each $u\in V(\gr(\widehat{E}_i))$ has even degree. Therefore, we can find an Eulerian cycle $\mathsf{RC}_i$ in $\gr(\widehat{E}_i)$.

\subsection{Algorithm~\ref{alg:MakeVCInAEvenDeg}: Making the Degree of Each Vertex in $\vc'$ Even}

We now describe and prove the correctness of Algorithm~\ref{alg:MakeVCInAEvenDeg}.
This algorithm gets as input a spanning tree of $G[\vc']$, $T$, one of the multisets $\widehat{E}_i$, denoted by $\widehat{E}$ and the vertex cover $\vc'$ from Algorithm~\ref{alg:AppAlg}. The algorithm, in every iteration, while $|V(T)|\geq 2$, chooses a leaf $u$ in $T$. If $u$ has even degree in $\gr(\widehat{E})$, we delete $u$ from $T$; otherwise, $u$ has odd degree in $\gr(\widehat{E})$, so we add $\{u,v\}$ to $\widehat{E}$, where $v$ is the neighbor of $u$ in $T$, and then we delete $u$ from $T$. Observe that the degree of $u$ in $\gr(\widehat{E})$ does not change after we delete it from $T$. Next, we prove that after the computation of Algorithm~\ref{alg:MakeVCInAEvenDeg}, the degree of each vertex in $\gr(\widehat{E})$ is even:

\begin{lemma}\label{lem:MakeEv}
	Let $1\leq i\leq k$, let $\widehat{E}_i$ be the multiset obtained from Algorithm \ref{alg:AppAlg} in Line~\ref{alg1:L28}, let $\vc'$ be the vertex cover computed by Algorithm \ref{alg:AppAlg}, and let $T$ be a spanning tree of $G[\vc']$. Then, $\mathsf{MakeVCEvenDeg}(T,\widehat{E}_i,\vc')$ runs in time $\OO(|\vc'|)$, adds at most $|\vc'|$ edges to $\widehat{E}_i$, and by its end every vertex in $\gr(\widehat{E}_i)$ has even degree. 
\end{lemma}

\begin{proof}
	First, observe that every $u\in V\setminus \vc'$ has even degree in $\gr(\widehat{E})$, and since we do not add any edges incident to $u$, its degree is even through all the stages of the algorithm. Now, we show by induction on the number of iterations of the while loop in Line~\ref{alg2:L3} that, at the beginning of any iteration, the following conditions hold:
	\begin{enumerate}
		\item $T'$ is a tree.
		\item Every $u\in V(G)\setminus V(T')$ has even degree in $\gr(\widehat{E})$.
	\end{enumerate}
	For the base case, notice that at the beginning of the first iteration, $T'$ is a spanning tree of $G[\vc']$, and $V(T')=\vc'$. So, from the observation stated in the beginning of this proof, we get that both conditions hold.
	
	Now, let $i>1$ be the $i$-th iteration of the while loop. From the inductive hypothesis, the conditions hold in the beginning of the $(i-1)$-th iteration. We show that they hold at the end of the $(i-1)$-th iteration, and so they hold in the beginning of the $i$-th iteration. From the condition of the while loop, we get that $T'$ has at least two vertices, so there is a leaf in $T'$. We have the following two cases:
	
	\smallskip\noindent{\bf Case 1.}
	If $u$ has even degree in $\gr(\widehat{E})$, then the algorithm deletes $u$ from $T'$. Now, $u$ is added to $V(G)\setminus V(T')$ and its degree is even in $\gr(\widehat{E})$. The degrees of the rest of the vertices from $V(G)\setminus V(T')$ in $\gr(\widehat{E})$ stay unchanged. Therefore, we get that the second condition holds.
	
	\smallskip\noindent{\bf Case 2.} In this case, $u$ has odd degree in $\gr(\widehat{E})$. Now, the algorithm adds $\{{u,v}\}$ to $\widehat{E}$, where $v$ is the neighbor of $u$ in $T'$, so now $u$ has even degree in $\gr(\widehat{E})$. The degrees of the rest of the vertices from $V(G)\setminus V(T')$ in $\gr(\widehat{E})$ stay unchanged. Thus, we get that the second condition holds.
	
	In addition, observe that, $u$ is a leaf. So, in both cases, $T'$ remains connected after deleting $u$ from it, so $T'$ remains a tree. Therefore, in both cases, both the conditions hold. This ends the proof of the inductive hypothesis. 
	
	Now, the while loop ends when $T'$ has less than two vertices. In this case, $T'$ is an isolated vertex $u$. By the inductive hypothesis we proved, every $v\in V(G)\setminus V(T')=V(G)\setminus \{u\}$ has even degree in $\gr(\widehat{E})$. Since the number of vertices with odd degree is even in every graph, we get that the degree of $u$ in $\gr(\widehat{E})$ is even. Therefore, by the end of the algorithm, every $u\in V$ has even degree in $\gr(\widehat{E})$. 
	
	Now, in every iteration of the while loop, we delete a vertex from $T'$, so we have at most $|\vc'|$ iterations. The rest of the calculations are done in $\OO(1)$ runtime, so the runtime of the algorithm is $\OO(|\vc'|)$. At every iteration of the while loop we add at most one edge to $\widehat{E}$, so the algorithm adds at most $|\vc'|$ edges to $\widehat{E}$. This ends the proof.
\end{proof}

\begin{algorithm}[!t]
	\SetKwInOut{Input}{Input}
	\SetKwInOut{Output}{Output}
	\medskip
	{\textbf{function} $\mathsf{MakeVCEvenDeg}$}$(T=(E_T,\vc'),\widehat{E},\mathsf{VC'})$\;
	$T'\gets T$\;

	\While{\label{alg2:L3}$|V(T')|\geq 2$}
	{
		Choose a leaf $u\in V(T)$\;
		\If{$u$ has odd degree in $\gr(\widehat{E})$}
		{Let $v$ be the neighbour of $u$ in $T'$\;
			$\widehat{E}\gets \widehat{E}\cup \{{u,v}\}$\;}
			$T'\gets T'\setminus \{u\}$\;
	}

	\caption{$\mathsf{MakeVCEvenDeg}$}
	\label{alg:MakeVCInAEvenDeg}
\end{algorithm}


\subsection{Correctness and Running Time of Algorithm~\ref{alg:AppAlg}}\label{subsection:corrApp}
We aim to prove the correctness of the following lemma:

\begin{lemma}\label{lem:algApp}
Let $G$ be a connected graph, let $\vi\in V(G)$, let $k\in \mathbb{N}$ and let $\mathsf{VC}$ be a vertex cover of $G$. Then, Algorithm \ref{alg:AppAlg} with the input $(G,k,\vi,\vc)$ runs in time $\OO((|V(G)|+|E(G)|)\cdot k)$ and returns a solution $\{\mathsf{RC}_i\}_{i=1}^k$ for \cg with $k$ agents such that $\mathsf{Val}(\{\mathsf{RC}_i\}_{i=1}^k)\leq B+4|\vc|$, where $B$ is the minimum budget for the instance $(G,k,\vi)$. 
\end{lemma}

We split Lemma~\ref{lem:algApp}, into two lemmas: Lemma~\ref{lem:algApp1} where we prove the correctness of Algorithm \ref{alg:AppAlg}, and Lemma~\ref{lem:algApp2} where we analyze its runtime.

\begin{lemma}\label{lem:algApp1}
	Let $G$ be a connected graph, let $\vi\in V(G)$, let $k\in \mathbb{N}$ and let $\mathsf{VC}$ be a vertex cover of $G$. Then, Algorithm \ref{alg:AppAlg} with the input $(G,k,\vi,\vc)$ returns a solution $\{\mathsf{RC}_i\}_{i=1}^k$ for \cg with $k$ agents such that $\mathsf{Val}(\{\mathsf{RC}_i\}_{i=1}^k)\leq B+4|\vc|$, where $B$ is the minimum budget for the instance $(G,k,\vi)$. 
\end{lemma}

Towards the proof of Lemma~\ref{lem:algApp1}, first, we prove that Algorithm~\ref{alg:AppAlg} returns a solution for \cg with $k$ agents:

\begin{lemma}\label{lem:algApp12}
	Let $G$ be a connected graph, let $\vi\in V(G)$, let $k\in \mathbb{N}$ and let $\mathsf{VC}$ be a vertex cover of $G$. Then, Algorithm \ref{alg:AppAlg} with the input $(G,k,\vi,\vc)$ returns a solution $\{\mathsf{RC}_i\}_{i=1}^k$ for \cg with $k$ agents. 
\end{lemma}

\begin{proof}
We prove that $\{\mathsf{RC}_i\}_{i=1}^k$ is a solution for \cg with $k$ agents for the instance $(G,k,\vi)$. We begin by showing that for every $1\leq i\leq k$, $\mathsf{RC}_i$ is a robot cycle. From Lemma \ref{lem:MakeEv}, we get that after Line~\ref{alg1:L31} in Algorithm~\ref{alg:AppAlg}, each $u\in V(\mathsf{Graph}(\widehat{E}_i))$ has even degree. In addition, $\widehat{E}_i$ contains the edge set of a spanning tree of $G[\vc']$, therefore $\mathsf{Graph}(\widehat{E}_i)$ is connected. Thus, there exists an Eulerian cycle in $\mathsf{Graph}(\widehat{E}_i)$, and so $\mathsf{RC}_i$, constructed in Line~\ref{alg1:L335}, is well defined. Moreover, observe that $\vi\in V(\gr(\widehat{E}_i))$. So, by Observation~\ref{obs:EuiIsRob}, $\mathsf{RC}_i$ is a robot cycle in $G$. Now, since every edge belongs to at least one $\widehat{E}_i$, it holds that $E(G)\subseteq \widehat{E}_1\cup \ldots \cup\widehat{E}_k$, thus $E(\mathsf{RC}_1)\cup E(\mathsf{RC}_2)\cup,\dots,\cup E(\mathsf{RC}_k)=E(G)$. So ,$\{\mathsf{RC}_i\}_{i=1}^k$ is a solution for \cg with $k$ agents for the instance $(G,k,\vi)$.
\end{proof}

Now, we prove the correctness of Lemma~\ref{lem:algApp1}.

\begin{proof}
In Lemma~\ref{lem:algApp12} we proved that Algorithm \ref{alg:AppAlg} with the input $(G,k,\vi,\vc)$ returns a solution $\{\mathsf{RC}_i\}_{i=1}^k$ for \cg with $k$ agents. 

Now, let $\{\mathsf{RC}'_i\}_{i=1}^k$ be a solution for the instance $(G,k,\vi)$ with minimum $\mathsf{Val}(\{\mathsf{RC}'_i\}_{i=1}^k)$. Let $\indd=V(G)\setminus \vc'$. By Observation~\ref{obs:robIsEu}, for every $1\leq i\leq k$, $\mathsf{RC}'_i$ is an Eulerian cycle in the robot cycle-graph of $\mathsf{RC}'_i$. Therefore, each vertex in the robot cycle-graph of $\mathsf{RC}'_i$ has even degree. Notice that $\widehat{E}_\mathsf{IND}$, defined in Lines~\ref{alg1:L5}--\ref{alg1:L9} in Algorithm~\ref{alg:AppAlg}, is a multiset of minimum size such that (i) $\{\{u,v\}\in E(G)~|~u\in \indd \}\subseteq \widehat{E}_\mathsf{IND}$ and (ii) each $u\in \indd$ has even degree in $\gr(\widehat{E}_\mathsf{IND})$. So, the total number of edges (with repetition) with an endpoint in $\indd$ in $\mathsf{RC}'_1$,\ldots,$\mathsf{RC}'_k$ is greater or equal to $|\widehat{E}_\mathsf{IND}|$. Furthermore, in Lines~\ref{alg1:L211}--\ref{alg1:L24} in Algorithm~\ref{alg:AppAlg}, every edge with both endpoints in $\vc'$ is added to exactly one $\widehat{E}_i$. In addition, observe that we allocate, in each iteration of the loops in Lines~\ref{alg1:L144} and~\ref{alg1:L211}, the edges to a multiset with minimum elements.
Therefore, we get that for the multisets $\widehat{E}_i$, for every $1\leq i\leq k$, defined up until Line~\ref{alg1:L24}, it follows that $\mathsf{max}\{|\widehat{E}_1|,|\widehat{E}_2|,\ldots, |\widehat{E}_k|\}\leq\mathsf{max}\{|E(\mathsf{RC}'_1)|,|E(\mathsf{RC}'_2)|,\ldots, |E(\mathsf{RC}'_k)|\}$. 
Now, in Line~\ref{alg1:L27}, we add the edges of a spanning tree of $G[\vc']$ to $\widehat{E}_i$, for every $1\leq i\leq k$, so we add $|\vc'|$ edges. In Line~\ref{alg1:L30}, by Lemma~\ref{lem:MakeEv}, we add at most $|\vc'|$ additional edges to $\widehat{E}_i$, for every $1\leq i\leq k$. Recall that $\vc'$ is obtained from $\vc$ by adding at most $|\vc|-1$ vertices to make $G[\vc']$ connected, and by adding $\vi$. So, $|\vc'|\leq 2|\vc|$. Overall, we get that $\mathsf{Val}(\{\mathsf{RC}_i\}_{i=1}^k)\leq \mathsf{max}\{|E(\mathsf{RC}'_1)|,|E(\mathsf{RC}'_2)|,\ldots, |E(\mathsf{RC}'_k)|\}+2|\vc'|=B+2|\vc'|\leq B+4|\vc|$. This ends the proof.     
\end{proof}

\begin{lemma}\label{lem:algApp2}
	Let $G$ be a connected graph, let $\vi\in V(G)$, let $k\in \mathbb{N}$ and let $\mathsf{VC}$ be a vertex cover of $G$. Then, Algorithm \ref{alg:AppAlg} with the input $(G,k,\vi,\vc)$ runs in time $\OO((|V(G)|+|E(G)|)\cdot k)$.
\end{lemma}

\begin{proof}
We analyze the runtime of Algorithm~\ref{alg:AppAlg} as follows. Making $\vc'$ connected in Line \ref{alg1:L3}, takes $\OO(|V(G)|+|E(G)|)$ time. In Line~\ref{alg1:L16}, to find $1\leq i\leq k$ such that $|\widehat{E}_i|$ is minimum, we do as follows. At the $j$-th iteration $i=j (mod~k)+1$. Thus, there are at most $\OO|E(G)|$ iterations, and every iteration takes $\OO(1)$ time. We mark the last index of multiset to get pairs of edges in Line~\ref{alg1:L16} by $t$. Observe that by the end of Line~\ref{alg1:L20}, $|\widehat{E}_1|=\cdots=|\widehat{E}_t|=|\widehat{E}_{t+1}|+2=\cdots =|\widehat{E}_{k}|+2$. Then, in Line~\ref{alg1:L21} to find $1\leq i\leq k$ such that $|\widehat{E}_i|$ is minimum, we first take $i=t +1(mod~k)$ until $k$. Then, we add the next edge to sets in the reverse order, that is, from $k$ to $1$ and so on. Observe that the chosen $i$ is indeed such that $|\widehat{E}_i|$ is minimum. Therefore, there are at most $|E(G)|$ iterations in Line~\ref{alg1:L20}, each takes $\OO(1)$ time. Finding a spanning tree of $G[\vc']$ in Line~\ref{alg1:L25} takes $\OO(|V(G)|+|E(G)|)$ time~\cite{cormen2022introduction}. By Lemma~\ref{lem:MakeEv}, each iteration in Line~\ref{alg1:L30} takes $\OO(|\vc|)$ time, so in total, all the $k$ iterations take $\OO(k\cdot |\vc|)$ time. Finding an Eulerian cycle in Line~\ref{alg1:L335} in each $\gr(\widehat{E}_i)$ takes $\OO(|E(G)|)$ time~\cite{fleischner1991x}, so in total, the $k$ iterations take $\OO((|V(G)|+|E(G)|)\cdot k)$ time. Therefore, Algorithm~\ref{alg:AppAlg} runs in time $\OO((|V(G)|+|E(G)|)\cdot k)$. This completes the proof. 
\end{proof}

The proofs of Lemmas~\ref{lem:algApp1} and \ref{lem:algApp2} conclude the correctness of Lemma~\ref{lem:algApp}.

Note that there exists a $2$-approximation for computing a minimal vertex cover of an input graph~\cite{papadimitriou1998combinatorial}. This result together with Lemma~\ref{lem:algApp} conclude the correctness of Theorem~\ref{the:approx}.


\section{\WO-Hardness for \cg}\label{sec:Hardk}
In this section, we aim to prove the following theorem: 

\hardTheorem*
	
We prove Theorem~\ref{th:w1hard} by showing a reduction from {\sc Exact Bin Packing} (see Definition~\ref{def:ExBinPack}).

First, we show that unary {\sc Exact Bin Packing} is \WOH with respect to $k$. It is known that unary {\sc Bin Packing} is \WOH with respect to $k$ \cite{DBLP:journals/jcss/JansenKMS13}. So, we give a reduction from {\sc Bin Packing} to {\sc Exact Bin Packing} in order to prove the following lemma:

\begin{lemma}\label{lem:w1h}
	Unary {\sc Exact Bin Packing} is \WOH with respect to $k$.
	\end{lemma}

\begin{proof}
Let $(I,s,B,k)$ be an instance of {\sc Bin Packing} problem. Let $t=B\cdot k-\sum_{i\in I}s(i)$ and let $s':I\cup\{i_1,\ldots,i_t\}\rightarrow \mathbb{N}$ be a function defined as follows. For every $i\in I$, $s'(i)=s(i)$, and for every $i_\ell\in \{i_1,\ldots,i_t\}$, $s'(i_\ell)=1$. Observe that $(I\cup\{i_1,\ldots,i_t\},s',B,k)$ is an instance of {\sc Exact Bin Packing}. We show that $(I,s,B,k)$ is a yes-instance of {\sc Bin Packing} if and only if $(I\cup\{i_1,\ldots,i_t\},s',B,k)$ is a yes-instance of {\sc Exact Bin Packing}. 

Assume that $(I,s,B,k)$ is a yes-instance of {\sc Bin Packing}. Let $I_1,\ldots,I_k$ be a partition of $I$ into disjoint sets such that for every $1\leq j\leq k$, $\sum_{i\in I_j} s(i)\leq B$. For every $1\leq j\leq k$, let $t_j=B-\sum_{i\in I_j} s(i)$.  Let $I_1',\ldots,I_k'$ be a partition of $\{i_1,\ldots,i_t\}$ into $k$ disjoint sets such that for every $1\leq j\leq k$, $|I_j'|=t_j$. Observe that there exists such a partition since $\sum_{1\leq j\leq k}t_j=t$. Clearly, $I_1\cup I_1',\ldots,I_k\cup I_k'$ is a partition of $I\cup\{i_1,\ldots,i_t\}$ into disjoint sets such that for every $1\leq j\leq k$, $\sum_{i\in I_j\cup  I'_j} s(i)=B$. Therefore, $(I\cup\{i_1,\ldots,i_t\},s',B,k)$ is a yes-instance of {\sc Exact Bin Packing}. 

Now, assume that $(I\cup\{i_1,\ldots,i_t\},s',B,k)$ is a yes-instance of {\sc Exact Bin Packing}.  Let $I_1,\ldots,I_k$ be a partition of $I\cup\{i_1,\ldots,i_t\}$ into disjoint sets such that for every $1\leq j\leq k$, $\sum_{i\in I_j} s(i)=B$. Observe that $I_1\setminus \{i_1,\ldots,i_t\} ,\ldots,I_k\setminus \{i_1,\ldots,i_t\}$ is a partition of $I$ into disjoint sets such that for every $1\leq j\leq k$, $\sum_{i\in I_j} s(i)\leq B$. Therefore, $(I,s,B,k)$ is a yes-instance of {\sc Bin Packing}.

Clearly, the reduction works in polynomial time when the input is in unary. Thus, since unary {\sc Bin Packing} is \WOH with respect to $k$ \cite{DBLP:journals/jcss/JansenKMS13}, unary {\sc Exact Bin Packing} is \WOH with respect to $k$.
\end{proof}

\subsection{Reduction From {\sc Exact Bin Packing} to \cg}
Given an instance $(I,s,B,k)$ of {\sc Exact Bin Packing} problem, denote by $\mathsf{BinToRob}(I,s,B,k)$ the instance of \cg defined as follows. First, we construct the graph $T$ as follows. For each $i\in I$ we create a star with $s(i)-1$ leaves. We connect each such star with an edge to a vertex $r$. Formally, $V(T)=\{v^{i},v^{i}_{1}\ldots,v^i_{s(i)-1}~|~i\in I\}\cup \{r\}$ and $E(T)=\{\{v^i,v^i_j\}~|~i\in I, 1\leq j\leq s(i)-1\}\cup \{\{r,v_i\}~|~i\in I\}$. Now, we define $\mathsf{BinToRob}(I,s,B,k)=(T,r,k,2B)$. See Figure~\ref{fig:hardness} for an example. Next, we prove the correctness of the reduction:

\begin{figure}
	\centering
	\begin{subfigure}[t]{.3\textwidth}
		\centering
		\includegraphics[page=1]{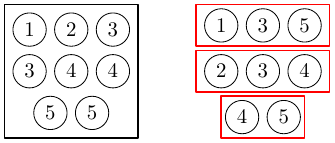}
		\caption{}
		\label{}
	\end{subfigure}\hfill
	\begin{subfigure}[t]{.6\textwidth}
		\centering
		\includegraphics[page=2]{figures/robotExplore}
		\caption{}
		\label{}
	\end{subfigure}
	\caption{An illustration of a {\sc Exact Bin Packing} instance, a solution (in sub-figure (a)) and the equivalent instance of \cg constructed by the $\mathsf{BinToRob}$ function (in sub-figure (b)).}
	\label{fig:hardness}
\end{figure}

\begin{lemma}\label{lem:hardness}
Let $(I,s,B,k)$ be an instance of {\sc Exact Bin Packing}. Then, $(I,s,B,k)$ is a yes-instance if and only if $\mathsf{BinToRob}(I,s,B,k)$ is a yes-instance of \cg.
\end{lemma} 

\begin{proof}
First, assume that $(I,s,B,k)$ is a yes-instance. Let $I_1,\ldots,I_k$ be a partition of $I$ into disjoint sets such that for every $1\leq j\leq k$, $\sum_{i\in I_j} s(i)=B$. We prove that $\mathsf{BinToRob}(I,s,B,k)=(T,r,k,2B)$ is a yes-instance of \cg, by showing that there exist $k$ multisets $\widehat{E}_1,\ldots,\widehat{E}_k$ such that the conditions of Lemma~\ref{obs:equivsol} are satisfied. For every $1\leq j\leq k$, let $\widehat{E}_j=\{\{v^i,v^i_t\},\{v^i,v^i_t\}~|~$ $i\in I_j, 1\leq t\leq s(i)-1\}\cup \{\{v^i,r\},\{v^i,r\}\}$. Clearly, $r\in V(\gr(\widehat{E}_j))$, $\gr(\widehat{E}_j)$ is connected, and every vertex in $\gr(\widehat{E}_j)$ has even degree. Therefore, Conditions~\ref{obs:equivsol1} and~\ref{obs:equivsol2} are satisfied. In addition, since $I=I_1\cup\ldots \cup I_k$, we have that $E\subseteq \widehat{E}_1\cup \ldots \cup\widehat{E}_k$, so Condition~\ref{obs:equivsol3} is satisfied. Now, for every $1\leq j\leq k$, $|\widehat{E}_j|=|\{\{v^i,v^i_t\},\{v^i,v^i_t\}~|~$ $i\in I_j, 1\leq t\leq s(i)-1\}\cup \{\{v^i,r\},\{v^i,r\}\}|=\sum_{i\in I_j} 2(s(i)-1)+\sum_{i\in I_k}2=2\sum_{i\in I_k}s(i)=2B$. Thus, Condition~\ref{obs:equivsol4} is satisfied. Therefore, all the conditions of Lemma~\ref{obs:equivsol} are satisfied, so $\mathsf{BinToRob}(I,s,B,k)$ is a yes-instance of \cg.

Now, we prove the reverse direction. Assume that $\mathsf{BinToRob}(I,s,B,k)=(T,r,k,2B)$ is a yes-instance of \cg. From Lemma~\ref{obs:equivsol}, there exist $k$ multisets $\widehat{E}_1,\ldots,\widehat{E}_k$ such that the conditions of Lemma~\ref{obs:equivsol} hold. Let $1\leq j\leq k$. We first show that every $\{u,v\}\in \widehat{E}_j$ appears at least twice in $\widehat{E}_j$. Let $\{u,v\}\in \widehat{E}_j$. We the following two cases:

\smallskip\noindent{\bf Case 1: $\{u,v\}=\{v^i,v^i_t\}$ for some $i\in I$ and $1\leq t\leq s(i)-1$.} From Condition~\ref{obs:equivsol2} of Lemma~\ref{obs:equivsol}, $v^i_t$ has even degree in $\gr(\widehat{E}_j)$. Since $\{v^i,v^i_t\}$ is the only edge having $v^i_t$ as an endpoint in $T$, $\{v^i,v^i_t\}$ appears an even number of times in $\widehat{E}_j$, and so it appears at least twice in $\widehat{E}_j$.

\smallskip\noindent{\bf Case 2: $\{u,v\}=\{v^i,r\}$ for some $i\in I$.} From Condition~\ref{obs:equivsol2} of Lemma~\ref{obs:equivsol}, $v^{i}$ has even degree in $\gr(\widehat{E}_j)$. From Case 1, each $\{v^i,v^i_t\}\in \widehat{E}_j$ appears an even number of times in $\widehat{E}_j$. Therefore, since $r$ is the only neighbor of $v^i$ other than $v^i_t$, $1\leq t\leq s(i)-1$, $\{v^i,r\}$ appears an even number of times, which is greater or equal to $2$, in $\widehat{E}_j$.

Now, observe that $|E(T)|=\sum_{i\in I}s(i)=B\cdot k$, and from Condition~\ref{obs:equivsol4} of Lemma~\ref{obs:equivsol}, $\sum_{1\leq j\leq k}|\widehat{E}_j|\leq 2B\cdot k$. In addition, from Condition~\ref{obs:equivsol3} of Lemma~\ref{obs:equivsol}, (1): $E(T)\subseteq \widehat{E}_1\cup \ldots \cup\widehat{E}_k$. So, since we have already proved that for every $1\leq j\leq k$, each $\{u,v\}\in \widehat{E}_j$ appears at least twice in $\widehat{E}_j$, we get that for every $1\leq j<j'\leq k$, $\widehat{E}_j\cap \widehat{E}_{j'}=\emptyset$, and $\sum_{1\leq\ell \leq k}|\widehat{E}_\ell|= 2B\cdot k$; in turn, for every $1\leq j\leq k$, $|\widehat{E}_j|= 2B$, and each $\{u,v\}\in \widehat{E}_j$ appears exactly twice in $\widehat{E}_j$. Moreover, from Conditions~\ref{obs:equivsol1} and~\ref{obs:equivsol2} of Lemma~\ref{obs:equivsol}, for every $1\leq j\leq k$, $r\in V(\gr(\widehat{E}_j))$ and $\gr(\widehat{E}_j)$ is connected. Therefore, for every $1\leq j\leq k$ and $i\in I$, if $v^i\in V(\gr(\widehat{E}_j))$ then $\{\{v^i,v^i_t\}~|~ 1\leq t\leq s(i)-1\}\cup \{r,v^i\}\subseteq \widehat{E}_j$. Thus, for every $1\leq j<j'\leq k$, (2): $V(\gr(\widehat{E}_j))\cap V(\gr(\widehat{E}_{j'}))=\{r\}$. 

Now, for every $1\leq j\leq k$, let $I_j=\{i\in I~|~v^i\in V(\gr(\widehat{E}_{j}))\}$. By (1) and (2), $I_1,\ldots,I_k$ is a partition of $I$ into disjoint sets. We show, that for every $1\leq j\leq k$, $\sum_{i\in I_j} s(i)=B$. Let $1\leq j\leq k$. Then, $\sum_{i\in I_j} s(i)= \sum_{i\in I_j}|\{\{v_i,v_{i_t}\}~|~ 1\leq t\leq s(i)-1\}\cup \{r,v_i\}|=\frac{1}{2}|\widehat{E}'_{j'}|=\frac{1}{2}\cdot2\cdot B=B$. Therefore $I_1,\ldots,I_k$ is a solution for $(I,s,B,k)$, so $(I,s,B,k)$ is a yes-instance of {\sc Exact Bin Packing} problem. This ends the proof.
\end{proof}

Clearly, the reduction works in polynomial time when the input is in unary. In addition, observe that the treedepth of the tree, obtained by the reduction, is bounded by $3$. Now, recall that, by Lemma~\ref{lem:w1h},
unary {\sc Exact Bin Packing} is \WOH with respect to $k$. Thus, we conclude from Lemma~\ref{lem:hardness} the correctness of Theorem~\ref{th:w1hard}.

\newpage

\bibliography{Refs,RefsIntro}

\end{document}